%% file: main.tex
\documentclass[11pt]{article}
\usepackage[utf8]{inputenc}
\usepackage[T1]{fontenc}
\usepackage[english]{babel}
\usepackage{mwe}
\usepackage{amsmath,mathtools}
\usepackage[dvipsnames]{xcolor}

\PassOptionsToPackage{authoryear,round}{natbib}
\usepackage{natbib}
\usepackage{etoolbox}
\AtEndDocument{\let\clearpage\relax}

\input{header.tex}

\usepackage{amsthm,bm}
\usepackage{lineno}

\usepackage{tikz,pgfplots}
\pgfplotsset{compat=1.5}
\usetikzlibrary{intersections}
\usetikzlibrary{patterns}
\usepackage{multirow}
\usepackage{graphicx}

\usepackage{subcaption}
\usepackage{pgfpages}
\usepackage{pgfplotstable}

\usepackage[shortlabels]{enumitem}
\usepackage[colorlinks=true, linkcolor=blue, citecolor=blue, hypertexnames=false]{hyperref}  
\usepackage{nameref}
\usepackage{cleveref}

\usepackage{color}              %

\usepackage[showdeletions]{color-edits}

\usepackage{setspace}
\usepackage[shortlabels]{enumitem}
\setlist[itemize]{itemsep=2pt, topsep=2pt}
\setlist[enumerate]{itemsep=2pt, topsep=2pt}

\expandafter\def\expandafter\normalsize\expandafter{%
    \normalsize%
    \setlength\abovedisplayskip{6pt}%
    \setlength\belowdisplayskip{6pt}%
    \setlength\abovedisplayshortskip{6pt}%
    \setlength\belowdisplayshortskip{6pt}%
}

\makeatletter
\def\thm@space@setup{%
  \thm@preskip=5pt
  \thm@postskip=5pt
}

\addauthor[Yanwei]{ys}{red}
\addauthor[Chiwei]{cy}{cyan}
\addauthor[Fupeng]{fs}{yellow}
\definecolor{jworange}{RGB}{239,138,98}
\addauthor[Jiahua]{jw}{jworange}

\usepackage[right=1.0in, top=1in, bottom=1.0in, left=1.0in]{geometry}

\usepackage{lmodern}

\usepackage{multirow}

\usepackage{lscape}
\usepackage[final]{pdfpages}
\usepackage{amsthm}
\usepackage[authoryear,round]{natbib}
\usepackage{amsfonts}
\usepackage{mathtools}
\usepackage{bbm}
\usepackage{dsfont}

\usepackage{minitoc}

\usepackage{bbm}
\newtheorem{lemma}{Lemma}
\newtheorem{proposition}{Proposition}
\newtheorem{corollary}{Corollary}
\newtheorem{theorem}{Theorem}
\newtheorem*{theorem*}{Theorem}
\newtheorem{assumption}{Assumption}

\newtheorem{definition}{Definition}

\newtheorem{remark}{Remark}
\newtheorem{example}{Example}
\newtheorem{observation}{Observation}

\usepackage{authblk}

\date{}

\title{The Role of Prescreening in Auctions with Predictions\thanks{The first two authors contributed equally.}}

\begin{document}

\author[1]{Yanwei Sun}
\author[2]{Fupeng Sun}
\author[3]{Chiwei Yan}
\author[4]{Jiahua Wu}
\affil[1]{Imperial College Business School, \texttt{yanwei@imperial.ac.uk}}
\affil[2]{Imperial College Business School, \texttt{f.sun23@imperial.ac.uk}}
\affil[3]{University of California, Berkeley, \texttt{chiwei@berkeley.edu}}
\affil[4]{Imperial College Business School,  \texttt{j.wu@imperial.ac.uk}}

\maketitle

\vspace{-1.5em}
\begin{abstract}
Sellers often prescreen potential bidders, restricting participation to a select group of capable participants. Recent advances in machine learning and generative AI make this strategy increasingly viable by enabling the cost-effective identification of high-quality bidders. However, the practice departs from classic auction theory, which usually favors broad competition over selective exclusion. In this paper, we examine whether and under what conditions bidder prescreening can be justified. We analyze a setting in which bidders have independent and identically distributed private valuations, and the seller observes noisy signals generated by a valuation predictor. The seller determines how many top bidders to admit and, after receiving signals, selects exactly that many with the highest signal-based rankings.
We demonstrate that an auction with prescreening is equivalent to a standard auction (i.e., without prescreening) but with \textit{correlated} valuations. Our analysis shows that, although admitting fewer bidders leads to revenue losses in both second-price and first-price auctions, a more accurate predictor can mitigate or even fully offset these losses. In contrast, prescreening can significantly boost revenue in all-pay auctions; notably, when the predictor is perfect, admitting only two bidders is optimal. All results remain valid in the presence of reserve prices. 
\end{abstract}

\begin{quote}
    \emph{``The success of transactions depends even more on what happens \textbf{before} and after the auction. Understanding the transaction as a whole requires one to ask \textbf{who participates}...''} --- \emph{Putting Auction Theory to Work} \citep{milgrom_2004_putting_auctiontheory_to_work}
\end{quote}

\doparttoc %
\faketableofcontents %

\input{sections/1-introduction}

\input{sections/2-model}

\input{sections/3-Beliefs}

\input{sections/5-First-price-Auctions}

\input{sections/4-allpay}

\input{sections/6-rev_ranking}

\vspace{-1em}

\input{sections/extension}

\input{sections/7-conclusion}

\vspace{-0.6em}
\begingroup            %
\footnotesize
\bibliographystyle{plainnat}
\bibliography{ref}
\endgroup   

\newpage

\appendix

\renewcommand{\theequation}{\thesection-\arabic{equation}}
\renewcommand{\theproposition}{\thesection-\arabic{proposition}}
\renewcommand{\thelemma}{\thesection-\arabic{lemma}}
\renewcommand{\thetheorem}{\thesection-\arabic{theorem}}
\renewcommand{\thedefinition}{\thesection-\arabic{definition}}
\pagenumbering{arabic}
\renewcommand{\thepage}{App-\arabic{page}}

\setcounter{equation}{0}
\setcounter{proposition}{0}
\setcounter{definition}{0}
\setcounter{lemma}{0}
\setcounter{theorem}{0}

\input{sections/appendix}
\end{document}

%% file: header.tex
\allowdisplaybreaks

\usepackage[normalem]{ulem}
\usepackage{amssymb}

\usepackage{soul}

\DeclareMathOperator*{\argmax}{\text{argmax}}

\newcommand{\HB}{\mathsf{HB}}
\newcommand{\FGM}{\mathsf{FGM}}

\newcommand{\R}{\mathsf{R}}
\newcommand{\lar}{\mathsf{lar}}
\newcommand{\mar}{\mathsf{mar}}
\newcommand{\fosd}{\mathsf{FOSD}}

\newcommand{\C}{\mathsf{C}}
\newcommand{\1}{\mathds{1}}

\newcommand{\I}{\mathcal{I}}
\newcommand{\barI}{\Bar{\mathcal{I}}}

\newcommand{\sigmaFP}{\sigma^{\mathsf{FP}}}
\newcommand{\sigmaAP}{\sigma^{\mathsf{AP}}}

\newcommand{\AP}{\mathsf{AP}}
\newcommand{\FP}{\mathsf{FP}}
\newcommand{\SP}{\mathsf{SP}}

\newcommand{\PRD}{\mathsf{PRD}}
\newcommand{\pqdorder}{\succeq_{\mathsf{pqd}}}
\newcommand{\NP}{\mathsf{NP}}
\newcommand{\PP}{\mathsf{PP}}
\newcommand{\HP}{\mathsf{HP}(\gamma)}

\newcommand{\RHR}{\mathsf{RHR}}

\let\oldmarginpar\marginpar
\renewcommand{\marginpar}[1]{\oldmarginpar{\raggedright #1}}

\newcommand{\MA}{\texttt{MyersonAuction}}

\newcommand{\FBM}{\texttt{MRAuction}}

\usepackage{booktabs,tabularx,makecell,amsmath,threeparttable}

\newcolumntype{Y}{>{\centering\arraybackslash}X}

\usepackage{algorithm}      %
\usepackage{algpseudocode}  %

%% file: sections/1-introduction.tex
\section{Introduction}\label{sec:intro}

Sellers often prescreen potential bidders to ensure that only a limited number of serious and capable participants are included in the bidding pool \citep{milgrom_2004_putting_auctiontheory_to_work}. 
The motivation is pragmatic—capacity limits, timing targets, and the desire to intensify competition.
For example, in offline settings such as high-profile auctions organized by Christie’s or Sotheby’s, participation is often restricted due to physical venue constraints, with only a select group of VIPs or major collectors invited to reinforce an atmosphere of exclusivity. In online contexts, such as Google Ads, the platform admits only the most relevant bidders (advertisers) to each auction, ensuring that auctions can be executed within a predetermined time frame.
The common thread across these applications is that only a fixed and predetermined number of participants can be admitted due to external resource constraints.

A similar phenomenon arises in contests, often modeled as all-pay auctions because participants expend effort that becomes a sunk cost regardless of the outcome. An example
is the RoboMaster competition, an annual intercollegiate robotics contest. To maximize audience engagement and commercial appeal, the organizers screen résumés, admit only the most promising teams, and disclose both the total number of applicants and the number ultimately accepted. These steps signal a highly competitive environment, motivating teams to exert extraordinary effort.\footnote{Organizers seek significant effort for two reasons. First, their mission is to ``lead global robotics competitions and drive the development of the robotics industry through the most rigorous competition rules''. (see \url{https://www.robomaster.com/en-US/robo/overview?djifrom=nav}) Second, various spin-off activities, described on the same webpage, depend on intense competition to attract a broad audience and ensure commercial viability.} They also reflect practical constraints such as venue capacity, limited time slots, and staffing, which make accommodating every registered team infeasible.

This practice is becoming increasingly viable with recent advances in machine learning and artificial intelligence. Generative AI now simulates human behaviour with striking realism \citep{immorlica_2024_generative,li_2025_llm_persona}, while predictive models trained on historical bidding data increasingly provide accurate estimates of bidder valuations \citep{munoz_2017_auctions_predictions,lu_2024_competitiveauctions_predictions}. These tools offer sellers a cost‑effective means of identifying promising participants, allowing them to prescreen and admit only those bidders deemed most capable.

Yet this practice, although driven by practical constraints, runs counter to standard auction theory, which generally suggests that competition trumps discrimination: admitting more bidders typically yields a more efficient allocation of the asset \citep{levin_1994_aer_auctions_entry,bulow_1994_auctions_negotiations}. In this paper, we develop a theoretical model to investigate whether and under what conditions bidder prescreening can be justified. We consider a setting in which bidders have independently and identically distributed (i.i.d.)\ private valuations, and the seller receives noisy signals from a valuation predictor. The seller decides the number of bidders to admit before signals are realized, reflecting real-world resource constraints, and, once the signals are observed, admits the top-ranked bidders. Intuitively, this prescreening process not only determines the number of participants but also acts as a form of information disclosure: bidders admitted into the auction are, on average, more likely to possess high valuations, especially when the \emph{prediction accuracy}, which we formalize via the positive quadrant dependence order between valuations and signals, is high.

We find that even when the prior distribution is i.i.d., the posterior beliefs of admitted bidders become \textit{asymmetric}, depending parametrically on their own private valuations. We characterize these posterior beliefs in closed form: the posterior density is equal to the product of the prior density and the admission probability, up to a normalizing constant. Furthermore, we show that an auction with prescreening is equivalent to a standard auction, i.e., one without prescreening, in which bidder valuations are statistically \textit{correlated} and drawn from a symmetric joint distribution. In other words, the prescreening process endogenously induces correlation among bidders' valuations. Notably, the standard notion of \textit{affiliation}, commonly assumed in the auction literature, does not generally hold in this equivalent formulation, except in special cases such as when the predictor is perfect.

We then investigate optimal prescreening in three commonly used auction formats: second-price, first-price, and all-pay auctions. In second-price auctions, truthful bidding remains a dominant strategy, and the seller’s revenue equals the expected second-highest valuation among the admitted bidders. This expected value increases with the number of admitted bidders, reflecting the noisy nature of the predictor. By admitting more participants, the seller reduces the risk of excluding genuinely high-valuation bidders due to imperfect prediction, thereby making it optimal to admit all bidders. In contrast, prescreening in first-price auctions influences both the distribution of admitted bidders’ valuations and their equilibrium bidding strategies. This dual impact complicates the analysis of the optimal number of admitted bidders. Nevertheless, we show, perhaps surprisingly, that admitting all bidders remains optimal under mild conditions on the predictor. This is because the seller’s expected revenue in first-price auctions depends on both the highest valuation among admitted bidders and their equilibrium bidding strategies. The highest valuation increases with the number of admitted bidders, following similar logic as in second-price auctions. The effect of prescreening on equilibrium strategies arises solely through its influence on the \emph{reverse hazard rate} of the winning probability, which captures how sensitively the probability of winning responds to an infinitesimal increase in the bid. If the reverse hazard rate with all bidders admitted is at least as high as that under any smaller admitted number, then admitting all bidders is optimal.

In both second-price and first-price auctions, although admitting all bidders is generally optimal, the revenue loss from admitting fewer bidders decreases (weakly) with the accuracy of the predictor. Moreover, under a perfect predictor, admitting \emph{any} number of bidders results in the same revenue. The practical implication here is that even when sellers face constraints that limit the number of participants, using an accurate valuation predictor can substantially reduce---or, in the case of perfect predictors, completely offset---revenue losses.

For all-pay auctions, we identify a sufficient condition, based on the product of a bidder's true type and the normalization term of their posterior belief, under which a symmetric and strictly monotonic (SSM) equilibrium strategy exists. We further show that this condition is also necessary when the predictor is perfect. When an SSM equilibrium exists, we prove its uniqueness and provide an explicit characterization. For all-pay auctions, prescreening affects the seller’s revenue through three channels: the number of bidders making payments, i.e., the number of participants, the valuation distribution of admitted bidders, and the bidders’ equilibrium strategies. As a result, determining the optimal number of admitted bidders involves balancing these competing effects, making the solution more nuanced.
We show that under a perfect predictor, admitting \emph{only two} bidders is optimal.
In this case, the equilibrium strategy decreases with the number of admitted bidders, and the increase in bids from the two selected participants outweighs the effect of having fewer paying bidders. This result holds for \emph{any} prior distribution and \emph{any} total number of bidders.
At the other end of the spectrum, when signals are completely uninformative, admitting all bidders maximizes expected revenue. 
For intermediate cases between these extremes, our characterizations of equilibrium strategies and other key quantities enable a computationally tractable determination of the optimal number of admitted bidders.

The optimality results above allow us to establish a revenue ranking, showing that under optimal prescreening, all-pay auctions yield higher revenue than both second-price and first-price auctions. This finding contrasts with the classic revenue equivalence result in the independent private values setting \citep{myerson_1981_optimal_auction}, as well as with prior comparisons under affiliated valuations, which suggest that the revenue ranking between all-pay and second-price auctions is generally indeterminate \citep{milgrom_1982_auctiontheory_competitive_bidding, krishna_1997_all_pay_affiliation}.
Moreover, we show that under a perfect predictor and when the total number of bidders is sufficiently large, the revenue from an all-pay auction with the optimal number of admitted bidders exceeds the optimal revenue from first-price and second-price auctions with one additional potential bidder. Combined with the classic result of \citet{bulow_1994_auctions_negotiations}, this implies that the revenue from an all-pay auction with optimal prescreening can surpass that of \emph{any auction format} without prescreening, underscoring the powerful role of prescreening in enhancing revenue in all-pay settings.

Finally, we explore several extensions of the base model. First, we incorporate reserve prices into the previously analyzed auctions and find that all our results remain valid. Second, we examine the joint design of the auction format and the number of admitted bidders.
Finally, we consider an alternative objective commonly studied in the all-pay auction literature, i.e., maximizing the expected highest bid \citep{moldovanu_2006_contest_architecture, Jason_Optimal_Crowdsourcing_Contest}, and demonstrate that restricting access can improve the expected highest bid, even when the predictor is completely uninformative.

\smallskip

\noindent
\textbf{Related Works.} Our work is closely related to the extensive literature on auction analysis, particularly studies on \emph{interdependent private valuations} and \emph{endogenous bidder participation}. While we begin with the assumption of independent private valuations, the prescreening process endogenously induces correlation among the valuations of admitted bidders. The auction literature on interdependent valuations with positive correlation often assumes affiliated valuations, as pioneered by \citet{milgrom_1982_auctiontheory_competitive_bidding}, to ensure analytical tractability. In particular, under affiliation, \citet{milgrom_1982_auctiontheory_competitive_bidding} show that an SSM equilibrium strategy always exists in first-price auctions and that second-price auctions yield higher expected revenue than first-price auctions. \citet{krishna_1997_all_pay_affiliation} extend this analysis to all-pay auctions and find that the existence of an SSM equilibrium may fail under affiliation. They provide a sufficient condition for the existence of an SSM equilibrium and show that, under this condition, all-pay auctions yield higher revenue than first-price auctions.
Under the same affiliation condition, \cite{pekevc_2008_revenueranking_random_number} show that when the number of bidders is unknown, first-price auctions can outperform second-price auctions, particularly when the uncertainty is substantial.
Our work differs from the existing literature in two key respects: (1) The correlation structure among bidders' valuations is induced by the prescreening process, which varies with the number of admitted bidders. As a result, the valuations of admitted bidders are generally not affiliated. We provide sufficient conditions to ensure the existence of SSM equilibria in first-price and all-pay auctions. (2) The literature typically focuses on revenue rankings across auction formats under a fixed correlation structure. In contrast, our analysis requires comparing auction outcomes across different correlation structures induced by prescreening.

A substantial literature has also examined auctions that incorporate endogenous bidder participation, where bidders decide whether to enter based on expected payoffs and entry cost. For example, the seminal paper by \citet{samuelson_1985_entrycosts} considers a setting in which potential bidders first observe their private values and then decide whether to enter the auction by incurring an entry cost.
In contrast, \citet{levin_1994_aer_auctions_entry} assume that bidders choose, via a mixed strategy, whether to participate in the auction before observing any information about their valuations. \citet{gentry_2014_affiliated_signal_entry} generalize these frameworks by allowing the signals observed prior to entry to be (more generally) affiliated with the bidders’ valuations. Bidders’ endogenous entry can itself act as a screening mechanism, limiting auction participation. However, the literature generally assumes that bidders have independent private valuations and decide whether to enter based solely on their own valuation (or a signal thereof), resulting in i.i.d.\ valuations among entrants. %

Similar to the setting considered in this paper, several studies focus explicitly on mechanisms implemented before or after auctions. \citet{ye_2007_GEB_indicative_twostage_auctions} study a two-stage auction with costly entry and investigate how the seller should determine the optimal number of admitted bidders with indicative bidding in the first stage. %
In the context of contests, \citet{moldovanu_2006_contest_architecture} study a two-stage contest, modeled as an all-pay auction, in which participants are randomly divided into subgroups in the first stage, and the winner of each subgroup advances to the second stage. They show that, under certain conditions, it is optimal to partition participants into two groups in the first stage. 
Using a mechanism-design approach, \citet{lu_2014_optimal_twostageauction_info_acquisition} examine which partially informed bidders a seller should admit to the auction and, conditional on admission, bidders acquire additional (costly) information about their own valuations. \citet{golrezaei_2017_auctions_information_acquisition} consider a similar setting, with the difference that all bidders---including those who do not acquire new information--- participate in the auction.
Relatedly, an inspection/verification strand shows how the seller can supplement bidding with ex‑post inspection: \citet{alaei_2024_deferred_inspection_reward} design a deferred‑inspection‑and‑reward mechanism that incentivizes truthful reporting, and %
\cite{li_2025_mechanism_inspection_exclusion} consider a setting in which the seller can inspect bidders’ valuations after the bidding and can exclude bidders found to be bidding non-truthfully from future auctions.
Although the mechanisms studied in the literature differ, the admitted bidders’ valuations remain i.i.d. (or independent). 
By contrast, the prescreening mechanism considered in this paper selects bidders based on their rank, thereby inducing correlation among the valuations of the admitted bidders. Also common in procurement auctions within manufacturing, buyers (i.e., auctionier) typically screen the qualifications of suppliers (i.e., bidders) in advance to ensure that the auction winner is eligible to be awarded a contract \citep{wan_2009_rfq_qualification_screening, wan_2012_pay_suppliers_qualification, zhang_2021_now_later_quailification, chen_2022_procurement_preaward_cost_reduction}. In these settings, screening is based on suppliers’ qualification levels, which are orthogonal to their private costs (analogous to valuations in our model), so the cost distribution remains unaffected by the screening process.

Our paper also contributes to the emerging literature on \textit{decision making with predictions}. A common theme in this literature is that the designer is employed with a predictor, often a machine learning model trained on agents’ characteristics and historical behavior, to estimate agents’ private types. This enables algorithmic mechanism design that leverages prediction errors as an alternative to worst-case analysis \citep{munoz_2017_auctions_predictions,lu_2024_competitiveauctions_predictions,balkanski_2024_mechanism_annotated_list,feng_2024_robust_staffing_prediction}.%
We extend this literature by examining, from a Bayesian perspective, how a seller should select the number of admitted bidders in auctions when equipped with valuation predictors. 
Two studies are relevant to our work. \citet{sun_2024_contests} study contests with entry restrictions using an all-pay auction model, assuming that the designer knows the exact ranking of participants’ private types. This setting is a special case of our framework, corresponding to perfect prediction of bidders’ valuations. Whereas their objective is to maximize the expected highest bid in an all‑pay auction, we focus on revenue maximization across the three common auction formats considered in this paper, including the all‑pay auction, under general predictors. The recent work by \citet{lobel_2025_auction_hallucinations} also examines a related setting in which the seller receives signals about bidders' valuations from machine learning or AI models. There are two key differences between their work and ours. First, their analysis focuses on how ML/AI-based signals influence the design of the optimal auction format, whereas our primary interest lies in prescreening strategies. Second, the two studies differ in a crucial technical assumption: in their model, bidders observe the signals, %
while in ours, the signals are observed only by the seller, who decides the number of admitted bidders before the realization of these signals. This assumption allows us to isolate and examine the informational effects of prescreening. A promising direction for future research would be to develop a unified framework that jointly addresses auction design, bidder admission, and information disclosure.

Lastly, our paper contributes to the fast-growing literature on information design under the commitment assumption \citep{kamenica_2011_persuasion,candogan_2020_info_operations,bergemann_2022_info_disclosure_auctions,cai_2024_algorithmic_info_disclosure_auctions}, 
in that the prescreening process acts as a form of information disclosure. The beliefs of admitted bidders about others’ valuations are updated based on the number of admitted bidders. However, unlike the existing literature, the prescreening mechanism in our model simultaneously affects both the number of participants and their beliefs, albeit in a limited way.

\smallskip 

\noindent
\textbf{Organization.}
We present the formal setup in \Cref{sec:model}. \Cref{sec:beliefs} analyzes bidders’ posterior beliefs and establishes an equivalent game. We then study second-price and first-price auctions in \Cref{sec:firstprice_second_auctions}, followed by an analysis of all-pay auctions in \Cref{sec:all-pay_auctions}. In \Cref{sec:rev_ranking}, we compare these auction formats and derive a revenue ranking result. \Cref{sec:extension} extends our analysis to include reserve prices, the joint design of auctions and prescreening, and the objective of maximizing the highest bid in all-pay auctions. \Cref{sec:conclusion} concludes the paper. Auxiliary results are provided in \Cref{app_sec:auxiliary_results} and \Cref{app_subsec:beliefs}, and all proofs appear in \Cref{app_sec:proofs}.

%% file: sections/2-model.tex
\section{Model Setup}\label{sec:model}

Consider a group of $m \geq 2$ potential bidders interested in acquiring a single indivisible item. The number of bidders, $m$, is exogenously given and is common knowledge to both the seller and all bidders. 
We assume both bidders and the seller are risk-neutral.
Let $\I^0$ denote the set of these $m$ bidders. Each bidder $i \in \I^0$ has a privately known valuation $V_i \in [0,1]$. 
Empowered by a valuation predictor, the seller observes a signal $s_i \in [0,1]$ about bidder $i$'s private valuation $v_i$.\footnote{Without loss of generality, we normalize both valuations and signals to lie in the interval $[0,1]$.} We denote the realizations of random variables $V_i$ and $S_i$ as $v_i$ and $s_i$, respectively.
The valuation-signal pair $(V_i, S_i)$ can be modeled as jointly drawn from a distribution $F \in \Delta([0,1]^2)$, independently across bidders. Each bidder $i$ observes only her own valuation $V_i$, while the seller observes the signal profile $S := (S_i)_{i \in \mathcal{I}_0}$, but not the bidders’ valuations. We assume that the signals are observed \emph{only} by the seller, reflecting the practical reality that the seller typically has exclusive access to bidders' historical data, which enables the use of a machine-learning-based predictor. 

In particular, we denote the marginal distribution of bidder valuations by $F_1 \in \Delta([0,1])$, which is assumed to have a strictly positive density $f_1$. Similarly, let $F_2 \in \Delta([0,1])$ represent the marginal distribution of signals, with strictly positive density $f_2$. Define $F_{1 \mid 2}(v_i \mid s_i)$ and $F_{2 \mid 1}(s_i \mid v_i)$ as the conditional cumulative distribution functions (CDFs). With this, we can interpret the signal-generating process as follows: bidders' private valuations $(V_i)_{i \in \I^0}$ are independently and identically distributed (i.i.d.) draws from the distribution $F_1$, referred to as the \emph{prior distribution}. Each bidder $i$ is privately informed of her valuation realization $V_i = v_i$. Subsequently, a signal $S_i \sim F_{2 \mid 1}(\cdot \mid V_i = v_i)$ is generated and privately observed by the seller. In the literature, it is standard to refer to $F_{2 \mid 1}$ as the \emph{predictor}, representing a probabilistic mapping from valuations to signals. However, with a slight abuse of notation, we refer to the joint distribution $F \in \Delta([0,1]^2)$ over valuation-signal pairs $(V_i, S_i)$ as the predictor throughout the paper, since, given $F_1$, $F_{2 \mid 1}$ uniquely determines $F$ (and vice versa).

Before the auction begins, the seller decides the number of bidders $n\,(\leq m)$ to be admitted, prior to any realizations. After receiving signals, the seller selects the top $n$ bidders based on the ranking of $(\mathbb{E}[V_i \mid S_i=s_i])_{i \in \I^0}$, i.e., posterior expectations of bidders' valuations after observing the signals. Ties are broken uniformly at random. To avoid the trivial case, we assume $n\geq 2$. Admitted bidders observe the admission outcome---that is, who is admitted and who is not---along with their own private valuations and the common knowledge shared by all, and update their beliefs about their opponents' private valuations via Bayes' rule. An auction is then conducted among these $n$ bidders. 
In this paper, we examine three common auction formats: first-price, second-price, and all-pay auctions.
Our base model considers auctions without a reserve price, but all results continue to hold in the presence of a reserve price, as shown in \Cref{subsec:reserve_price}. Mainly, the seller's objective is to maximize the expected revenue. 
 
\begin{assumption}
\label{assum:increasing_conditional_expectation}
The conditional expectation $\mathbb{E}[V_i\mid S_i=s_i]$ is non-decreasing
in $s_i\in[0,1]$.
\end{assumption}

For analytical tractability in analyzing the posterior beliefs of admitted bidders, we impose \Cref{assum:increasing_conditional_expectation} on the predictor $F(\cdot,\cdot)$ throughout the paper. This assumption ensures that the ranking of conditional expectations $(\mathbb{E}[V_i \mid S_i = s_i])_{i \in \I^0}$ remains the same 
as the ranking of signals $s = (s_i)_{i \in \I^0}$. 
\Cref{assum:increasing_conditional_expectation} implies that valuations and signals are non-negatively correlated.
\Cref{app_sec:discusions_assumption} provides a brief discussion of the analysis in the absence of this assumption.\footnote{Note that if $\mathbb{E}[V_i \mid S_i = s_i]$ is non-increasing in $s_i \in [0,1]$, all of our results still hold. Specifically, in this case, one can relabel the signal as $\tilde{s}_i := 1 - s_i$ and define a new predictor $\tilde{F}(v_i, \tilde{s}_i) := F(v_i, 1 - \tilde{s}_i)$.
}

By Sklar's theorem \citep{sklar_1959_copula_}, any two-dimensional joint distribution (i.e., any predictor) can be expressed as $F(v_i, s_i) = C\left(F_1(v_i), F_2(s_i)\right),$ where \( C: [0,1]^2 \to [0,1] \) is a \emph{copula} function, and \( F_1, F_2 \in \Delta([0,1]) \) are the marginal distributions.\footnote{A copula is a multivariate cumulative distribution function with uniform marginals on \([0, 1]\); see Definition 2.2.2 in \cite{nelsen_2006_copula}.}  As such, throughout the remainder of the analysis, we use copulas to model the relation between bidders' valuations and signals. This representation cleanly separates the marginal distributions from the dependence structure, enabling a flexible and rigorous framework for modeling the predictor. In what follows, we present several common copulas \( C(\cdot, \cdot) \) which, when combined with \emph{any} marginals, satisfy \Cref{assum:increasing_conditional_expectation}. The proof of \Cref{exm:copulas_increasing_expectation} is deferred to the appendix.

\begin{example}
\label{exm:copulas_increasing_expectation}
For any marginal distributions \( F_1, F_2 \in \Delta([0,1]) \), the predictors constructed using the following copulas satisfy \Cref{assum:increasing_conditional_expectation}:
   \begin{enumerate}
   \item Comonotonic copula: for $x,y\in[0,1]$,
      $ C(x,y) = \min\{ x,y\}.$

  \item Ali–Mikhail–Haq (AMH) copula: for $x,y\in[0,1]$, $  C(x,y) = \frac{xy}{1 - \alpha\cdot (1-x) (1-y)}$ with any $\alpha\in[0,1]$.
       
       \item Farlie–Gumbel–Morgenstern (FGM) copula: for $x,y\in[0,1]$, $  C(x, y) = xy(1 + \alpha\cdot (1 - x)(1 - y))$
      with any $ \alpha \in [0, 1]$.\footnote{Note that the feasibility conditions for both the AMH and FGM copulas require $\alpha \in [-1,1]$ (see Examples 3.1.2 and 3.1.9 in \cite{nelsen_2006_copula}). Under both copulas, when $\alpha \in [0,1]$, the conditional expectation $\mathbb{E}[V_i \mid S_i = s_i]$ is non-decreasing in $s_i \in [0,1]$, satisfying \Cref{assum:increasing_conditional_expectation}. When $\alpha \in [-1,0)$, the conditional expectation becomes non-increasing; however, as noted in Footnote~3, one can relabel the signal so that our results continue to hold.
      }
   \end{enumerate}

\noindent Moreover, if both predictors $C(F_1(\cdot), F_2(\cdot))$ and $\tilde{C}(F_1(\cdot), F_2(\cdot))$ satisfy \Cref{assum:increasing_conditional_expectation}, then any convex combination of them also satisfies \Cref{assum:increasing_conditional_expectation}.
\end{example}

The copulas in \Cref{exm:copulas_increasing_expectation} span a broad spectrum of dependence structures between valuations and signals. For instance, the comonotonic copula corresponds to the case of perfect predictions.
That is, if the predictor takes the form
\begin{align}
\label{eq:predictor_perfect}
    F(v_i, s_i) = \min\{F_1(v_i), F_2(s_i)\},
\end{align}
then knowing the signal $s_i$ uniquely determines the valuation $v_i$. We refer to \Cref{eq:predictor_perfect} as the \emph{perfect predictor}, abbreviated as $\PP$. On the other hand, for both the FGM and AMH copulas, setting $\alpha = 0$ yields a predictor in which the signal and valuation are statistically independent:
\begin{align}
\label{eq:predictor_null}
    F(v_i, s_i) = F_1(v_i) \cdot F_2(s_i).
\end{align}
This case is referred to as the \emph{null predictor}, abbreviated as $\NP$, since observing the signal provides no additional information about the bidder's valuation.

By the convex-combination property in \Cref{exm:copulas_increasing_expectation},\footnote{Our framework also accommodates scenarios where signals take the form $s_i = \pi(v_i) + \epsilon_i$, with $\epsilon_i$ representing a noise term and $\pi(\cdot)$ a deterministic function. Given the distributions of $v_i$ and $\epsilon_i$, this relationship implicitly defines a joint distribution over valuation-signal pairs. As long as the resulting joint distribution satisfies \Cref{assum:increasing_conditional_expectation}, our results remain valid.} the following predictor also satisfies \Cref{assum:increasing_conditional_expectation}:
with probability $\gamma$, the signal is fully informative, and with probability $1 - \gamma$, the signal is completely uninformative and independently drawn from $F_2$. Such a predictor can be constructed by defining:
\begin{align}
\label{eq:predictor_hallucination}
  F(v_i,s_i) = \gamma \cdot \min\{F_1(v_i),F_2(s_i)\} + (1-\gamma) \cdot F_1(v_i) F_2(s_i),
\end{align}
for some $\gamma \in [0,1]$. 
We refer to \Cref{eq:predictor_hallucination} as a \emph{hallucinatory predictor} (abbreviated as $\HP$) for the following reasons. The hallucinatory predictor has a conditional distribution given by $F_{2 \mid 1}(s_i \mid v_i) = \gamma \cdot \1\{s_i \geq F_2^{-1}(F_1(v_i))\} + (1 - \gamma) \cdot F_2(s_i)$. With probability $\gamma$, the signal is fully informative: $s_i = F_2^{-1}(F_1(v_i))$, or alternatively $v_i = F_1^{-1}(F_2(s_i))$ almost surely. With probability $1 - \gamma$, the signal is independent of the bidder’s valuation. Intuitively, the parameter $\gamma$ captures the \emph{prediction accuracy} of the hallucinatory predictor: a higher $\gamma$ indicates that the signal is more informative about the bidder’s true valuation.

A natural question is how to define prediction accuracy for general predictors that lack a parametric form. Within our framework, we define prediction accuracy using the concept of \emph{positive quadrant dependence}, as formalized below.

\begin{definition}[Prediction Accuracy]
\label{def:accuracy}
Consider two predictors $F$ and $\tilde{F}$ that share identical marginal distributions, i.e., $F_1 = \tilde{F}_1$ and $F_2 = \tilde{F}_2$. We say that predictor $F$ exhibits \emph{higher prediction accuracy} (or is \emph{more accurate}) than predictor $\tilde{F}$ if $F$ dominates $\tilde{F}$ in the \emph{positive quadrant dependence (PQD) order}, denoted by $F \pqdorder \tilde{F}$. That is, $\tilde{F}(v_i, s_i) \leq F(v_i, s_i)$ for all $(v_i, s_i) \in [0,1]^2$.
\end{definition}

A few remarks are in order. First, when comparing different predictors, we assume that the marginal distributions of both valuations and signals are fixed and identical across predictors (although we do \emph{not} require $F_1 = F_2$). Consequently, any differences across predictors arise solely from the copula functions, which capture the underlying dependence structure between bidders' valuations and signals. %
Second, the PQD order is \emph{not} equivalent to the standard notion of two-dimensional first-order stochastic dominance \citep{shaked_2007_book_stochastic}; %
rather, it captures the strength of dependence between bidders' valuations and signals. A notable feature of the PQD order is that it preserves the ordering of several widely used measures of association, including the correlation coefficient, Kendall's~$\tau$, Spearman's~$\rho$, and Blomqvist's~$q$ \citep{yanagimoto_1969_partial_orderings, shaked_2007_book_stochastic}. %
As such, a more accurate predictor, as defined in \Cref{def:accuracy}, implies stronger correlation between signals and valuations, meaning that the signal 
contains more information about
the bidder's true valuation. We further compare the accuracy of the different predictors introduced in \Cref{exm:copulas_increasing_expectation} using \Cref{def:accuracy}, and find that the resulting ordering of prediction accuracy is consistent with intuitive notions of what constitutes a more accurate predictor. The proof of the results in \Cref{obs:higher_accurate_common_predictors} is provided in the appendix.

\begin{observation}
\label{obs:higher_accurate_common_predictors}
Given any marginal distributions $F_1$ and $F_2$, and any $\gamma \in [0,1]$, it holds that
\begin{align*}
 \text{\emph{perfect predictor \eqref{eq:predictor_perfect}}} &\pqdorder 
 \text{\emph{hallucinatory predictor \eqref{eq:predictor_hallucination} with parameter $\gamma$}} \\
 &\pqdorder \text{\emph{AMH predictor with $\alpha=\gamma$ (as defined in \Cref{exm:copulas_increasing_expectation}(ii)}}) \\
 &\pqdorder \text{\emph{FGM predictor with $\alpha=\gamma$ (as defined in \Cref{exm:copulas_increasing_expectation}(iii)}}) \\
 &\pqdorder \text{\emph{null predictor \eqref{eq:predictor_null}}}.
\end{align*}
Moreover, the prediction accuracy of both the AMH and FGM predictors increases with $\alpha$, and the prediction accuracy of the hallucinatory predictor increases with $\gamma$.
\end{observation}

\textit{Discussions on Model Assumptions.} In our model, we assume that the seller decides the number of admitted bidders prior to signal realizations. Practically, auctioneers (such as Google Ads and contest organizers) typically face hard resource constraints, requiring them to limit the number of participants to a predetermined quantity. This assumption also ensures technical tractability.
It is well known that designing a (globally) optimal mechanism when the seller holds private information is challenging \citep{myerson_1983_mechanismdesign_informedprincipal}, as the mechanism itself may reveal additional information. By assuming that the seller commits to the admission rule ex ante, we are able to isolate and examine the informational effect of prescreening as driven by the number of admitted bidders. %
Moreover, in the principal-agent model, \citet{maskin_1990_principalagent_informedprincipal} show that it is without loss of generality to restrict the mechanism to be independent of the principal's private information. This assumption is, therefore, without loss of generality, provided this result carries over to the context of prescreening. Whether such an extension holds remains an open question for future research.

We also assume that only the seller observes the valuation signals for all bidders, using a common predictor $F(\cdot, \cdot)$. This assumption has two key aspects. First, the restriction that only the seller observes the signals reflects practical scenarios in which the seller has exclusive access to relevant data, such as historical bidding behavior, and is therefore uniquely positioned to apply machine learning algorithms to predict bidder valuations.
Second, the assumption of a common predictor for all bidders is practically justified by the fact that the seller typically deploys a single (best) prediction model uniformly across all participants. From a theoretical perspective, allowing predictors to vary across bidders leads to significant complications. 
In particular, heterogeneous predictors lead to (inherently) asymmetric beliefs, causing symmetric equilibria to fail to exist. However, explicitly characterizing the asymmetric equilibria is typically intractable \citep{maskin_2000_asymmetricauction, lebrun_2006_uniqueness_firstprice_auctions}.

We assume that the seller selects the top bidders based on the conditional expectation $\mathbb{E}[V_i \mid S_i = s_i]$, which is a natural and commonly used criterion. More generally, upon observing the signal $s_i$, the seller's posterior belief about bidder $i$'s valuation is given by the conditional distribution $F_{1 \mid 2}(\cdot \mid S_i = s_i)$, and the seller may alternatively rank bidders based on other statistics, such as quantiles. In fact, if we strengthen \Cref{assum:increasing_conditional_expectation} by assuming \emph{positive regression dependence} (PRD) %
of $V_i$ on $S_i$---formally, that the conditional CDF $F_{1 \mid 2}(\cdot \mid S_i = s_i)$ is non-increasing in $s_i \in [0,1]$, denoted by $\PRD(V_i \mid S_i)$\footnote{This implies \Cref{assum:increasing_conditional_expectation}, since $\mathbb{E}[V_i \mid S_i = s_i] = \int_0^1 \left[1 - F_{1 \mid 2}(v_i \mid s_i)\right] dv_i$ is non-decreasing in $s_i$.}---then selection based on the conditional expectation or on \emph{any} quantile yields the same ranking. This is because under $\PRD(V_i \mid S_i)$, the posterior distribution exhibits first-order stochastic dominance: $F_{1 \mid 2}(\cdot \mid S_i = s_i) \succeq_{\text{FOSD}} F_{1 \mid 2}(\cdot \mid S_i = s'_i)$ for all $s_i \geq s'_i$. The property $\PRD(V_i \mid S_i)$ holds for a wide range of predictors, including all cases in \Cref{exm:copulas_increasing_expectation}.

\textit{Notations and Glossary.}
For ease of reference, we summarize below the main notation and glossary used throughout the paper.
For convenience, we may write the admitted number as $n \in [2,m]$, meaning explicitly that $n \in \{2,3,\dots,m\}$.
Let $\C^{n}_{k} \coloneqq \binom{n}{k}$ denote the binomial coefficient for integers $0 \le k \le n$. Given a vector $x \in \mathbb{R}^{n}$
and an index set $\mathcal{I} \subseteq \{1,\dots,n\}$, let $x_{\mathcal{I}}$ be the subvector of $x$ whose components are indexed by $\mathcal{I}$. 

A function $g \colon \mathbb{R}^{n} \to \mathbb{R}$ is \emph{symmetric} if, for every bijection $\pi \colon \{1,\dots,n\} \to \{1,\dots,n\}$, we have $g(x_{1},\dots,x_{n}) \;=\; g\bigl(x_{\pi(1)},\dots,x_{\pi(n)}\bigr)$. A function $g$ is \emph{supermodular} on a domain $\mathcal{X} \subseteq \mathbb{R}^{n}$ if, for all $x,x' \in \mathcal{X}$, $g(x \wedge x') \;+\; g(x \vee x') \;\ge\; g(x) \;+\; g(x'),$ where $x \wedge x'$ and $x \vee x'$ denote the component-wise minimum and maximum of $x$ and $x'$, respectively. A function $g$ is \emph{affiliated} on $\mathcal{X}$ if, for all $x,x' \in \mathcal{X}$, $g(x \wedge x')\,g(x \vee x') \;\ge\; g(x)\,g(x').$ For a strictly positive function $g$, affiliation is equivalent to \emph{log-supermodularity}, i.e.,\ the property that $\ln g$ is supermodular.  
When an affiliated function $g$ serves as a joint density, we say that a random vector $V$ sampled from $g$ is \emph{affiliated}.

%% file: sections/3-Beliefs.tex
\section{Beliefs and An Equivalent Game}
\label{sec:beliefs}

In this section, we characterize the posterior beliefs of the admitted bidders, given a fixed admission number $n$. In particular, we show in \Cref{subsec:equivalent_game} that an auction with prescreening is equivalent to a standard auction without prescreening but with correlated bidder valuations. These results serve as a foundation for the analysis and will later be used to characterize equilibrium strategies when we focus on specific auction formats. 

Let $\I \subseteq \I^0$ denote a realized set of admitted bidders with $|\I| = n$. Given $\I$, define $\bar{\I} := \I^0 \setminus \I$ as the set of eliminated bidders. For an admitted bidder $i \in \I$, let $\I_{-i}$ denote the set of the remaining $n - 1$ admitted bidders. Let $v_{-i} = (v_j)_{j \in \I_{-i}}$ denote the vector of valuations of all admitted bidders other than the admitted bidder~$i$, and let $v = (v_i, v_{-i})$ denote the vector of valuations of all admitted bidders. Finally, let $\beta(v_{-i} \mid \I, v_i; n, F)$ denote bidder $i$'s belief---represented by the joint density---over the valuations $V_{-i}$ of all other admitted bidders, conditional on the admission set $\I$ and bidder $i$'s own valuation $v_i$, when the predictor is $F$.\footnote{Note that in the posterior belief, there is no need to condition separately on the admitted number $n$, as this information is already captured by the admission set $\I$. Additionally, in the expression ``$\beta\bigl(v_{-i} \mid \I, v_i; n, F\bigr)$,'' the symbols $n$ and $F$ denote parameters rather than conditioning variables.}

\begin{lemma}[Posterior Beliefs]
\label{thm:beta_posterior}
Admitted bidder $i$'s posterior belief $\beta(v_{-i}\mid\I,v_i;n,F)$ over the private valuations $V_{-i}$ of all other admitted bidders is
 \begin{align}\label{eq:beta_posterior}
 \beta(v_{-i} \mid \I,v_i;n,F)
 = \underbrace{\kappa(v_i;n,F)}_{\normalfont{\textrm{normalizing term}}}
\cdot 
\underbrace{\psi(v;n,F)}_{\normalfont{\textrm{admission probability}}}
 \cdot 
 \underbrace{\prod_{j\in \I_{-i}} f_1(v_j)}_{\normalfont{\textrm{prior}}}
 ,~\forall\, v_{-i}\in [0,1]^{n-1}.
\end{align}
The admission probability is given by: for all $v\in [0,1]^n$,\footnote{The admission probability equals one when all bidders are admitted, i.e., when $n = m$. For notational convenience, we define the integral in~\Cref{eq:admission_prob} to be one in this case.} 
\begin{align}
\label{eq:admission_prob}
\psi(v;n,F)=  \Pr\left\{
\max_{j\in \barI} S_j \leq \min_{k\in \I}S_k
\,\biggm|\, v\right\} =
\int_0^1\prod_{k\in \I}[1-F_{2\mid 1}(x\mid v_k)] dF_2^{m-n}(x).
\end{align}
The normalizing term $\kappa(v_i;n,F)$ ensures that $\int_{[0,1]^{n-1}}\beta(v_{-i}\mid \I,v_i;n,F)dv_{-i}=1$.
\end{lemma}

\Cref{thm:beta_posterior} reveals a clear structure for the posterior belief of admitted bidder~$i$: it is proportional to the product of the prior belief and the admission probability $\psi(v; n, F)$. 
Due to \Cref{assum:increasing_conditional_expectation}, the admission probability $\psi(v; n, F)$ is the probability that all eliminated bidders' signals $S_{\bar{\I}}$ are less than or equal to those of all admitted bidders $S_{\I}$, conditional on the vector of \emph{admitted} bidders' valuations~$v$. %
It can be shown that the admission probability $\psi(v; n, F)$ is a symmetric function in $v \in [0,1]^n$, a property that plays a crucial role in ensuring the existence of a symmetric equilibrium.

Under the null predictor, selection is performed uniformly at random. In this case, the admission probability is given by $\psi(v; n, \gamma) = \frac{1}{\C^m_n}$,\footnote{This follows from the fact that under the null predictor, we have $\psi(v; n, \NP) = \Pr\{ \max_{j \in \bar{\I}} S_j \leq \min_{i \in \I} S_i \mid v \} = \Pr\{ \max_{j \in \bar{\I}} S_j \leq \min_{i \in \I} S_i \}$, which is the probability that the minimum of $n$ i.i.d.\ random variables exceeds the maximum of $m - n$ i.i.d.\ random variables drawn from the same distribution $F_2$. This probability equals $\frac{1}{\C^m_n}$.} which is independent of $v$, and the posterior belief simplifies to $\beta(v_{-i} \mid \I, v_i; n, F) = \prod_{j \in \I_{-i}} f_1(v_j)$, which coincides with the prior. Similarly, the posterior belief also coincides with the prior when $n = m$, i.e., in the absence of prescreening, regardless of the predictor. On the other hand, under the perfect predictor $\PP$ (as defined in \Cref{eq:predictor_perfect}),
we have $S_k = F_2^{-1}(F_1(V_k))$ for all $k \in \I^0$, and thus:
\begin{align}
\label{eq:gamma_1_admissionprob}
\psi(v;n,\PP)
=\Pr\left\{
\max_{j\in \barI} F_2^{-1}(F_1(V_j)) \leq \min_{k\in \I}F_2^{-1}(F_1(V_k)) \,\biggm|\, v\right\}
=\Pr\left\{
\max_{j\in \barI} V_j \leq \min_{k\in \I}V_k \,\biggm|\, v\right\}
=F^{m-n}_1\left(\min_{k\in \I}v_k\right).
\end{align}
In this case, it can be shown that the admission probability $\psi(v; n, \PP)$ is log-supermodular (and also supermodular) in $v \in [0,1]^n$.

Since the admission probability $\psi(v; n, F)$ is a symmetric function, bidder~$i$'s posterior belief $\beta(v_{-i} \mid \I, v_i; n, F)$ is also symmetric in $v_{-i} \in [0,1]^{n-1}$. As a result, the one-dimensional marginals derived from it are all identical. We denote this marginal distribution as $\beta^{\mathsf{mar}}(\cdot \mid \I, v_i; n, F)$. Note that, unless $n = m$ or under the null predictor, generally $\beta(v_{-i} \mid \I, v_i; n, F) \neq \prod_{j \in \I_{-i}} \beta^{\mathsf{mar}}(v_j \mid \I, v_i; n, F)$. That is, bidder~$i$'s posterior beliefs about the valuations of other admitted bidders are, in general, correlated.

The properties of the normalizing term are summarized in \Cref{lem:kappa_vi}. \Cref{lem:kappa_vi}(i) establishes that $\kappa(v_i; n, F)$ is non-increasing in $v_i$ for any predictor satisfying $\PRD(S_i \mid V_i)$, i.e., $F_{2 \mid 1}(\cdot \mid v_i)$ is non-increasing in $v_i \in [0,1]$. Intuitively, under $\PRD(S_i \mid V_i)$, a bidder with a higher valuation $v_i$ is more likely to be admitted, as the seller is more likely to observe a higher signal. This increases the value of $\psi(v; n, F)$, thereby reducing the value of the normalizing term $\kappa(v_i; n, F)$. Notably, all examples presented in \Cref{exm:copulas_increasing_expectation}, including any convex combinations thereof, satisfy this condition. Furthermore, \Cref{lem:kappa_vi}(ii) shows that $\kappa(v_i; n, F)$ is bounded below by $\C^{m-1}_{n-1}$, which is always at least one.

\begin{lemma}[Normarlizing Term]
\label{lem:kappa_vi}
For any $n\in[2,m]$, we have 
\begin{enumerate}[(i)]
 \item for any predictor $F$ satisfying $\PRD(S_i\mid V_i)$, $\kappa(v_i;n,F)$ is non-increasing in $v_i\in [0,1]$;
     \item for any $v_i\in[0,1]$, $\kappa(v_i;n,F)\geq \C^{m-1}_{n-1}$.
\end{enumerate}
\end{lemma}

\subsection{An Equivalent Game}
\label{subsec:equivalent_game}

Although the private valuations of potential bidders are i.i.d., the prescreening process induces correlation in the posterior beliefs of admitted bidders regarding the valuations of other admitted bidders. A natural question arises: does there exist an equivalent game (without the prescreening process) involving only the admitted bidders, in which their valuations follow a joint distribution such that each bidder’s belief about the valuations of others, conditional on their own valuation, coincides with the posterior belief given in~\Cref{eq:beta_posterior}? This question is important because there is an extensive literature on auctions with \emph{affiliated} valuations \citep{milgrom_1982_auctiontheory_competitive_bidding, krishna_1997_all_pay_affiliation, pekevc_2008_revenueranking_random_number, castro_2007_affiliation_positive_dependence}. If such an equivalent game exists, we can potentially leverage existing results from that stream of literature to analyze equilibrium behavior and examine how those results are related
to our setting---particularly in cases, as we will discuss later, where the valuations are correlated but not affiliated.
Fortunately, the answer is yes, and we are able to characterize the joint density over the bidders’ valuations in an equivalent game, as stated in the theorem below. %

\begin{theorem}[Equivalent Game]
\label{lem:joint_dis_g}
Given a set of bidders $\I$ with $|\I|=n$, we define
\begin{align}
\label{eq:joint_dist_g}
g(v;\I,n,F) = \C^m_n\cdot \psi(v;n,F)\cdot \prod_{j\in\I} f_1(v_j),~ \forall v\in [0,1]^n.   
\end{align}
\begin{enumerate}
    \item $g(v;\I,n,F)$ is the unique (almost everywhere) symmetric joint density over valuations $v\in[0,1]^n$ such that %
    the conditional density $g(v_{-i}\mid v_i;\I,n,F)=\beta(v_{-i}\mid \I,v_i;n,F)$ for any $i\in \I$ and any $v\in [0,1]^n$.
    \item An auction with prescreening that admits $n$ bidders (from a total of $m$ potential bidders) under an i.i.d.\ prior $f_1$ and predictor $F$ is equivalent to a standard auction, i.e., one without prescreening, with $n$ bidders whose joint valuation distribution is given by $g(\cdot; \I, n, F)$ in \Cref{eq:joint_dist_g}. This equivalence holds in the sense that bidders' posterior beliefs, equilibrium strategies, and the seller’s ex-ante expected utility are all identical across the two settings.
\end{enumerate}
\end{theorem}

For simplicity, we omit $\I$ from $g(v; \I, n, F)$ and other related terms (introduced later) whenever there is no risk of confusion. Similar to \Cref{thm:beta_posterior}, \Cref{lem:joint_dis_g}(i) shows that the joint density $g(\cdot; n, F)$ is proportional to the prior multiplied by the admission probability $\psi(v; n, F)$. An interesting observation from~\Cref{eq:joint_dist_g} is that the normalizing term becomes the binomial coefficient $\C^m_n$ independent of the predictor $F$, whereas the normalizing term $\kappa(v_i; n, F)$ in \Cref{thm:beta_posterior} depends on it.

\Cref{lem:joint_dis_g}(ii) establishes the equivalence between the prescreening setting and a standard auction with correlated valuations. In the latter, the vector of valuations is jointly drawn from the distribution $g(v; n, F)$, and each bidder observes only their own valuation. Bidder~$i$'s belief about the others' valuations is given by the conditional distribution $g(v_{-i} \mid v_i; n, F)$, which coincides with $\beta(v_{-i} \mid \I, v_i; n, F)$. Consequently, all participating bidders hold the same beliefs in both settings, resulting in identical equilibrium strategies. From the seller's perspective, the two settings are also equivalent. This is because the seller decides the number $n$ \emph{before} any realizations occur. When making this choice, the seller has no additional information and must contemplate the entire game process.\footnote{A similar argument arises in many sequential games, including multi-stage auctions \citep{moldovanu_2006_contest_architecture} and Bayesian persuasion \citep{kamenica_2011_persuasion}.} 
As a result, the seller's decision can be viewed as simultaneously selecting the number of participants, $n$, and the joint distribution $g(v;n,F)$ over bidders' valuations, though these two elements are intertwined. This result holds for \textit{any} given auction format.

\begin{remark}
\label{remark:affiliation}
When all bidders are admitted, or under the null predictor or the perfect predictor, the joint density $g(v; n, F)$ is affiliated in $v \in [0,1]^n$; otherwise, it is in general \emph{not} affiliated.
\end{remark}

Unlike much of the auction literature with correlated valuations, where \emph{affiliation}, i.e., the assumption that the joint density is log-supermodular, is often imposed for tractability \citep{milgrom_1982_auctiontheory_competitive_bidding, krishna_1997_all_pay_affiliation, pekevc_2008_revenueranking_random_number, castro_2007_affiliation_positive_dependence}, the joint density $g(\cdot; n, F)$ in our framework is generally \emph{not} affiliated, except in certain special cases, as discussed in \Cref{remark:affiliation}. The cases of admitting all bidders or under the null predictor are straightforward: in these cases, we have $g(v; n, F) = \prod_{i \in \I} f_1(v_i)$, which is log-supermodular by construction. Under the perfect predictor, we have $\psi(v; n, \PP) = F_1^{m-n}\bigl(\min_{i \in \I} v_i\bigr)$ based on \Cref{eq:gamma_1_admissionprob}, which can also be shown to be log-supermodular in $v \in [0,1]^n$. Since the product preserves log-supermodularity, it follows that $g(v; n, F)$ is log-supermodular in this case as well. However, for general cases, the admission probability $\psi(v; n, F)$ is not log-supermodular, and thus $g(v; n, F)$ is not affiliated.

\smallskip

We now examine comparative statics of the distribution $g(\cdot; n, F)$, focusing in particular on how it varies with the admitted number of bidders $n$ and the prediction accuracy as defined in \Cref{def:accuracy}. This analysis sheds light on how the joint distribution of bidders’ valuations is affected by the seller’s selection decision and the precision of the prediction technology.
Since the joint density $g(\cdot; n, F)$ is symmetric, the one-dimensional marginals derived from it are all identical. We denote this marginal CDF by $G^{\mathsf{mar}}(\cdot; n, F)$. Let $G^{\lar}(\cdot; n, F)$ denote the CDF of the largest order statistic of a random vector $V$ drawn from $g(\cdot; n, F)$; that is, $G^{\lar}(\cdot; n, F)$ is the CDF of $\max_{i \in \I} V_i$ with $V \sim g(\cdot; n, F)$.

\begin{lemma}[Impact of the Admitted Number and the Prediction Accuracy]
\label{lem:FOSD_number}
We have
\begin{enumerate}
    \item For any predictor $F$ such that $\PRD(V_i\mid S_i)$ holds, the marginal distribution $G^{\mar}(\cdot;n,F)$ decreases in $n\in [2,m]$ in the sense of the first-order stochastic dominance; that is, $ G^{\mar}(x;n^\prime,F) \leq G^{\mar}(x;n,F)$, for any $x\in [0,1]$ and $n \geq n^\prime$. 
      
    \item For any predictior $F $ and any integer $k \in [1, n^\prime]$, the distribution of the $k^{\mathsf{th}}$ largest valuation sampled from $g(\cdot; n, F)$ first-order stochastically dominates that of the $k^{\mathsf{th}}$ largest valuation sampled from $g(\cdot; n', F)$, for any $n \geq n^\prime\in [2,m]$. In particular, the two distributions are identical under the perfect predictor.

    \item For any admitted number $n \in [2, m]$, both the marginal distribution and the distribution of the $k^{\mathsf{th}}$ largest valuation, for any $k\in [1,n]$, increase with prediction accuracy in the sense of first-order stochastic dominance.

\end{enumerate}
\end{lemma}

\Cref{lem:FOSD_number}(i) implies that the average valuation of an admitted bidder is higher when fewer bidders are admitted under any predictor satisfying $\PRD(V_i\mid S_i)$. 
Note that $\PRD(V_i \mid S_i)$ is stronger than \Cref{assum:increasing_conditional_expectation}. Nonetheless, all predictor examples presented in \Cref{exm:copulas_increasing_expectation}, including those constructed via their convex combinations, satisfy $\PRD(V_i \mid S_i)$. \Cref{lem:FOSD_number}(ii) shows that when focusing on the top bidders, i.e., those with the highest valuations, the average valuation of a top-$k$ bidder decreases as fewer bidders are admitted. At first glance, these two results may appear to contradict one another, but they are, in fact, consistent---since the comparisons are made under different numbers of admitted bidders. \Cref{lem:FOSD_number}(ii) arises because the seller’s observed signals are not fully accurate unless under the perfect predictor. As a result, bidders with the highest valuations may be screened out, and this likelihood decreases as the seller admits more bidders. On the other hand, admitting more bidders increases the probability of including individuals with lower signals. Even though the signals are imperfect, a lower signal generally corresponds to a lower valuation on average. This leads to a lower 
valuation among average admitted bidders when more are admitted, as demonstrated in \Cref{lem:FOSD_number}(i). The first two parts of \Cref{lem:FOSD_number} play a crucial role in determining the optimal number of admitted bidders in later sections. 

With respect to prediction accuracy, when the predictor is perfect, bidders with the highest valuations are guaranteed to be admitted. Consequently, the distributions of the order statistics remain unaffected by the number of admitted bidders, as established in \Cref{lem:FOSD_number}(ii). More generally, as the predictor becomes more accurate, admitted bidders are increasingly likely to have high valuations, and the distribution of the $k^{\mathsf{th}}$ order statistic, for any $k\in [1,n]$, shifts upward in the sense of first-order stochastic dominance, as shown in \Cref{lem:FOSD_number}(iii).
These results play a critical role in analyzing the impact of prediction accuracy on auction revenue.

Unfortunately, results analogous to those in \Cref{lem:FOSD_number} do not hold for the conditional distribution $g(v_{-i} \mid v_i; n, F)$, i.e., the bidders' posterior beliefs. This is because the conditional distribution also depends critically on the bidder’s realized valuation. %
We find that under the perfect predictor, a bidder’s belief about the valuations of others exhibits a stochastic ordering with respect to the bidder's own private valuation.

\begin{lemma}
\label{prop:stochastic_domiance_gamma_1}
Under the perfect predictor, we have
\begin{align*}
\beta^{\mar}(\cdot \mid \I,v_j;n,\PP) \succeq_{\fosd} \beta^{\mar}(\cdot \mid \I,v_i;n,\PP) \succeq_{\fosd} f_1,~ \forall\, v_j \geq v_i.
\end{align*}
\end{lemma}

Recall that $\beta^{\mar}(\cdot \mid \I, v_i; n, \PP)$ denotes bidder~$i$'s posterior belief about a single opponent’s valuation, conditional on her own valuation $v_i$. \Cref{prop:stochastic_domiance_gamma_1} states that, under the perfect predictor, an admitted bidder with a higher valuation perceives her opponents as more likely to have high valuations, and that her posterior belief stochastically dominates the prior. 
Beyond the case of the perfect predictor, the results in \Cref{prop:stochastic_domiance_gamma_1} generally do \textit{not} hold, except in the trivial cases where $n = m$ or under the null predictor. 
To provide clear intuition, let us consider the hallucinatory predictor \eqref{eq:predictor_hallucination} as an example.
Under such a predictor, two competing forces shape a bidder's posterior belief: (1) with probability $\gamma$, the seller’s observed signal equals the bidder’s true valuation. Conditioning on this event, a bidder with a low valuation infers that other admitted bidders are also likely to have low valuations, as shown in \Cref{prop:stochastic_domiance_gamma_1}; (2) with probability $1 - \gamma$, the seller’s signal is drawn randomly from the prior distribution, making the bidder's posterior belief independent of her private valuation. Conditioning on this event, the marginal distribution of a bidder’s posterior belief generally decreases in $n$ and increases in $\gamma$, as implied by \Cref{lem:FOSD_number}. If a bidder with a lower valuation is admitted, she is more likely to attribute her admission to randomness, i.e., case (2), which disrupts the monotonic relationship between her own valuation and posterior belief. As a result, posterior beliefs generally fails to exhibit a clear stochastic order when the predictor is imperfect.%

%% file: sections/5-First-price-Auctions.tex
\section{Second-price and First-price Auctions}
\label{sec:firstprice_second_auctions}

Next, we focus on three commonly used auction formats and study the bidders’ equilibrium strategies and the seller’s optimal choice of the number of bidders to admit. 
We start by examining the second-price auction in \Cref{sec:second_price_auction}, followed by the first-price auction in \Cref{subsec:first_price_auctions}. All-pay auctions will be considered in \Cref{sec:all-pay_auctions}.

\subsection{Second-price Auctions}
\label{sec:second_price_auction}

In a second-price auction, all participating bidders simultaneously submit their bids. The bidder with the highest bid wins the auction but pays the second-highest bid, while all other bidders incur no payment. 
Since bidders know their own private valuations,
bidding one’s true valuation is a dominant strategy.

\begin{proposition}[\citealt{milgrom_1982_auctiontheory_competitive_bidding}]
\label{prop:equilibrium_secondprice}
In second-price auctions, truthful bidding is a dominant equilibrium strategy.
\end{proposition}

Since only the winner pays, and specifically pays the second-highest bid, it follows from \Cref{prop:equilibrium_secondprice} that the expected revenue in a second-price auction equals the expected second-highest valuation in the valuation vector $V$, where $V$ is drawn from the joint distribution $g(v; n, F)$ defined in \Cref{eq:joint_dist_g}. Combined with \Cref{lem:FOSD_number}, which characterizes the comparative statics of the second-highest valuation with respect to the number of admitted bidders and the prediction accuracy, we obtain the following results.

\begin{theorem}[Optimal Prescreening and Revenue Loss]
\label{thm:opt_admittednumber_secondprice}
In second-price auctions:
\begin{enumerate}
  \item The expected revenue is (weakly) increasing in the number of admitted bidders. Consequently, it is optimal to admit all bidders.
    \item For any fixed number of admitted bidders, the revenue loss from admitting fewer bidders is (weakly) decreasing in the prediction accuracy. Moreover, under a perfect predictor, there is no revenue loss, i.e., admitting \emph{any} number of bidders yields the same expected revenue.
\end{enumerate}
\end{theorem}

\Cref{thm:opt_admittednumber_secondprice}(i) primarily follows from \Cref{lem:FOSD_number}(ii), which states that the second-highest valuation increases with the number of admitted bidders. Consequently, admitting all bidders is optimal. \Cref{thm:opt_admittednumber_secondprice}(ii) illustrates that although admitting fewer bidders due to potential hard resource constraints is suboptimal, a more accurate predictor can compensate for this revenue loss. This result primarily follows from \Cref{lem:FOSD_number}(iii), showing that the second-highest valuation increases with prediction accuracy. In the extreme case of a perfect predictor, there is no revenue loss regardless of the number of admitted bidders. Intuitively, under a perfect predictor, the seller consistently admits the truly highest-valued bidders, leaving the second-largest valuation unchanged. Thus, admitting any number of bidders yields the same expected revenue.

\begin{remark}[\textbf{Predictions vs. Auctions}]
\Cref{thm:opt_admittednumber_secondprice} presents a “predictions vs. auctions” result that is reminiscent in spirit of the “auctions vs. negotiations” result in \cite{bulow_1994_auctions_negotiations}. Specifically, our analysis shows that, given the predictor, admitting more bidders in second-price auctions leads to higher expected revenue. However, revenue losses from admitting fewer bidders can be effectively mitigated by employing more accurate predictors. Moreover, the expected revenue under a perfect predictor is (weakly) higher than that under any non-perfect predictor, even when more bidders are admitted in the latter case.

\end{remark}

\subsection{First-price Auctions}
\label{subsec:first_price_auctions}

We now turn to the first-price auction, in which all participating bidders simultaneously submit their bids. The bidder with the highest bid wins and pays her bid, while all other bidders make no payment. In what follows, we first characterize the equilibrium bidding strategy and then analyze the seller’s optimal prescreening policy.

\subsubsection{Equilibrium Analysis}
\begin{definition}[Equilibrium Definition]
\label{def:equilibrium_1stprice}
For first-price auctions, a strategy $\sigma:[0,1]\to \mathbb{R}_+$ is a symmetric Bayesian-Nash equilibrium if, for all $i\in \I$ and all $v_i\in[0,1]$, it satisfies
\begin{align}
\label{eq:def_equilibrium_firstprice}
  \sigma(v_i)\in {\normalfont{\argmax}}_{b_i\in \mathbb{R}_+}~  \Pr\{\sigma(V_j)<b_i,\forall j\in \I_{-i}\} \cdot \left[v_i - b_i\right],  
\end{align}
where $V_{-i}$ is drawn from bidder $i$'s belief $g(v_{-i}\mid v_i;n,F)$, as defined in \Cref{eq:beta_posterior}.
\end{definition}

We assume that the strategy $\sigma(\cdot)$ is differentiable and strictly increasing, as is standard in the auction literature \citep{milgrom_1982_auctiontheory_competitive_bidding,milgrom_2004_putting_auctiontheory_to_work,krishna2009auction}. Under this assumption, the winning probability $\Pr\{\sigma(V_j) < b_i,\ \forall j \in \mathcal{I}_{-i}\}$ can be simplified as follows:
\begin{align*}
 \Pr\{\sigma(V_j)<b_i,\forall j\in \I_{-i}\}
  = \Pr\left\{ \max_{j\in \I_{-i}}V_j < \sigma^{-1}(b_i)\right\} 
  = H(\sigma^{-1}(b_i)\mid v_i;n,F),
\end{align*}
where we denote $H(\cdot \mid v_i; n, F)$ as the CDF of the largest valuation among the $n-1$ other bidders, drawn from the conditional density $g(v_{-i} \mid v_i; n, F)$. Formally,
\begin{align}
\label{eq:H_CDF}
H(x \mid v_i; n, F) := \int_{[0, x]^{n-1}} g(v_{-i} \mid v_i; n, F) \, dv_{-i}.
\end{align}
Let $h(x \mid v_i; n, F) = \frac{\partial H(x \mid v_i; n, \gamma)}{\partial x}$ denote the corresponding density function. The function $H(x \mid v_i; n, F)$ can be interpreted as the winning probability of a bidder with valuation $v_i$ who bids $\sigma(x)$, while all other bidders follow the same symmetric and strictly monotonic (SSM) strategy $\sigma(\cdot)$.\footnote{Throughout, we use the term \emph{monotonicity} to refer specifically to strictly increasing strategies, as strictly decreasing equilibrium strategies are not feasible in our context.} With this, we derive the equilibrium strategy under first-price auctions as follows.

\begin{proposition}[Equillibrium Strategy in First-price Auctions]
\label{thm:SSM_firstprice}
In first-price auctions:
\begin{enumerate}[(i)]
    \item  An SSM equilibrium strategy exists if $\FP(\tilde{v}_i, v_i; n, F)$, defined by 
    \Cref{eq:firstprice_condition} in the \Cref{app_sec:proofs},
    is non-negative for all $\tilde{v}_i \in [0, v_i]$ and non-positive for all $\tilde{v}_i \in [v_i, 1]$, for any given $v_i \in [0, 1]$. This condition is always satisfied under the null or perfect predictors, or when all bidders are admitted.

    \item If an SSM equilibrium strategy exists, it is unique and is given by the following expression:
    \begin{align}
    \label{eq:SSM_equilibrium_firstprice}
     \sigma^{\FP}
     (v_i;n,F) = v_i - \int_0^{v_i}\exp\left(-\int_{t}^{v_i} \RHR(x; n, F)dx\right)dt,~ \forall v_i \in [0,1],
    \end{align}
    where $\RHR(x; n, F) := \frac{h(x \mid x; n, F)}{H(x \mid x; n, F)}.$
   \item Under the perfect predictor, the equilibrium strategy is \emph{independent} of the number of admitted bidders. Specifically, for \emph{any} $n \in [2, m]$, we have
\begin{align*}
    \sigma^{\FP}(v_i; n, \PP) = v_i - \int_0^{v_i} \frac{F_1^{m-1}(x)}{F_1^{m-1}(v_i)} \, dx, \quad \forall v_i \in [0,1].
\end{align*}
This expression also coincides with the equilibrium strategy $\sigma^{\FP}(v_i; n = m, F)$ when all bidders are admitted for any predictor $F$.
\end{enumerate}
\end{proposition}

The condition in \Cref{thm:SSM_firstprice}(i) ensures that the utility of bidder~$i$ with valuation~$v_i$, when bidding $\sigma^{\FP}(\tilde{v}_i; n, F)$ for some $\tilde{v}_i \in [0,1]$ while all other bidders follow the strategy $\sigma^{\FP}(\cdot; n, F)$ defined in~\Cref{eq:SSM_equilibrium_firstprice}, is increasing in $\tilde{v}_i$ over $[0, v_i]$ and decreasing over $[v_i, 1]$. This unimodality implies that bidder~$i$'s utility is maximized at $\tilde{v}_i = v_i$, thereby confirming that $\sigma^{\FP}(\cdot; n, F)$ is indeed a best response and thus constitutes a symmetric Bayesian-Nash equilibrium. The second part of \Cref{thm:SSM_firstprice}(i), concerning the existence of an SSM equilibrium strategy, follows primarily from \Cref{remark:affiliation}, which establishes that the joint density $g(\cdot; n, F)$ is \emph{affiliated} in $v \in [0,1]^n$ when all bidders are admitted, or under the null or perfect predictor. As shown by \citet{milgrom_1982_auctiontheory_competitive_bidding}, any private-value first-price auction with affiliated types admits an SSM equilibrium.

The unique equilibrium strategy, as characterized in \Cref{thm:SSM_firstprice}(ii), is obtained by solving the ordinary differential equation that arises from the first-order condition.
\Cref{thm:SSM_firstprice}(iii) reveals, perhaps unexpectedly, that the equilibrium strategy is \emph{independent} of the number of admitted bidders under a perfect predictor. From \Cref{eq:SSM_equilibrium_firstprice}, it is evident that all effects of prescreening on the equilibrium strategy are captured by $\RHR(x; n, F) = \frac{h(x \mid x; n, F)}{H(x \mid x; n, F)},$ which is the \emph{reverse hazard rate} of the winning probability when all bidders follow an SSM strategy. Intuitively, the reverse hazard rate captures how sensitively the probability of winning responds to an infinitesimal change in the current bid. Under a perfect predictor, the bidder with the highest valuation is always admitted, regardless of the number of bidders selected. As a result, for any bid level $x$, the rate at which the winning probability changes remains constant with respect to the admitted number. Specifically, the reverse hazard rate under an perfect predictor is given by %
\begin{align}
\label{eq:reverse_hazard_rate_perfect_predictor}
\RHR(x; n, \PP) = \frac{(F_1^{m-1}(x))'}{F_1^{m-1}(x)} 
= \frac{(m - 1) \cdot f_1(x)}{F_1(x)},
\end{align}
which is independent of the admitted number $n$. Consequently, the equilibrium strategy remains unchanged with respect to $n$. Additional examples and discussions that connect our results to the existing literature are provided in \Cref{app_subsec:condition_firstprice_SSM}.

\subsubsection{Optimal Prescreening}

We now proceed to characterize the optimal number of admitted bidders in first-price auctions, assuming that an SSM equilibrium strategy exists. Given the total number of potential bidders $m$ and the predictor $F$, the expected revenue from a first-price auction admitting $n$ bidders is
\begin{align}
\label{eq:rev_n_firstprice}
 \R^{\FP}(n;F,m) = \mathbb{E}[\sigmaFP(V_{(1;n,F)};n,F)],
\end{align}
where the equilibrium bidding strategy $\sigmaFP(\cdot;n,F)$ is given by \Cref{eq:SSM_equilibrium_firstprice}, and $V_{(1;n,F)}$ denotes the highest valuation among the $n$ admitted bidders, i.e., $V_{(1;n,F)} \sim G^{\lar}(\cdot;n,F)$. We denote by $\R_\ast^{\FP}(F,m)$ the expected revenue by optimally choosing the number of admitted bidders.

Prescreening in first-price auctions affects not only the distribution of admitted bidders’ valuations, such as $V_{(1;n,F)}$, but also, in contrast to second-price auctions, their equilibrium bidding strategies, i.e., $\sigma^{\mathrm{FP}}(\cdot;n,F)$. While the highest valuation among admitted bidders, $V_{(1;n,F)}$, is guaranteed to increase with the admitted number in the sense of first-order stochastic dominance (see \Cref{lem:FOSD_number}(ii)), the corresponding equilibrium bidding strategy is not necessarily monotonic in $n$. This distinction complicates the analysis of the optimal number of admitted bidders. In what follows, we provide a sufficient condition under which it is optimal to admit all bidders in the first-price auction.

\begin{theorem}[Optimal Prescreening in First-price Auctions]
\label{thm:opt_admittednumber_firstprice}
Suppose that an SSM equilibrium strategy exists in the first-price auction. Then, for any predictor $F$ satisfying the condition
\begin{align}
\label{eq:condition_firstprice}
  \RHR(x;n,F) \leq \frac{(m-1)\cdot f_1(x)}{F_1(x)}, \quad \forall\,x \in [0,1],~\forall\,n \in [2,m],
\end{align}
it is revenue-maximizing to admit all bidders.

\end{theorem}

The proof of \Cref{thm:opt_admittednumber_firstprice} mainly follows from \Cref{thm:SSM_firstprice}, which establishes that the equilibrium strategy depends solely on, and is increasing in, the reverse hazard rate. Condition~\eqref{eq:condition_firstprice} ensures that the reverse hazard rate $\RHR(x; n, F)$ when all bidders are admitted (i.e., $n = m$) is greater than or equal to the reverse hazard rate under any smaller number of admitted bidders. Notably, the right-hand side of condition \eqref{eq:condition_firstprice} is precisely $\RHR(x; n = m, F)$, which is independent of the predictor $F$, since all bidders are admitted, and depends only on the prior distribution $F_1$ and the total number of potential bidders $m$. Combined with the result from \Cref{lem:FOSD_number}(ii), which shows that the distribution of the highest valuation among admitted bidders increases in the sense of first-order stochastic dominance with $n$, this implies that admitting all bidders maximizes expected revenue. Interestingly, the term $\RHR(x; n = m, F)$ for any predictor $F$ coincides with $\RHR(x; n, \PP)$ under the perfect predictor for any $n$, since the distribution of the largest order statistic remains unchanged in both settings.

We now present a sufficient condition, expressed directly in terms of the copula function, that is easier to verify and guarantees that condition~\eqref{eq:condition_firstprice} holds. Moreover, we show that this condition is satisfied by all predictors discussed in \Cref{exm:copulas_increasing_expectation}.

\begin{proposition}
\label{lem:condition_example_firstprice}
A sufficient condition for inequality~\eqref{eq:condition_firstprice} to hold is that the copula function \( C(v_i, s_i) \) satisfies
\begin{align}
\label{eq:sufficient_condition_firstprice}
v_i \cdot \frac{\partial C(v_i, s_i)}{\partial v_i} + 
s_i \cdot \frac{\partial C(v_i, s_i)}{\partial s_i}
- C(v_i, s_i) \geq 0, \quad \forall v_i, s_i \in [0,1].
\end{align}
Moreover, the following statements hold:
\begin{enumerate}
    \item All three copulas presented in \Cref{exm:copulas_increasing_expectation} satisfy condition~\eqref{eq:sufficient_condition_firstprice}.
    \item If two copulas \( C(\cdot,\cdot) \) and \( \tilde{C}(\cdot,\cdot) \) both satisfy~\eqref{eq:sufficient_condition_firstprice}, then any convex combination of these copulas also satisfies~\eqref{eq:sufficient_condition_firstprice}.
\end{enumerate}
\end{proposition}

\Cref{lem:condition_example_firstprice} shows that a broad class of predictors satisfy condition~\eqref{eq:condition_firstprice} in \Cref{thm:opt_admittednumber_firstprice}. For example, the convex-combination argument in \Cref{lem:condition_example_firstprice}(ii) ensures that the hallucinatory predictor defined in~\Cref{eq:predictor_hallucination} satisfies condition~\eqref{eq:sufficient_condition_firstprice} and, consequently, condition~\eqref{eq:condition_firstprice}. Similar reasoning applies to any convex combinations of the AMH and FGM predictors. Furthermore, it is important to highlight that the sufficient condition~\eqref{eq:sufficient_condition_firstprice} is \emph{independent} of the marginal distributions. Consequently, any copula satisfying this condition can be combined with arbitrary marginal distributions $F_1$ and $F_2$ to generate a diverse set of valid predictors.%

Combined with \Cref{thm:opt_admittednumber_firstprice}, \Cref{lem:condition_example_firstprice} implies that admitting all bidders is optimal in a first-price auction under relatively mild conditions. The intuition behind this result lies in the behavior of the reverse hazard rate, which typically increases with the number of admitted bidders. This rate captures the sensitivity of the winning probability to marginal increases in the bid. As the pool of competitors grows, the distribution shifts to the right, i.e., $H(x \mid x; n, F)$ decreases in $n$ for any $x\in [0,1]$, and the reverse hazard rate tends to steepen, meaning that an increase in a bidder's bid can yield a larger improvement in their probability of winning. Indeed, it can be shown that $\mathsf{RHR}(x; n, F)$ scales linearly with $n$ should admitted bidders' valuations be i.i.d. A steeper reverse hazard rate induces a more aggressive equilibrium bidding strategy, as established in \Cref{thm:SSM_firstprice}. Combined with the fact that the distribution of the highest valuation improves, in the sense of first-order stochastic dominance, as more bidders are admitted (\Cref{lem:FOSD_number}(ii)), the joint effect of bidders with higher valuations and more aggressive bidding behavior implies that the expected revenue is generally maximized when all bidders are admitted.

Next, we analyze how prediction accuracy affects expected revenue in first-price auctions.

\begin{proposition}[Impact of Prediction Accuracy]
\label{thm:rev_prediction_firstprice}
Suppose that an SSM equilibrium strategy exists in the first-price auction.
\begin{enumerate}
    \item Fix any number of admitted bidders $n$. If predictor $F$ is more accurate than $\tilde{F}$ in the sense of \Cref{def:accuracy}, and satisfies $\mathsf{RHR}(x; n, F) \geq \mathsf{RHR}(x; n, \tilde{F})$ for all $x \in [0,1]$, then the expected revenue under $F$ is (weakly) greater than that under $\tilde{F}$.
    
    \item Under a perfect predictor, admitting \emph{any} number of bidders yields the same expected revenue. %
    Moreover, for any predictor $F$ satisfying condition~\eqref{eq:condition_firstprice}, the expected revenue under the perfect predictor is (weakly) greater than that under $F$, for any number of admitted bidders.
\end{enumerate}
\end{proposition}

Analogous to \Cref{thm:opt_admittednumber_secondprice}(ii), \Cref{thm:rev_prediction_firstprice} demonstrates that, under mild conditions, the revenue loss from admitting fewer bidders can be mitigated by employing a more accurate predictor in the first-price auction as well. In particular, \Cref{thm:rev_prediction_firstprice}(i) compares two arbitrary predictors and establishes that if one predictor is more accurate and its associated reverse hazard rate dominates that of the other, then the expected revenue is (weakly) higher under the more accurate predictor. Numerical analysis suggests that the condition $\mathsf{RHR}(x;n,F) \geq \mathsf{RHR}(x;n,\tilde{F})$ tends to hold when predictor $F$ is more accurate than $\tilde{F}$. This observation may be attributed to the fact that a more accurate predictor increases the likelihood of admitting bidders with high valuations. Consequently, the distribution of the highest valuation shifts to the right, making marginal increases in bids more impactful in improving the probability of winning, thereby resulting in a steeper reverse hazard rate. The relationship between predictor accuracy and reverse hazard rates is analytically established in the special case of $m = 3$, as summarized in \Cref{exm:hazard_prediction} below; however, a general proof for arbitrary $m$ remains elusive.

\begin{example}\label{exm:hazard_prediction}
For $m = 3$, within the predictor space composed of hallucinatory predictors and FGM predictors,
higher prediction accuracy implies a higher reverse hazard rate for any admitted number.

\end{example}
\Cref{thm:rev_prediction_firstprice}(ii) shows that the perfect predictor dominates ``any'' other predictor, and there is no revenue loss from admitting fewer bidders when a perfect predictor is employed. This result is driven by two key factors. First, with a perfect predictor, the bidder with the highest valuation is always admitted, ensuring that the maximum valuation remains unchanged regardless of the number of admitted bidders. Second, as established in \Cref{thm:SSM_firstprice}(iii), the equilibrium strategy under a perfect predictor is also independent of the number of admitted bidders, owing to the invariance of the reverse hazard rate.

\begin{figure}[ht]
    \centering
    \resizebox{0.6\textwidth}{!}{%
        \input{fig/firstprice/rev_n.tikz}%
    }
    \caption{Revenue vs. Admitted Number in First-price Auctions under Hallucinatory Predictors}
    \label{fig:firstprice_rev_n}
    \vspace{0.5em}    
      Note. \textit{We consider the setting with $m = 7$ and a power-law prior distribution $F_1(x) = x^{0.2}$.}
\end{figure}

We present numerical results for the expected revenue as a function of the number of admitted bidders in the first-price auction in \Cref{fig:firstprice_rev_n}. For illustration, we focus on the hallucinatory predictor, as it includes the full spectrum of prediction accuracy,
from the null predictor to the perfect predictor, as the parameter $\gamma$ varies from 0 to 1. Other predictors yield similar patterns. The numerical results in \Cref{fig:firstprice_rev_n} further corroborate our analytical findings. First, under the perfect predictor (i.e., when \(\gamma=1\)), revenue remains constant regardless of the number of admitted bidders. Second, for any given predictor, revenue is maximized when all bidders are admitted, consistent with \Cref{thm:opt_admittednumber_firstprice}, and it is (weakly) increasing in the admitted number. Third, for a fixed number of admitted bidders, revenue is (weakly) increasing in prediction accuracy---measured by the parameter \(\gamma\) under the hallucinatory predictor---and is bounded above by the revenue under the perfect predictor, in line with \Cref{thm:rev_prediction_firstprice}. These results illustrate that revenue losses from admitting fewer bidders can be mitigated, or even fully offset, by employing more accurate predictors.

%% file: fig/firstprice/rev_n.tikz
\begin{tikzpicture}
  \begin{axis}[
    xlabel={$n$},
    ylabel={Revenue},
    xmin=2, xmax=7,
    ymin=0, ymax=0.35,
    xtick={2,3,4,5,6,7},
    legend style={
      at={(1.02,0.5)},
      anchor=west,
      draw=none,
      /tikz/every even column/.append style={column sep=1em}
    },
    grid=none,
    width=9cm,
    height=7cm
  ]

    \addplot[
      color=black,
      mark=*,
      solid,
      thick
    ] table[
      x index=0,
      y index=1,
      col sep=space
    ] {fig/firstprice/data/revenue_firstprice_different_n_m_7.txt};
    \addlegendentry{$\gamma=1.00$}

    \addplot[
      color=blue,
      mark=square*,
      dashed,
      thick
    ] table[
      x index=0,
      y index=2,
      col sep=space
    ] {fig/firstprice/data/revenue_firstprice_different_n_m_7.txt};
    \addlegendentry{$\gamma=0.85$}

    \addplot[
      color=red,
      mark=triangle*,
      mark size=4pt,
      dashdotted,
      thick
    ] table[
      x index=0,
      y index=3,
      col sep=space
    ] {fig/firstprice/data/revenue_firstprice_different_n_m_7.txt};
    \addlegendentry{$\gamma=0.78$}

    \addplot[
      color=magenta,
      mark=diamond*,
      mark options={fill=magenta},
      solid,
      thick
    ] table[
      x index=0,
      y index=4,
      col sep=space
    ] {fig/firstprice/data/revenue_firstprice_different_n_m_7.txt};
    \addlegendentry{$\gamma=0.30$}

    \addplot[
      color=gray,
      mark=otimes*,
      solid,
      thick
    ] table[
      x index=0,
      y index=5,
      col sep=space
    ] {fig/firstprice/data/revenue_firstprice_different_n_m_7.txt};
    \addlegendentry{$\gamma=0.00$}

  \end{axis}
\end{tikzpicture}

%% file: sections/4-allpay.tex
\section{All-pay Auctions}\label{sec:all-pay_auctions}
We now proceed to all-pay auctions, in which all participating bidders simultaneously submit bids; the highest bidder wins the item, but all participants pay their bids regardless of the outcome. In what follows, we first characterize the equilibrium strategy in \Cref{subsec:equilibrium_allpay} for a given number of admitted bidders. We then analyze the seller’s optimal choice of admitted bidders in \Cref{subsec:optimal_number_allpay}.

\subsection{Equilibrium Analysis}
\label{subsec:equilibrium_allpay}

The equilibrium definition for all-pay auctions is almost identical to that in \Cref{def:equilibrium_1stprice} for first-price auctions, except that \Cref{eq:def_equilibrium_firstprice} is replaced with
\begin{align}
\label{eq:utility_allpay}
     \Pr\{\sigma(V_j)<b_i,\forall j\in \I_{-i}\} \cdot v_i - b_i.
\end{align}
The difference arises because, in all-pay auctions, all bidders must pay their bids regardless of whether they win the item. Based on bidder $i$'s posterior belief described in \Cref{eq:beta_posterior}, \Cref{eq:utility_allpay} can be rewritten as follows, when all other bidders adopt an SSM strategy $\sigma(\cdot)$:
\begin{align}
\label{eq:utility_inflatedtype_allpay}
\left( \int_{[0,\,\sigma^{-1}(b_i)]^{n-1}} 
    \psi(v; n, F)\,
    \prod_{j \in \I_{-i}} f(v_j)\, dv_{-i} \right)
    \cdot 
    \underbrace{\kappa(v_i; n, F)\, v_i}_{\text{``Inflated'' type}}
    \;-\; b_i.
\end{align}
It is as if bidder~$i$ has an effective type given by $\theta(v_i; n, F) := \kappa(v_i; n, F)\, v_i$.
We refer to this as bidder~$i$'s \textit{inflated type}, %
since $\kappa(v_i; n, F) \geq \C^{m-1}_{n-1} \geq 1$, as established in \Cref{lem:kappa_vi}. We next provide sufficient conditions---one of which is based on the inflated types of bidders---that guarantee the existence of an SSM equilibrium, and offer an explicit characterization when such an equilibrium exists.

\begin{proposition}[Equilibrium Strategy]
\label{thm:strictlymonotone_equilibrium_allpay}
In all-pay auctions, for any admitted number $n\in [2,m]$, 
\begin{enumerate}[(i)]
    \item A symmetric and strictly monotone (SSM) equilibrium exists if, for any given $v_i \in [0,1]$, the expression $v_i\, h(\tilde{v}_i \mid v_i; n, F) - \tilde{v}_i\, h(\tilde{v}_i \mid \tilde{v}_i; n, F)$ is non-negative for all $\tilde{v}_i \in [0, v_i]$ and non-positive for all $\tilde{v}_i \in [v_i, 1]$.

    \item For any predictor satisfying $\PRD(S_i\mid V_i)$,\footnote{As discussed following \Cref{lem:kappa_vi}, this condition holds for all cases presented in \Cref{exm:copulas_increasing_expectation}, including any convex combinations thereof. Moreover, for a symmetric predictor $F(\cdot,\cdot)$, the property $\PRD(S_i \mid V_i)$ is equivalent to $\PRD(V_i \mid S_i)$, which in turn implies \Cref{assum:increasing_conditional_expectation}.} a sufficient condition for part~(i) is that the inflated type $\theta(v_i; n, F)$ is non-decreasing in $v_i$ for all $v_i \in [0,1]$. This condition is also necessary for the existence of an SSM equilibrium under the perfect predictor.

    \item If SSM equilibrium strategies exist, they are unique and given by the following expression:
    \begin{align}
    \label{eq:SSM_equilibrium_all_pay}
        \sigma^{\mathsf{AP}}(v_i;n,F)
  =\int_0^{v_i} x h(x\mid x;n,F)dx,~ \forall v_i\in [0,1].
    \end{align}
\end{enumerate}
\end{proposition}

The expression $v_i\, h(\tilde{v}_i \mid v_i; n, F) - \tilde{v}_i\, h(\tilde{v}_i \mid \tilde{v}_i; n, F)$ in \Cref{thm:strictlymonotone_equilibrium_allpay}(i) is, in fact, the derivative of the expected utility of bidder~$i$ with valuation $v_i$ with respect to $\tilde{v}_i$ when she bids $\sigma^{\mathsf{AP}}(\tilde{v}_i; n, F)$, while all other bidders follow the strategy $\sigma^{\mathsf{AP}}(\cdot; n, F)$. Condition~(i) ensures that bidder~$i$'s utility is maximized when she bids $\sigma^{\mathsf{AP}}(v_i; n, F)$, thereby guaranteeing the existence of an SSM equilibrium. Note that \Cref{thm:strictlymonotone_equilibrium_allpay}(i) extends beyond our particular prescreening setting: it applies to any private-valuation all-pay auction with arbitrarily correlated types. \Cref{thm:strictlymonotone_equilibrium_allpay}(ii) provides a more easily verifiable sufficient condition based on the inflated type of bidders, $\theta(v_i; n, F)$. Combined with \Cref{lem:kappa_vi}(i), it follows that $\kappa(v_i; n, F)$ must decrease \emph{sublinearly} in $v_i$ for the inflated type $\theta(v_i; n, F) = v_i \cdot \kappa(v_i; n, F)$ to be (weakly) increasing in $v_i$. The uniqueness result in \Cref{thm:strictlymonotone_equilibrium_allpay}(iii) follows from the fact that the ordinary differential equation, derived from the first-order condition, admits a unique solution. The conditions in \Cref{thm:strictlymonotone_equilibrium_allpay}(i) and (ii) ensure that this solution indeed corresponds to a maximizer.

We now seek to further understand the existence and structure of SSM equilibrium strategies by examining two special cases. The first is when $n = m$, that is, the seller admits all potential bidders. In this case, we have $h(\tilde{v}_i\mid  v_i; n = m, F) = m F_1^{m-1}(\tilde{v}_i) f_1(\tilde{v}_i)$, and it is straightforward to verify that the condition in \Cref{thm:strictlymonotone_equilibrium_allpay}(i) is always satisfied. Therefore, an SSM equilibrium strategy always exists when $n = m$. In this case, the strategy in \Cref{eq:SSM_equilibrium_all_pay} simplifies to the following expression, which holds for any predictor:
\begin{align}
\label{eq:equilibrium_standard_allpay}
\sigma^{\mathsf{AP}}(v_i; n = m, F) = \int_0^{v_i} x \, dF_1^{m-1}(x), \quad \forall v_i \in [0,1].
\end{align}
This result is consistent with the standard literature on all-pay auctions under independent private values \citep{milgrom_2004_putting_auctiontheory_to_work}.

Secondly, we focus on the hallucinatory predictor \eqref{eq:predictor_hallucination} and the power-law prior distribution, i.e., $F_1(x) = x^c$ for $c > 0$, to examine how the condition in \Cref{thm:strictlymonotone_equilibrium_allpay}(ii) varies with different levels of prediction accuracy and different shapes of the power-law function.

\begin{proposition}\label{prop:larger_gamma_stricter_condition}
Given the prior distribution $F_1(x) = x^c$ for $c > 0$, the following statements hold:
\begin{enumerate}
    \item If the inflated type $\theta(v_i; n, \HP)$ is non-decreasing in $v_i \in [0,1]$ for a given $\gamma \in [0,1]$, then $\theta(v_i; n, \mathsf{HP}(\gamma^\prime))$ is also non-decreasing in $v_i \in [0,1]$ for any $\gamma^\prime \leq \gamma$.
    
    \item If the inflated type $\theta(v_i; n, \HP)$ is non-decreasing in $v_i \in [0,1]$ for a given $c > 0$, then it is also non-decreasing in $v_i \in [0,1]$ for any $c' \in [0, c]$.
    
\end{enumerate}
\end{proposition}

\Cref{prop:larger_gamma_stricter_condition} shows that the condition requiring a bidder’s inflated type to be non-decreasing in their private valuation $v_i$, and thus the existence of an SSM equilibrium, is more likely to be violated as the prediction accuracy increases or when the prior belief becomes ``stronger,'' in the sense that a prior distribution with a higher $c$ first-order stochastically dominates one with a lower $c$. The rationale behind both effects is similar. As prediction accuracy increases, signals become more accurate, leading each admitted bidder to expect their opponents to have higher valuations on average. This discourages aggressive bidding, since in all-pay auctions bidders pay regardless of the outcome.
Similarly, a more competitive environment, reflected by a higher value of $c$, can also deter high-valuation bidders from bidding aggressively.

\subsection{Optimal Prescreening}
\label{subsec:optimal_number_allpay}
The preceding analysis sets the stage for deriving the optimal number of admitted bidders.
Given the predictor $F$, and the total number of potential bidders~$m$, the expected revenue with $n$ admitted bidders is given by
\begin{align}
\label{eq:rev_n_allpay}
 \R^{\AP}(n;F,m) = n\cdot \mathbb{E}_{V_i\sim G^{\mar}(\cdot ; n,F)}[\sigmaAP(V_i;n,F)],
\end{align}
where $G^{\mar}(x;n,F)$ is the CDF of the marginal distribution of valuations among admitted bidders. Thus, the optimal expected revenue is given by $\R_\ast^{\AP}(F,m) := \max_{n\in[2,m]} \R^{\AP}(n;F,m).$ Note that, by \Cref{lem:joint_dis_g} about the equivalent game, the expectations are taken with respect to distributions derived from the joint distribution $g(\cdot; n, F)$ in~\Cref{eq:joint_dist_g}, rather than from the prior distribution $F_1$. 
By using Fubini's theorem, we can simplify the expected revenue in \Cref{eq:rev_n_allpay} as follows:
\begin{align}
\label{eq:allpay_revenue}
   \R^{\AP}(n;F,m) =n\cdot \int_0^1 x\, h(x\mid x;n,F)\cdot \left(1-G^{\mar}(x;n,F)\right)dx.
\end{align}
The seller's expected revenue is determined by three interrelated factors: (1) the number of bidders making payments, $n$ (as all participants pay their bids in all-pay auctions);  (2) the marginal distribution $ G^{\mar}(x; n, F)$, i.e., the average admitted bidders' valuations. 
According to \Cref{lem:FOSD_number}(i), $1 - G^{\mar}(x; n, F)$ decreases with $n$, thus working in the \emph{opposite} direction of the first factor; (3) the term $h(x \mid x; n, F)$, which captures the effect of bidders' equilibrium strategies. For a general predictor $F$, $h(x \mid x; n, F)$ is \emph{not} monotone in $n$. In fact, even under the hallucinatory predictor in \Cref{eq:predictor_hallucination} with arbitrary $\gamma$, no monotonicity of $h(x \mid x; n, F)$ can be guaranteed.
Thus, the optimal number of admitted bidders is determined by balancing these competing effects, making the solution more nuanced. Fortunately, we demonstrate that determining the optimal number of admitted bidders is computationally tractable, particularly for the hallucinatory predictor.

\begin{remark}[Computational Tractability]
\label{remark:generalsetting_allpay}
Determining the optimal admitted number requires evaluating \Cref{eq:allpay_revenue} for different values of $n$. The main challenge comes from the function $h(x \mid x; n, F)$, which involves an $(n-1)$-dimensional integral nested within another $(n-1)$-dimensional integral (due to the normalization term). Similarly, $G^{\mar}(x; n, F)$ also requires computing an $(n-1)$-dimensional integral. We show that each of these three $(n-1)$-dimensional integrals can be reduced to a one-dimensional integral. Moreover, for the hallucinatory predictor characterized in \Cref{eq:predictor_hallucination}, we derive closed-form expressions for these functions (see \Cref{app_subsec:beliefs}). Consequently, in this case, evaluating \Cref{eq:allpay_revenue} requires computing only \emph{one-dimensional} integrals, which are independent of both the total number and the admitted number.
\end{remark}

For the null and perfect predictors, we explicitly characterize the optimal prescreening policy.

\begin{theorem}[Optimal Prescreening in All-Pay Auctions]
\label{thm:opt_admittednumber_allpay}
Suppose an SSM equilibrium strategy exists for all $n \in [2, m]$ in all-pay auctions. %
\begin{enumerate}[(i)]
\item Under the perfect predictor, the unique revenue-maximizing admitted number is $n^\ast=2$.
\item Under the null predictor, admitting all bidders is optimal.

\end{enumerate}
\end{theorem}

\Cref{thm:opt_admittednumber_allpay}(i) shows that admitting only two bidders is optimal for maximizing expected revenue when the predictor is perfect. We find that, as a direct consequence of the proof of \Cref{thm:strictlymonotone_equilibrium_allpay}, the SSM equilibrium strategy given in \Cref{eq:SSM_equilibrium_all_pay} simplifies to the following under the perfect predictor:
  \begin{align}
\label{eq:SSM_equilibrium_gamma_1_all_pay}
        \sigma^{\mathsf{AP}}(v_i; n, \PP)
        = \int_0^{v_i} \frac{x}{J(F_1(x), n, m)}\, dF_1^{m-1}(x),
    \end{align}
where $J(F_1(x), n, m) = \C^{m-1}_{n-1} / \kappa(x; n, \PP)$, which is shown to be non-decreasing in the admitted number $n \in [2, m]$ in \Cref{lem:J_increasing}. As a result, the equilibrium strategy $ \sigma^{\mathsf{AP}}(v_i; n, \PP)$ is (weakly) decreasing in the admitted number $n$ for any valuation $v_i\in[0,1]$. This is because, under a perfect predictor, signals are fully informative, ensuring that bidders with the highest valuations are guaranteed to be admitted. In this case, as the number of admitted bidders increases, each bidder’s perceived likelihood of having the highest valuation decreases.\footnote{Specifically, for an admitted bidder with valuation $v_i$, this probability is given by $\frac{F_1^{m-1}(v_i)}{J(F_1(v_i), n, m)}$.} Intuitively, as $n$ increases, an admitted bidder is more likely to attribute her admission to the larger pool rather than to possessing the top valuation, prompting further bid shading in equilibrium.
We emphasize that this is a \emph{sample-wise} monotonicity in the sense that it holds for every valuation \(v_i \in [0,1]\), a property not commonly observed in all-pay auctions.
In fact, even in the case of i.i.d.\ valuations under the null predictor, the equilibrium strategy $\sigma^{\mathsf{AP}}(v_i; n, \NP)$ is \emph{not} monotonic in the number of participants $n$ for general values of $v_i$. 

Based on \Cref{eq:SSM_equilibrium_gamma_1_all_pay}, the expected revenue under the perfect predictor can be rewritten as:
\begin{align*}
 \R^{\AP}(n;\PP, m) 
 &= n\int_0^1 \int_0^{v_i} \frac{x}{J(F_1(x),n,m)}\, dF_1^{m-1}(x)\, dG^{\mar}(v_i;n,\PP)\\[1mm]
 &= \int_0^1 \frac{n}{J(F_1(x),n,m)} \cdot \left(1-G^{\mar}(x; n,\PP)\right) \cdot x\, dF_1^{m-1}(x),
\end{align*}
where the second equality follows from Fubini's theorem. Similar to the preceding discussions, all effects of prescreening on the expected revenue are captured by the term $\frac{n}{J(F_1(x),n,m)} \cdot \left(1 - G^{\mar}(x; n, \PP)\right)$, where the coefficient $n$ reflects the number of bidders who pay their bids, the term $1 / J(F_1(x), n, m)$ accounts for the impact of prescreening on the equilibrium strategy (as shown in \Cref{eq:SSM_equilibrium_gamma_1_all_pay}), and $G^{\mar}(x; n, \PP)$ captures the impact of the valuation distribution. For a bidder with a given valuation, their equilibrium bid decreases with a larger admitted pool, and the valuations of admitted bidders tend to be lower on average as more bidders are admitted, as established in \Cref{lem:FOSD_number}(i). As $n$ increases, the increase in the number of paying bidders is outweighed by the reduction in individual bid amounts, leading to the optimality of admitting only two bidders.

\Cref{thm:opt_admittednumber_allpay}(ii) shows that admitting all bidders is optimal when the predictor is completely uninformative. In this case, the setting is equivalent to a standard all-pay auction with $n$ bidders holding independent private valuations, and thus an SSM equilibrium strategy always exists. By the revenue equivalence theorem \citep{myerson_1981_optimal_auction}, the expected revenue in an all-pay auction equals that in a second-price auction, which corresponds to the expectation of the second-largest order statistic among $n$ i.i.d.\ random variables. This expectation increases with $n$ by \Cref{lem:FOSD_number}(ii). Hence, in terms of expected revenue, the optimal admitted number under the null predictor is $n^\ast = m$.

\begin{figure}[ht]
    \centering
    \resizebox{0.6\textwidth}{!}{%
        \input{fig/allpay/rev_n.tikz}%
    }
    \caption{Revenue vs. Admitted Number in All-pay Auctions under Hallucinatory Predictors}
    \label{fig:allpay_rev_n}
    \vspace{0.5em}    
      Note. \textit{We consider the setting with $m = 7$ and a power-law prior distribution $F_1(x) = x^{0.2}$. Under these conditions, it can be verified that an SSM equilibrium strategy exists for all admitted numbers $n \in [2, m]$ and for all $\gamma \in [0,1]$ under the hallucinatory predictors defined in~\Cref{eq:predictor_hallucination}.}
\end{figure}

\Cref{fig:allpay_rev_n} presents numerical results on how revenue varies with the admitted number for predictors with different levels of prediction accuracy, based on the hallucinatory predictors defined in~\Cref{eq:predictor_hallucination}. Several key observations are worth highlighting: (i) The revenue increases with the admitted number under the null predictor (i.e., $\gamma = 0$), whereas it decreases with $n$ under the perfect predictor (i.e., $\gamma = 1$). These findings are consistent with the analytical results in \Cref{thm:opt_admittednumber_allpay}. However, for intermediate levels of prediction accuracy, revenue is generally \emph{not} monotonic in the admitted number. Moreover, the pattern can be very complicated. For example, when $\gamma = 0.78$, revenue first increases, then decreases, and subsequently increases again as $n$ grows; similarly, for $\gamma = 0.85$, revenue initially decreases before increasing. This non-monotonicity highlights the difficulty of analytically characterizing the optimal prescreening policy under general predictors. (ii) Extending the analytical result that admitting two bidders is optimal when $\gamma = 1$, it is generally optimal to admit only two bidders when the predictor is sufficiently accurate (as illustrated by the curves for $\gamma \in \{0.85, 1\}$). (iii) For any given admitted number, revenue is (weakly) increasing in $\gamma$, highlighting the value of improved prediction accuracy.
Additional numerical results on the comparative statics of revenue with respect to the prior distribution and the total number of potential bidders are provided in \Cref{app_subsec:numerical}.

%% file: fig/allpay/rev_n.tikz
\begin{tikzpicture}
  \begin{axis}[
    xlabel={$n$},
    ylabel={Revenue},
    xmin=2, xmax=7,
    ymin=0, ymax=0.45,
    xtick={2,3,4,5,6,7},
    legend style={
      at={(1.02,0.5)},
      anchor=west,
      draw=none,
      /tikz/every even column/.append style={column sep=1em}
    },
    grid=none,
    width=9cm,
    height=7cm
  ]

    \addplot[
      color=black,
      mark=*,
      solid,
      thick
    ] table[
      x index=0,
      y index=1,
      col sep=space
    ] {fig/allpay/data/revenue_allpay_different_n_m_7.txt};
    \addlegendentry{$\gamma=1.00$}

    \addplot[
      color=blue,
      mark=square*,
      dashed,
      thick
    ] table[
      x index=0,
      y index=2,
      col sep=space
    ] {fig/allpay/data/revenue_allpay_different_n_m_7.txt};
    \addlegendentry{$\gamma=0.85$}

    \addplot[
      color=red,
      mark=triangle*,
      mark size=4pt,
      dashdotted,
      thick
    ] table[
      x index=0,
      y index=3,
      col sep=space
    ] {fig/allpay/data/revenue_allpay_different_n_m_7.txt};
    \addlegendentry{$\gamma=0.78$}

    \addplot[
      color=magenta,
      mark=diamond*,
      mark options={fill=magenta},
      solid,
      thick
    ] table[
      x index=0,
      y index=4,
      col sep=space
    ] {fig/allpay/data/revenue_allpay_different_n_m_7.txt};
    \addlegendentry{$\gamma=0.30$}

    \addplot[
      color=gray,
      mark=otimes*,
      solid,
      thick
    ] table[
      x index=0,
      y index=5,
      col sep=space
    ] {fig/allpay/data/revenue_allpay_different_n_m_7.txt};
    \addlegendentry{$\gamma=0.00$}

  \end{axis}
\end{tikzpicture}

%% file: sections/6-rev_ranking.tex
\section{Revenue Ranking}
\label{sec:rev_ranking}

In this section, we examine how prescreening affects the revenue ranking of the three auction formats analyzed in this paper.

\begin{theorem}[Revenue Ranking]
\label{thm:rev_ranking}
Suppose an SSM equilibrium strategy exists for both all-pay and first-price auctions.
\begin{enumerate}
    \item Under any predictor, the optimal revenue in all-pay auctions (weakly) dominates that in second-price auctions, i.e., $\R^{\AP}_\ast(F,m) \geq \R^{\SP}_\ast(F,m)$. Furthermore, under any predictor $F$ satisfying condition \eqref{eq:condition_firstprice}, first-price auctions are revenue-equivalent to second-price auctions at optimality, and both are dominated by all-pay auctions. That is,
    \begin{align}
    \R^{\AP}_\ast(F,m) \geq \R^{\FP}_\ast(F,m) =\R^{\SP}_\ast(F,m),  \label{eq:rev_ranking} 
    \end{align}
    where the inequality becomes strict under a perfect predictor.

    \item Moreover, under a perfect predictor,\footnote{Here and in \Cref{cor:Predictions vs. Negotiations}, we further assume that the marginal distribution \( F_1 \) and the number of initial bidders \( m \) satisfy certain dominated convergence condition given in \Cref{eq:dominated_ convergence_condition}.}
for sufficiently large \( m \), the expected revenue from an all-pay auction admitting only two bidders is higher than that of either a first-price or second-price auction that admits all bidders, even when the latter includes one additional bidder. That is,
    \begin{align}
      \R^{\AP}_\ast(\PP,m)
      =
      \R^{\AP}(n=2;\PP,m)
      \geq \R^{\FP}_\ast(\PP,m+1) =\R^{\SP}_\ast(\PP,m+1)=\R^{\FP}(n=m+1;F,m+1).  \nonumber  
    \end{align}
\end{enumerate}

\end{theorem}

Recall that \Cref{lem:condition_example_firstprice} demonstrates that condition \eqref{eq:condition_firstprice} is relatively mild, as it holds for all examples in \Cref{exm:copulas_increasing_expectation}, including any convex combinations of them. The rationale behind the revenue ranking in \Cref{thm:rev_ranking}(i) is as follows. It is optimal to admit all bidders in first-price and second-price auctions, as established in \Cref{thm:opt_admittednumber_secondprice} and \Cref{thm:opt_admittednumber_firstprice}. When all $m$ bidders are admitted, our model reduces to the standard setting with independent private valuations. By the classic revenue equivalence theorem \citep{myerson_1981_optimal_auction}, we have $\R^{\FP}_\ast(F,m) =\R^{\SP}_\ast(F,m)=\R^{\AP}(n=m;F,m).$ In contrast, admitting fewer bidders can increase revenue in all-pay auctions, particularly when the predictor is highly accurate (see \Cref{thm:opt_admittednumber_allpay}). This establishes the stated revenue ranking.

Revenue rankings across auction formats have been derived under various conditions in the literature. With independent private valuations, all three formats yield the same expected revenue \citep{myerson_1981_optimal_auction}. Under affiliation, first-price auctions are (weakly) dominated by both second-price \citep{milgrom_1982_auctiontheory_competitive_bidding} and all-pay auctions \citep{krishna_1997_all_pay_affiliation}, though no definitive order exists between the latter two. When the number of bidders is uncertain, however, \citet{pekevc_2008_revenueranking_random_number} show that first-price auctions can generate strictly higher revenue than second-price auctions under affiliation, especially when the uncertainty is large.
Our result in \Cref{thm:rev_ranking} differs from the existing literature in two key ways: (i) We do not compare auction formats under a fixed correlation structure. Instead, correlation arises endogenously from the prescreening process. Consequently, deriving the revenue ranking requires comparing formats across different correlation structures: from independent private valuations (when all bidders are admitted), to affiliation (under a perfect predictor), and to more general forms of correlation. (ii) The mechanism driving the revenue ranking in our setting differs fundamentally from prior work, where rankings are primarily attributed to differences in valuation correlation. Here, prescreening naturally induces positive correlation, potentially affiliation, among admitted bidders. However, in first-price and second-price auctions, this induced correlation does not offset the loss of competition from excluding bidders, except when the seller has perfect information about bidders’ valuations. By contrast, in all-pay auctions, the intensified competition among admitted bidders substantially boosts revenue, especially when the seller can pinpoint high-valuation bidders and signal to them that they are among the few who are admitted.

Under a perfect predictor and when the total number of bidders $m$ is sufficiently large, \Cref{thm:rev_ranking}(ii) shows that the revenue from an all-pay auction with the optimal admitted number is even greater than the optimal revenue from first-price and second-price auctions with one additional potential bidder. Combining this with the auctions-versus-negotiations result of \citet{bulow_1994_auctions_negotiations}---which states that, in the absence of prescreening, the revenue from a first-price auction with $m+1$ bidders exceeds that of the optimal auction format with $m$ bidders when the prior $F_1$ is regular---we obtain the following ``predictions-versus-negotiations'' result.

\begin{corollary}[\textbf{Predictions vs. Negotiations}]\label{cor:Predictions vs. Negotiations}
Suppose that an SSM equilibrium strategy exists in all-pay auctions, and the prior $F_1$ is regular; that is, the virtual value $v_i - \frac{1-F_1(v_i)}{f_1(v_i)}$ is non-decreasing in $v_i \in [0,1]$. Under a perfect predictor, and when the total number of bidders $m$ is sufficiently large, the revenue from an all-pay auction admitting only two bidders exceeds the revenue from any auction format that admits all bidders.
\end{corollary}

\begin{figure}[ht]
    \centering
    \resizebox{0.9\textwidth}{!}{%
        \input{fig/revenue_comparison/gamma=1/rev_comparison.tikz}%
    }
    \caption{Revenue Comparison under the Hallucinatory Predictor $\HP$}
    \label{fig:rev_comparison}
    \vspace{0.5em}    
      Note. \textit{The curves for second-price auctions are omitted, as they coincide with those of first-price auctions.}
\end{figure}

Note that the regularity condition on the prior distribution $F_1$ required for \Cref{cor:Predictions vs. Negotiations} is not needed for \Cref{thm:rev_ranking}(ii). This result holds as long as $m$ is sufficiently large and an SSM equilibrium strategy exists. Numerical examples comparing optimal all-pay auctions with $m$ potential bidders to optimal first-price auctions with $m$ and $m+1$ potential bidders are presented in \Cref{fig:rev_comparison}. The left panel illustrates the comparison under a perfect predictor, while the right panel considers a hallucinatory predictor with $\gamma=0.9$. The results suggest that \Cref{thm:rev_ranking}(ii) generally holds when the predictor is sufficiently accurate, and the total number of potential bidders $m$ does \textit{not} need to be very large for the conclusion to apply.

\Cref{thm:rev_ranking} presents the revenue rankings across auction formats at optimality. However, as discussed in \Cref{sec:intro}, sellers often face external resource constraints that limit the number of bidders they can admit. This poses a particular challenge for first-price and second-price auctions, which benefit from broad competition. Taking this into account, we now compare revenues across the three auction formats for a \emph{fixed} number of admitted bidders, which may not be optimal. In particular, let \(H_2(x \mid v_i; n, F) := \frac{\partial}{\partial v_i} H(x \mid v_i; n, F)\) denote the partial derivative of \(H(x \mid v_i; n, F)\) with respect to \(v_i\). We summarize the revenue comparison results as follows.

\begin{proposition}[Revenue Comparison with a Fixed Number of Admitted Bidders]
\label{prop:rev_comparison_given_admitted_number}
Given \emph{any} admitted number $n \in [2,m]$:
\begin{enumerate}
  \item Under a perfect predictor, all-pay auctions yield higher expected revenue than both
  second-price and first-price auctions, while the latter two yield identical expected revenue.
  \item For any predictor, if 
  \begin{align}
\label{eq:cond_second_first_comparison}
 \frac{\partial}{\partial x}\left(\frac{H(x \mid v_i; n, F)}{H_2(x \mid v_i; n, F)}\right) \leq 0, \quad \forall\,x \in [0,v_i],
  \end{align}
  then second-price auctions yield higher revenue than first-price auctions; otherwise, if the inequality in condition \eqref{eq:cond_second_first_comparison} is reversed, first-price auctions yield higher revenue than second-price auctions.
  \item For any predictor, if $H_2(v_i \mid v_i; n, F) \leq 0$, then all-pay auctions yield higher revenue than first-price auctions.
\end{enumerate}
\end{proposition}

\Cref{prop:rev_comparison_given_admitted_number}(i) follows directly from \Cref{thm:opt_admittednumber_secondprice}, \Cref{thm:rev_prediction_firstprice}, and \Cref{thm:opt_admittednumber_allpay}. Parts (ii) and (iii) mainly follow from Chapter 7 in \cite{krishna2009auction}, and provide sufficient conditions for establishing the revenue ranking in pairwise comparisons of auction formats. Notably, both condition \eqref{eq:cond_second_first_comparison} and the condition in (iii) are satisfied when bidders’ valuations are affiliated. In this case, our results recover the findings in the literature: second-price auctions weakly dominate first-price auctions in terms of expected revenue \citep{milgrom_1982_auctiontheory_competitive_bidding}, and all-pay auctions similarly outperform first-price auctions \citep{krishna_1997_all_pay_affiliation}. In contrast, in our setting, bidders’ valuations are generally not affiliated (see \Cref{remark:affiliation}), and hence first-price auctions may outperform second-price auctions.

\begin{figure}[ht]
    \centering
        \begin{subfigure}[b]{0.24\textwidth}
        \centering
        \resizebox{\textwidth}{!}{%
            \input{fig/revenue_comparison/different_mechanism/gamma_0.3.tikz}%
        }
        \caption{$\gamma=0.3$}
     
    \end{subfigure}
    \hfill 
   \begin{subfigure}[b]{0.24\textwidth}
        \centering
        \resizebox{\textwidth}{!}{%
            \input{fig/revenue_comparison/different_mechanism/gamma_0.78.tikz}%
        }
        \caption{$\gamma=0.78$}
     
    \end{subfigure}
    \hfill
  \begin{subfigure}[b]{0.24\textwidth}
        \centering
        \resizebox{\textwidth}{!}{%
            \input{fig/revenue_comparison/different_mechanism/gamma_0.85.tikz}%
        }
        \caption{$\gamma=0.85$}
     
    \end{subfigure}
    \hfill
     \begin{subfigure}[b]{0.24\textwidth}
        \centering
        \resizebox{\textwidth}{!}{%
            \input{fig/revenue_comparison/different_mechanism/gamma_1.tikz}%
        }
        \caption{$\gamma=1$}
     
    \end{subfigure}
      \caption{Revenue vs. Admitted Number under the Hallucinatory Predictor $\HP$}
    \label{fig:rev_comparison_given_admitted_number}
    \medskip    
      Note. \textit{We consider the setting with $m=7$ and a power-law prior distribution $F_1(x)=x^{0.2}$. The parameter $\gamma\in[0,1]$ is specified below each subfigure.
      }
\end{figure}

Our numerical results in \Cref{fig:rev_comparison_given_admitted_number} show that, when prediction accuracy is low and the admitted number is small (e.g., $n=2$ in \Cref{fig:rev_comparison_given_admitted_number}(a)), second-price auctions may yield higher revenue than both first-price and all-pay auctions. As either the admitted number increases or accuracy improves, first-price auctions often (weakly) outperform second-price auctions.
When prediction accuracy is high, all-pay auctions generally (weakly) dominate the other two formats. Overall, these patterns suggest that, with a less accurate predictor and few admission slots, a second-price auction may be preferable, whereas with high accuracy---even for a small number of slots---an all-pay auction is advantageous.

%% file: fig/revenue_comparison/gamma=1/rev_comparison.tikz
\begin{tikzpicture}
  \begin{axis}[
      xlabel={$m$},
      ylabel={Revenue},
      title={Prior Distribution $F_1(x)=x^{0.10}$},
      xlabel style={font=\Large},
      ylabel style={font=\Large},
      title style={font=\Large},
      tick label style={font=\large},
      legend pos=north west,
      width=10cm,
      height=8cm,
  ]
    \addplot[
        blue, 
        thick, 
        mark=*,
        mark options={fill=blue}
    ]
    table [x index=0, y index=1, col sep=tab] {fig/revenue_comparison/gamma=1/data/revenue_data_c_0.10.txt};
    \addlegendentry{$\R_\ast^{\AP}(\PP,m)$};

    \addplot[
        black,
        thick,
        dashed,
        mark=diamond*,
        mark options={fill=black}
    ]
    table [x index=0, y index=2, col sep=tab] {fig/revenue_comparison/gamma=1/data/revenue_data_c_0.10.txt};
    \addlegendentry{$\R^{\FP}(n=m;F,m)$};

    \addplot[
        red, 
        thick, 
        mark=square*,
        mark options={fill=red}
    ]
    table [x index=0, y index=3, col sep=tab] {fig/revenue_comparison/gamma=1/data/revenue_data_c_0.10.txt};
    \addlegendentry{$\R^{\FP}(n=m+1;F,m+1)$};
    
  \end{axis}
\end{tikzpicture}

\hspace{0.5cm}

\begin{tikzpicture}
  \begin{axis}[
      xlabel={$m$},
      ylabel={Revenue},
      title={Prior Distribution $F_1(x)=x^{0.20}$},
      xlabel style={font=\Large},
      ylabel style={font=\Large},
      title style={font=\Large},
      tick label style={font=\large},
      legend pos=north west,
      width=10cm,
      height=8cm,
  ]
    \addplot[
        blue, 
        thick, 
        mark=*,
        mark options={fill=blue}
    ]
    table [x index=0, y index=1, col sep=tab] {fig/revenue_comparison/general_gamma/data/rev_comparison_gamma_0.90_c_0.20.txt};
    \addlegendentry{$\R_\ast^{\AP}(\HP,m)$ with $\gamma=0.9$};

    \addplot[
        black,
        thick,
        dashed,
        mark=diamond*,
        mark options={fill=black}
    ]
    table [x index=0, y index=2, col sep=tab] {fig/revenue_comparison/general_gamma/data/rev_comparison_gamma_0.90_c_0.20.txt};
    \addlegendentry{$\R^{\FP}(n=m;F,m)$};

    \addplot[
        red, 
        thick, 
        mark=square*,
        mark options={fill=red}
    ]
    table [x index=0, y index=3, col sep=tab] {fig/revenue_comparison/general_gamma/data/rev_comparison_gamma_0.90_c_0.20.txt};
    \addlegendentry{$\R^{\FP}(n=m+1;F,m+1)$};
    
  \end{axis}
\end{tikzpicture}

%% file: fig/revenue_comparison/different_mechanism/gamma_0.3.tikz
\begin{tikzpicture}
  \begin{axis}[
    xlabel={$n$},
    ylabel={Revenue},
    label style={font=\huge},
    tick label style={font=\huge},
    xmin=2, xmax=7,
    ymin=0.1, ymax=0.32,
    xtick={2,3,4,5,6,7},
    legend style={
      at={(0.98,0.02)},      %
      anchor=south east,
      draw=none,
      fill=white,
      font=\huge,           %
      /tikz/every even column/.append style={column sep=1em}
    },
    grid=none,
    width=9cm,
    height=7cm
  ]

    \addplot[
      color=blue,
      mark=square*,
      dashed,
      thick
    ] table[
      x index=0,
      y index=2,
      col sep=space
    ] {fig/revenue_comparison/different_mechanism/data/gamma_0.3.txt};
    \addlegendentry{all-pay}
    \addplot[
      color=black,
      mark=*,
      solid,
      thick
    ] table[
      x index=0,
      y index=1,
      col sep=space
    ] {fig/revenue_comparison/different_mechanism/data/gamma_0.3.txt};
    \addlegendentry{first-price}

    \addplot[
      color=red,
      mark=triangle*,
      mark size=4pt,
      dashdotted,
      thick
    ] table[
      x index=0,
      y index=3,
      col sep=space
    ] {fig/revenue_comparison/different_mechanism/data/gamma_0.3.txt};
    \addlegendentry{second-price}
  \end{axis}
\end{tikzpicture}

%% file: fig/revenue_comparison/different_mechanism/gamma_0.78.tikz
\begin{tikzpicture}
  \begin{axis}[
    xlabel={$n$},
    ylabel={Revenue},
    label style={font=\huge},
    tick label style={font=\huge},
    xmin=2, xmax=7,
    ymin=0.2, ymax=0.32,
    xtick={2,3,4,5,6,7},
    legend style={
      at={(0.98,0.02)},   %
      anchor=south east,
      draw=none,
      fill=white,
      font=\huge,
      /tikz/every even column/.append style={column sep=1em}
    },
    grid=none,
    width=9cm,
    height=7cm
  ]

    \addplot[
      color=blue,
      mark=square*,
      dashed,
      thick
    ] table[
      x index=0,
      y index=2,
      col sep=space
    ] {fig/revenue_comparison/different_mechanism/data/gamma_0.78.txt};
    \addlegendentry{all-pay}
    
    \addplot[
      color=black,
      mark=*,
      solid,
      thick
    ] table[
      x index=0,
      y index=1,
      col sep=space
    ] {fig/revenue_comparison/different_mechanism/data/gamma_0.78.txt};
    \addlegendentry{first-price}

    \addplot[
      color=red,
      mark=triangle*,
      mark size=4pt,
      dashdotted,
      thick
    ] table[
      x index=0,
      y index=3,
      col sep=space
    ] {fig/revenue_comparison/different_mechanism/data/gamma_0.78.txt};
    \addlegendentry{second-price}
  \end{axis}
\end{tikzpicture}

%% file: fig/revenue_comparison/different_mechanism/gamma_0.85.tikz
\begin{tikzpicture}
  \begin{axis}[
    xlabel={$n$},
    ylabel={Revenue},
    label style={font=\huge},
    tick label style={font=\huge},
    xmin=2, xmax=7,
    ymin=0.2, ymax=0.33,
    xtick={2,3,4,5,6,7},
    legend style={
      at={(0.98,0.02)},   %
      anchor=south east,
      draw=none,
      fill=white,
      font=\huge,
      /tikz/every even column/.append style={column sep=1em}
    },
    grid=none,
    width=9cm,
    height=7cm
  ]

    \addplot[
      color=blue,
      mark=square*,
      dashed,
      thick
    ] table[
      x index=0,
      y index=2,
      col sep=space
    ] {fig/revenue_comparison/different_mechanism/data/gamma_0.85.txt};
    \addlegendentry{all-pay}

    \addplot[
      color=black,
      mark=*,
      solid,
      thick
    ] table[
      x index=0,
      y index=1,
      col sep=space
    ] {fig/revenue_comparison/different_mechanism/data/gamma_0.85.txt};
    \addlegendentry{first-price}

    \addplot[
      color=red,
      mark=triangle*,
      mark size=4pt,
      dashdotted,
      thick
    ] table[
      x index=0,
      y index=3,
      col sep=space
    ] {fig/revenue_comparison/different_mechanism/data/gamma_0.85.txt};
    \addlegendentry{second-price}
  \end{axis}
\end{tikzpicture}

%% file: fig/revenue_comparison/different_mechanism/gamma_1.tikz
\begin{tikzpicture}
  \begin{axis}[
    xlabel={$n$},
    ylabel={Revenue},
    label style={font=\huge},
    tick label style={font=\huge},
    xmin=2, xmax=7,
    ymin=0.3, ymax=0.42,
    xtick={2,3,4,5,6,7},
    legend style={
      at={(0.98,0.3)},   %
      anchor=south east,
      draw=none,
      fill=white,
      font=\huge,
      /tikz/every even column/.append style={column sep=1em}
    },
    grid=none,
    width=9cm,
    height=7cm
  ]

    \addplot[
      color=blue,
      mark=square*,
      dashed,
      thick
    ] table[
      x index=0,
      y index=2,
      col sep=space
    ] {fig/revenue_comparison/different_mechanism/data/gamma_1.txt};
    \addlegendentry{all-pay}

    \addplot[
      color=black,
      mark=*,
      solid,
      thick
    ] table[
      x index=0,
      y index=1,
      col sep=space
    ] {fig/revenue_comparison/different_mechanism/data/gamma_1.txt};
    \addlegendentry{first-price}

    \addplot[
      color=red,
      mark=triangle*,
      mark size=4pt,
      dashdotted,
      thick
    ] table[
      x index=0,
      y index=3,
      col sep=space
    ] {fig/revenue_comparison/different_mechanism/data/gamma_1.txt};
    \addlegendentry{second-price}
  \end{axis}
\end{tikzpicture}

%% file: sections/extension.tex
\section{Extensions}
\label{sec:extension}

In this section, we discuss three extensions of the base model: auctions with reserve prices, the joint design of auction formats and prescreening, and the commonly observed objective in all-pay auctions of maximizing the expected highest bid.

\subsection{Auctions with Reserve Prices}
\label{subsec:reserve_price}

In auctions, a reserve price is the minimum bid at which the seller is willing to transact. Reserve prices are commonly used in practice to increase the seller’s expected revenue. The proposition below shows that all our main results for the three auction formats considered continue to hold in the presence of a reserve price.

\begin{proposition}
\label{prop:reserve_prices}
The results on optimal prescreening for second-price (\Cref{thm:opt_admittednumber_secondprice}), first-price (Theorem \ref{thm:opt_admittednumber_firstprice} and \Cref{thm:rev_prediction_firstprice}), and all-pay auctions (\Cref{thm:opt_admittednumber_allpay}) remain valid for any reserve price $r \ge 0$.
\end{proposition}

The rationale for \Cref{prop:reserve_prices} is as follows. In a second-price auction with a reserve price, the item is sold to the highest bidder only if her bid meets or exceeds the reserve \(r\), and she pays the maximum of the reserve price and the second-highest bid. Truthful bidding remains a dominant strategy since bidders know their valuations. Consequently, the seller’s expected revenue depends solely on and increases with the highest and second-highest valuations, both of which increase with the number of admitted bidders and with prediction accuracy, as shown in \Cref{lem:FOSD_number}(ii) and (iii). Hence, \Cref{thm:opt_admittednumber_secondprice} continues to hold.

In a first-price auction with a reserve price $r$, the item is sold to the highest bidder, provided her bid meets or exceeds $r$. The winner pays an amount equal to her own bid. By adjusting the boundary condition in the first-order condition, the SSM equilibrium bidding strategy is given by:
\[
\sigma^{\text{FP}}(v_i; n, F) \;=\;
\begin{cases}
v_i - \displaystyle\int_{r}^{v_i} \exp\left(-\int_{t}^{v_i} \RHR(x; n, F)dx\right)\,dt, & \text{if } v_i \ge r,\\[6pt]
\text{any bid strictly below } r, & \text{if } v_i < r.
\end{cases}
\]
The cutoff valuation in the equilibrium strategy always equals the reserve price \(r\), remaining unchanged even though the joint valuation distribution varies with the number of admitted bidders. Consequently, given a fixed reserve price \(r\), the equilibrium strategy preserves all the properties of the scenario analyzed previously without a reserve price. Thus, the optimality result on admitting all bidders (\Cref{thm:opt_admittednumber_firstprice}) and the revenue-loss result for predictors of different accuracies (\Cref{thm:rev_prediction_firstprice}) remain valid.

In all-pay auctions with a reserve price, only the allocation rule is affected: the highest bidder wins the item only if her bid meets or exceeds $r$. The payment rule remains unchanged: regardless of the outcome, all bidders pay their respective bids. By similarly adjusting the boundary condition in the first-order condition, the SSM equilibrium strategy with a reserve price is:
\[
\sigma^{\FP}(v_i; n, F) \;=\;
\begin{cases}
r \;+\; \displaystyle\int_{\tau(n,F)}^{v_i} x\,h(x \mid x; n, F)\,dx, & \text{if } v_i \ge \tau(n,F), \\[6pt]
0, & \text{if } v_i < \tau(n,F).
\end{cases}
\]
Unlike in first-price auctions, the cutoff valuation $\tau(n,F)$ here exceeds the reserve price $r$, ensuring that a bidder with valuation $\tau(n,F)$ is indifferent between bidding zero and bidding exactly $r$. The proof further shows that, under a perfect predictor, this cutoff valuation increases with the number of admitted bidders. Recall from \Cref{eq:SSM_equilibrium_gamma_1_all_pay} that under a perfect predictor, the equilibrium strategy without a reserve price, $\int_0^{v_i} x\,h(x \mid x; n, \PP)\,dx$, decreases with the admitted number. The increasing cutoff valuation, which narrows the range of integration, further strengthens this decreasing property with respect to the admitted number. Consequently, under a perfect predictor, admitting only two bidders remains optimal in all-pay auctions with a reserve price. Conversely, under the null predictor, admitting all bidders is clearly optimal.

\subsection{Joint Design of Auctions and Prescreening}

We now consider the problem of jointly optimizing the auction format and the prescreening process. This involves two steps: (i) deriving the optimal auction for a given admitted number, and (ii) comparing the resulting revenue across different admitted numbers under the corresponding optimal auction. The first step is particularly challenging because the prescreening process induces a very general correlation structure among bidders’ valuations, as noted earlier. When valuations are correlated, the optimal auction can be complex, and auctions with a reserve price are generally \emph{not} optimal \citep{myerson_1981_optimal_auction}. Following \citet{cremer_1988_mechanism_fullsurplus} and \citet{mcafee_1992_mechanism_correlatedinfo}, we characterize the optimal mechanism under certain tractable conditions (defined later). To the best of our knowledge, when these conditions are not satisfied, characterizing the optimal mechanism remains an open problem.

For notational convenience, we omit the dependence on $F$ in both $H(x \mid v_i; n,F)$ and $H_2(x \mid v_i; n,F)$ whenever it is clear from the context. We now define a specific auction format, denoted by $\FBM(n; F)$,
when $n$ bidders are admitted based on a predictor $F$, as follows.

\begin{definition}
\label{def:firstbest_mechanism}
Let $b = (b_i)_{i \in \I}$ denote the bidding profile of $n$ admitted bidders. The \emph{McAfee-Reny} auction \emph{\FBM$(n; F)$} is defined as follows:
\begin{enumerate}
    \item The item is allocated to the bidder with the highest bid (report), with ties broken randomly.
    \item Bidder $i$ is charged the following payment:
    \begin{align*}
        &\text{if } b_i > \max_{j \in \I_{-i}} b_j: \quad 
        b_i \;-\; \frac{\bigl[H(b_i \mid b_i; n) + b_i H_2(b_i \mid b_i; n)\bigr] \cdot H(b_i \mid b_i; n)}{H_2(b_i \mid b_i; n)}
        \;+\; \frac{H(b_i \mid b_i; n) + b_i H_2(b_i \mid b_i; n)}{H_2(b_i \mid b_i; n)}; \\[6pt]
        &\text{otherwise}: \qquad
        b_i \;-\; \frac{\bigl[H(b_i \mid b_i; n) + b_i H_2(b_i \mid b_i; n)\bigr] \cdot H(b_i \mid b_i; n)}{H_2(b_i \mid b_i; n)}.
    \end{align*}

\end{enumerate}
\end{definition}
The mechanism $\FBM(n; F)$ defined in \Cref{def:firstbest_mechanism} also depends on the total number of potential bidders $m$. However, since $m$ is exogenously given, we omit this dependence in the notation for consistency with other terms. Note that $\FBM(n; F)$ resembles an all-pay auction in that bidders are charged regardless of whether they win the item.

\begin{proposition}[Theorem 4 in \citealt{mcafee_1992_mechanism_correlatedinfo}]
Suppose that $n$ bidders are admitted based on a predictor $F$. If for all $(x, v_i) \in [0,1]^2$, we have
\begin{align}
\label{eq:condition_FBmechanism}
H_2(x \mid v_i; n) < 0, 
\quad \text{and} \quad
\frac{\partial}{\partial v_i} \left(
\frac{H(x \mid v_i; n)}{H_2(x \mid v_i; n)}
\right) \ge -1.
\tag{$\mathsf{MRCon}$}
\end{align}
Then the mechanism \emph{\FBM}$(n; F)$ is optimal, which extracts all bidders’ surplus and achieves revenue 
$\mathbb{E}_{X \sim G^{\lar}(\cdot; n, F)}[X]$.
\end{proposition}

When condition \eqref{eq:condition_FBmechanism} holds, it can be shown that under the mechanism $\FBM(n; F)$, bidding truthfully constitutes a Bayesian–Nash equilibrium. This mechanism yields zero expected utility for all bidders and generates revenue equal to the expected highest valuation $\mathbb{E}_{X \sim G^{\lar}(\cdot; n, F)}[X]$, thereby establishing its optimality. When $n = m$ (the i.i.d.\ case), it is straightforward to verify that $H_2\bigl(x \mid v_i; n = m\bigr) = 0$, violating condition \eqref{eq:condition_FBmechanism}. In this case, the optimal mechanism is Myerson’s auction, denoted by $\MA(m)$ for $m$ bidders with i.i.d.\ prior $F_1(\cdot)$. In particular, under certain regularity conditions, $\MA(m)$ is a second-price auction with an appropriate reserve price \citep{myerson_1981_optimal_auction}.

\begin{proposition}
\label{prop:joint_optimality}
If the predictor $F$ satisfies condition \eqref{eq:condition_FBmechanism} for all admitted numbers $n \in [2, m-1]$, then the optimal mechanism is either to admit $m-1$ bidders using \emph{\FBM}$(m-1; F)$ or to admit all $m$ bidders using \emph{\MA}$(m)$. 
\end{proposition}

The condition in \Cref{prop:joint_optimality} guarantees that, when $n$ bidders are admitted, the mechanism \FBM$(n; F)$ is optimal and yields revenue equal to the expected highest valuation. By \Cref{lem:FOSD_number}(ii), this expected highest valuation increases with the number of admitted bidders. Hence, admitting $m-1$ bidders and using \FBM$(m-1; F)$ is (weakly) better than admitting any smaller number of bidders $n \le m-1$ under any alternative auction format. Consequently, determining the global optimum reduces to comparing the outcomes of admitting $m-1$ bidders with \FBM$(m-1; F)$ and admitting all $m$ bidders with \MA$(m)$.

\subsection{Highest Bid in All-pay Auctions}

All-pay auctions are also commonly used to model contests \citep{moldovanu_2006_contest_architecture,Jason_Optimal_Crowdsourcing_Contest}, where the primary metric of interest is sometimes the expected highest bid, reflecting the principle that typically only the best submission matters. Similar to the revenue-maximization problem discussed in \Cref{sec:all-pay_auctions}, explicitly characterizing the optimal admitted number for maximizing the expected highest bid under a general predictor remains elusive (but still computationally tractable). We therefore focus on the two special cases, namely, the null predictor and the perfect predictor.

\begin{proposition}
\label{prop:all-pay_highestbid}
Suppose there exists an SSM equilibrium strategy for all $n\in [2,m]$ in all-pay auctions. For maximizing the expected highest bid:
\begin{enumerate}[(i)]
\item Under the perfect predictor, $n^\ast=2$.

\item Under the null predictor, for the power-law prior distribution $F_1(x)=x^c$, where $c>0$: If $c\in (0,3], n^\ast=m;$ If $c\in \left(3,\frac{5+\sqrt{33}}{2}\right), n^\ast=m \wedge \frac{c^2-c+2}{c^2-3c}$;\footnote{When $\frac{c^2-c+2}{c^2-3c}$ is not an integer, $n^\ast$ is either $m \wedge \bigg\lceil{\frac{c^2-c+2}{c^2-3c}}\bigg\rceil$ or $ m \wedge \bigg\lfloor{\frac{c^2-c+2}{c^2-3c}}\bigg\rfloor$.} If $c\geq\frac{5+\sqrt{33}}{2}\approx 5.372$, $n^\ast=2$.
\end{enumerate}
\end{proposition}

Under the perfect predictor, the bidder with the highest valuation is guaranteed to be selected. Moreover, as shown in \Cref{eq:SSM_equilibrium_gamma_1_all_pay}, the equilibrium strategy $\sigma^{\mathsf{AP}}(v_i; n, \PP)$ decreases with the admitted number for any given valuation $v_i$. Combining these two properties yields \Cref{prop:all-pay_highestbid}(i), which states that admitting only two bidders is optimal. The result for this special case was previously established in \cite{sun_2024_contests}. On the other hand, under the null predictor, the expected highest bid in all-pay auctions is \textit{not} monotonic in the number of participants $n$ and depends heavily on the prior distribution. \Cref{prop:all-pay_highestbid}(ii) considers a power-law prior and shows that when $c$ is small (i.e., the environment is less competitive, with a large proportion of bidders having low valuations), admitting all bidders is optimal; when $c$ is large (i.e., the environment is highly competitive, with most bidders having high valuations), admitting only two bidders is optimal; and when $c$ lies in an intermediate range, the optimal number of admitted bidders is between $2$ and $m$. This highlights the potential benefits of prescreening, even when the predictor provides no useful information. Additional numerical results beyond these two boundary cases are provided in \Cref{app_subsec:numerical}.

%% file: sections/7-conclusion.tex
\section{Closing Remarks}
\label{sec:conclusion}

Our analysis reveals that the value of prediction-based prescreening is highly dependent on the auction format. For platforms that utilize first-price or second-price auctions, such as Google Ads, broad competition generally outperforms selective exclusion. This implies that when external constraints necessitate limiting participation, platforms face potential revenue loss. However, our findings suggest a clear remedy: investing in more accurate valuation predictors, such as advanced machine learning algorithms, can help compensate for this loss by ensuring the admitted bidders are of the highest quality. Conversely, our result on all-pay auctions offers a different directive for platforms like crowdsourcing contests. Here, less can be more; selectively limiting the auction to a few predicted high-valuation bidders can substantially enhance revenue, especially when the seller has access to an accurate predictor.

There are several promising directions for future research. First, this paper assumes that signals are observed only by the seller. An important extension would be to study how bidder admission and information disclosure can be jointly optimized. Second, our revenue ranking result---that the optimal revenue under all-pay auctions exceeds that of both first-price and second-price auctions---relies on the assumption of risk neutrality. Relaxing this assumption presents a fruitful avenue for further investigation. Finally, the concept of prescreening may extend beyond auction design to other operational settings, such as competitive pricing among firms. For example, in the San Francisco government’s Powered Scooter Share Permit Program,\footnote{\url{https://www.sfmta.com/projects/powered-scooter-share-permit-program}} only two scooter companies were granted licenses and subsequently competed for customer demand through pricing and other dimensions. Exploring how prescreening can be applied in such contexts to enhance social welfare would be a valuable direction for future work.

%% file: sections/appendix.tex
\part*{Appendix} 
\parttoc %

\setcounter{equation}{0}
\setcounter{proposition}{0}
\setcounter{lemma}{0}

\numberwithin{figure}{section}   %
\renewcommand{\thefigure}{\thesection.\arabic{figure}}

We provide auxiliary results in \Cref{app_sec:auxiliary_results}. All proofs can be found in \Cref{app_sec:proofs}. 
Further results under the hallucinatory predictors are provided in \Cref{app_subsec:beliefs}.

\section{Auxiliary Results}
\label{app_sec:auxiliary_results}

We provide further discussion of \Cref{assum:increasing_conditional_expectation} in \Cref{app_sec:discusions_assumption}, additional numerical results in \Cref{app_subsec:numerical}, further analysis of first-price auctions in \Cref{app_subsec:condition_firstprice_SSM}, and extension to multiple homogeneous items in \Cref{subsec:multiple_items}.

\subsection{Discussions of \Cref{assum:increasing_conditional_expectation}.}
\label{app_sec:discusions_assumption}

We now consider the scenario in which \Cref{assum:increasing_conditional_expectation} does not hold. In this case, the structure of the posterior belief described in \Cref{thm:beta_posterior} still applies---specifically, it remains proportional to the product of the prior and the admission probability. Consequently, \Cref{lem:joint_dis_g}, which characterizes the equivalent game, continues to hold. However, the admission probability in this setting becomes significantly more complex:
\begin{align*}
  \psi(v;n,F) = 
  \Pr\left\{\left(\max_{j\in \bar{\I}} \mathbb{E}[V_j\mid S_j]\right) \leq
  \left(\min_{k\in \I} \mathbb{E}[V_k\mid S_k]\mid v_k\right)\right\}.
\end{align*}
\Cref{alg:admissionprob} provides a procedure for computing the admission probability.
\begin{algorithm}[ht!]
\small
\caption{Admission Probability without \Cref{assum:increasing_conditional_expectation}}
\begin{algorithmic}[1]
\Require Total number of bidders $m$, predictor $F$, admitted number $n$, valuation vector $v\in[0,1]^n$.
\Ensure Admission probability $\psi(v;n,F)$.

\State Compute the function $\pi(s_i) = \mathbb{E}[V_i\mid S_i = s_i] = \int_0^1 1- F_{1\mid 2}(t \mid s_i)dt$.

\State Derive the distribution of the random variable $X_j = \pi(S_j)$, with $S_j \overset{\text{i.i.d.}}{\sim} F_2$.

\For{$k \in \I$}
    \State Derive the distribution of the random variable $Y_k(v_k) = \pi(Z_k)$, with $Z_k \sim F_{2\mid 1}(\cdot \mid v_k)$ independently.
\EndFor

\State Compute $\psi(v;n,F) = \Pr\{\max_{j\in \bar{\I}} X_j \leq \min_{k\in \I} Y_k(v_k)\}$.
\end{algorithmic}
\label{alg:admissionprob}
\end{algorithm}

In general, the complexity of $F_{1\mid 2}$ and $F_{2\mid 1}$ makes even the numerical computation of the admission probability challenging. For this reason, we focus on scenarios that satisfy \Cref{assum:increasing_conditional_expectation}. Nevertheless, as discussed earlier, many commonly used distributions meet this assumption (see \Cref{exm:copulas_increasing_expectation}).

\subsection{Additional Numerical Results}
\label{app_subsec:numerical}

\subsubsection{All-pay Auctions}

In this section, we present extensive numerical results on all-pay auctions beyond the extreme cases of null and perfect predictors. We first focus on the hallucinatory predictor defined in \eqref{eq:predictor_hallucination} and use a power-law prior distribution $F_1(x) = x^c$, with $c$ chosen to ensure that an SSM equilibrium exists for all $\gamma \in [0,1]$ (i.e., small $c$). Several patterns emerge: optimal revenue is weakly increasing in both the prediction accuracy $\gamma$ and the prior parameter $c$; prescreening significantly improves revenue---for example, when $m = 7$, $c = 0.2$, and $\gamma = 1$, admitting only two bidders yields a 31\% revenue increase over admitting all bidders. Moreover, with small $c$, the optimal prescreening strategy (in terms of both expected revenue and expected highest bid) often involves either admitting all bidders or only two, though this binary outcome does not hold in general (see \Cref{prop:all-pay_highestbid}(ii)).

\begin{figure}[ht]
    \centering
    \begin{subfigure}[b]{0.49\textwidth} %
    \includegraphics[width=\textwidth]{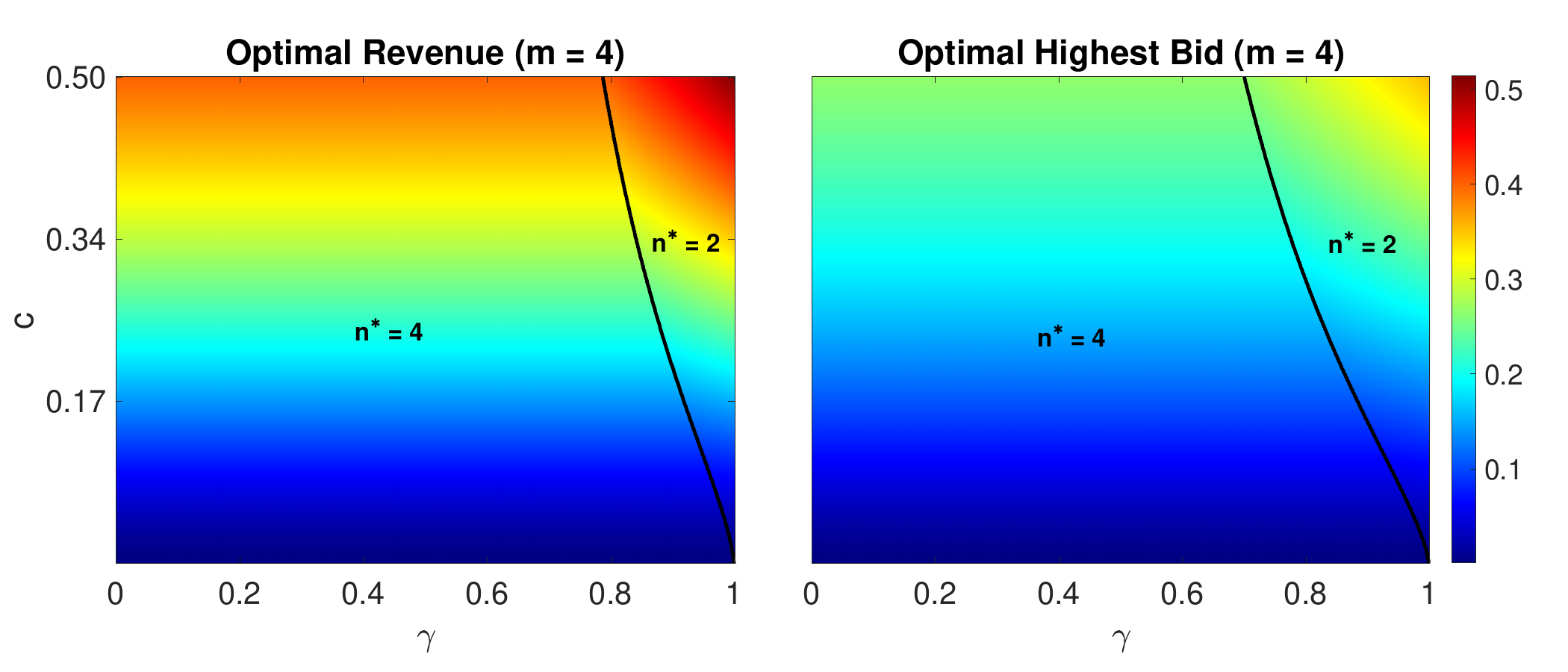}
    \caption{$m=4$}
    \label{fig:rev_m_4_allpay}
    \end{subfigure}
    \hfill 
    \begin{subfigure}[b]{0.49\textwidth}
    \includegraphics[width=\textwidth]{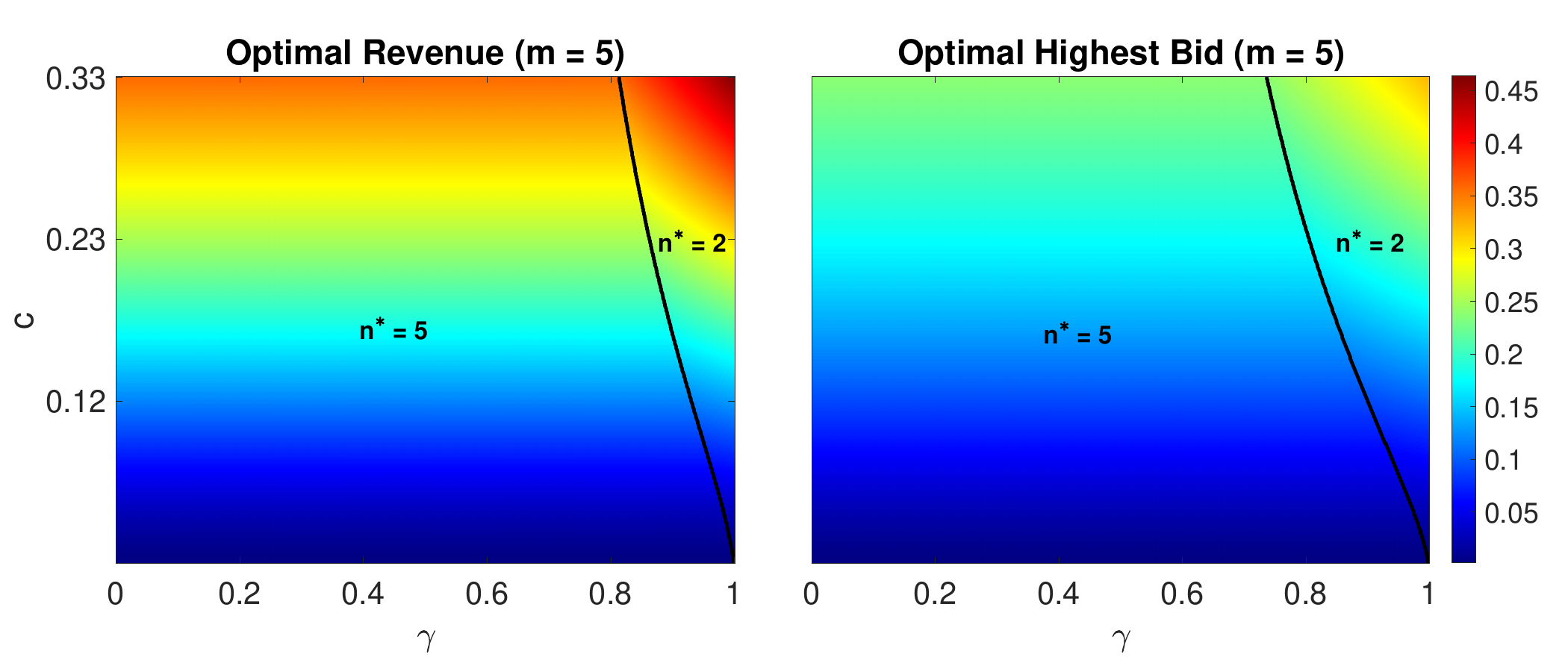}
    \caption{$m=5$}
    \label{fig:hb_m_5_allpay}
    \end{subfigure}

\begin{subfigure}[b]{0.49\textwidth}
    \includegraphics[width=\textwidth]{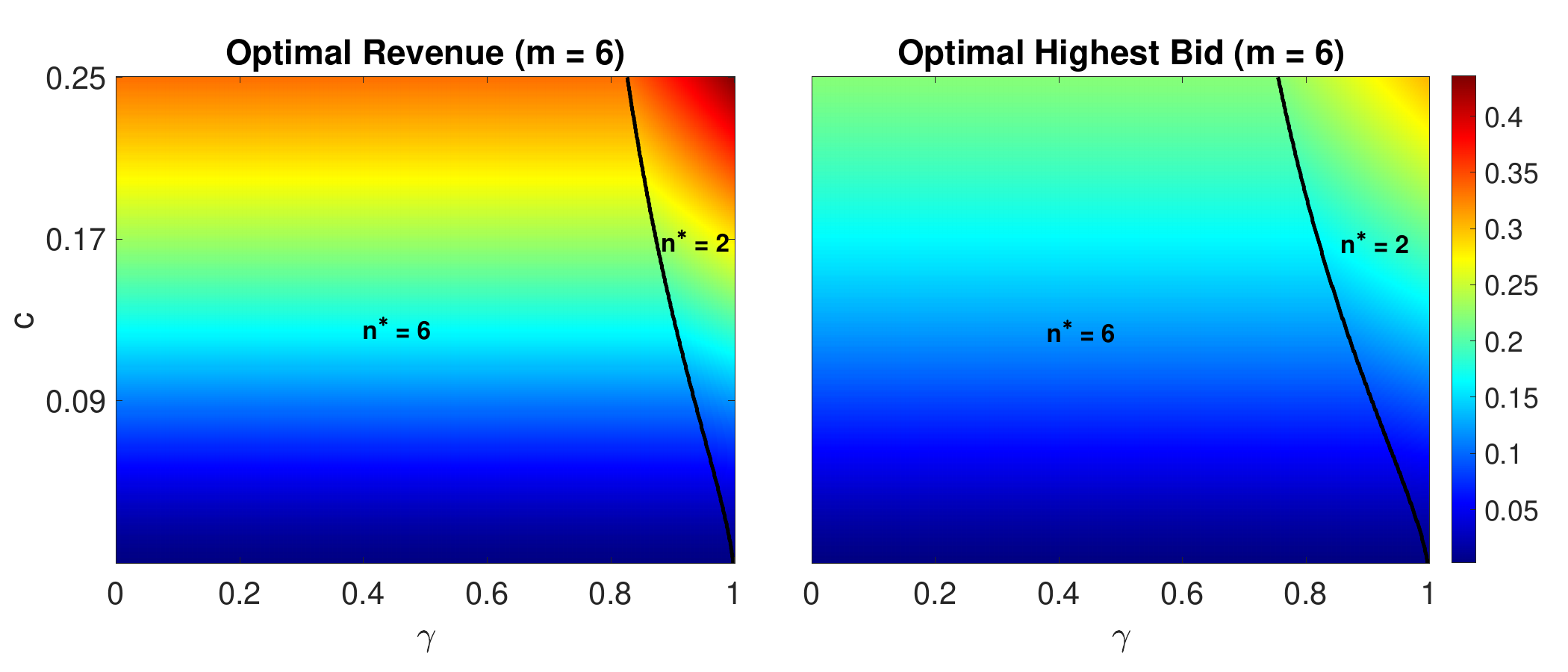}
    \caption{$m=6$}
    \label{fig:hb_m_6_allpay}
    \end{subfigure}
    \hfill
    \begin{subfigure}[b]{0.49\textwidth}
    \includegraphics[width=\textwidth]{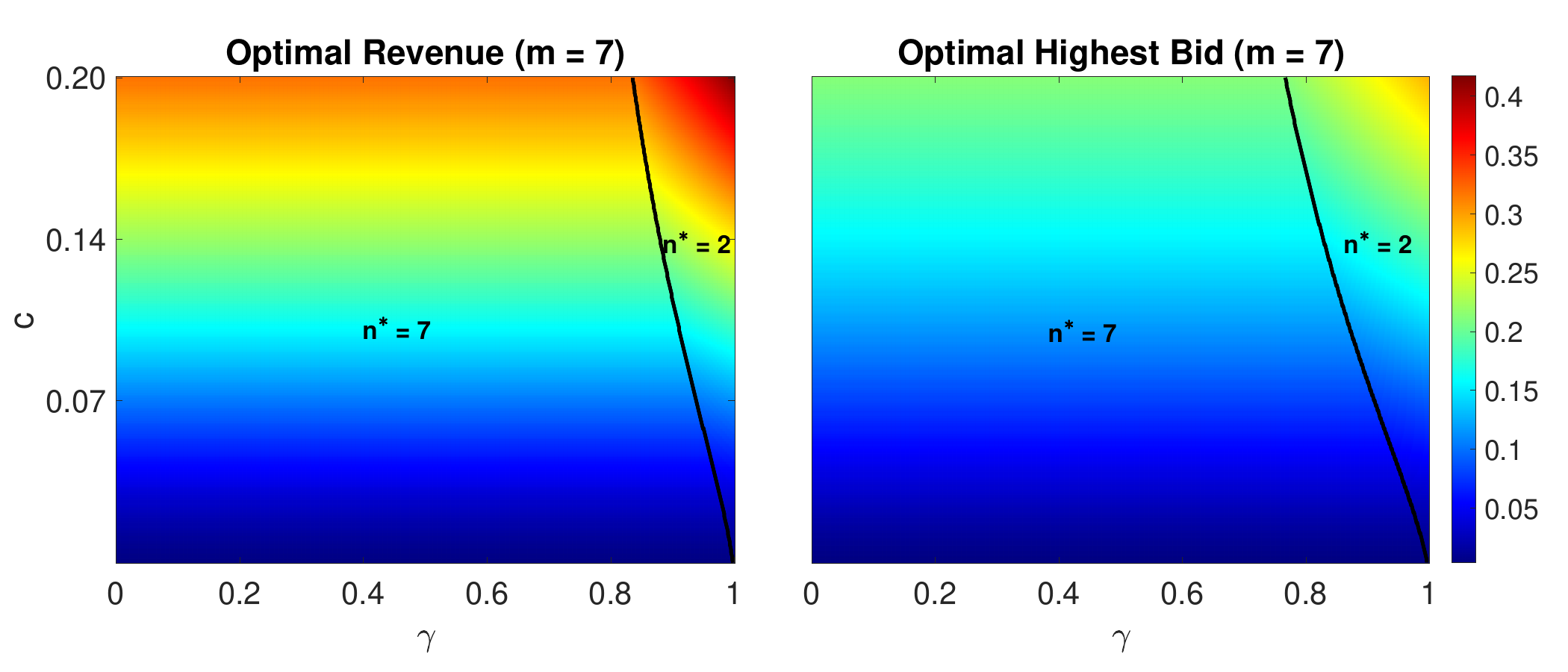}
    \caption{$m=7$}
    \label{fig:hb_m_7_allpay}
    \end{subfigure}
    
    \caption{Optimal Prescreening for All-pay Auctions under Hallucinatory Predictors \eqref{eq:predictor_hallucination}. 
    }
    \label{fig:general_m_allpay}
\medskip    
      Note. \textit{The prior distribution is $F_1(x)=x^c,c\in(0,1/(m-2)]$ when the total number is $m$. 
      Besides, $F_2=F_1$.
      The condition $c\leq 1/(m-2)$ guarantees that there exists an SSM equilibrium strategy for all $n\in\{2,3,\cdots,m\}$ and for {\normalfont{all}} $\gamma\in[0,1]$.
    The color represents the optimal expected revenue 
    in the left panel (optimal highest bid 
    in the right panel) given each $(\gamma,c)$.}
\end{figure}

\noindent\emph{FGM Copulas}.
\Cref{fig:FGM_allpay} illustrates how revenue in all-pay auctions varies with the number of admitted bidders under the FGM copula defined in \Cref{exm:copulas_increasing_expectation}(iii). The patterns resemble those observed with hallucinatory predictors. Specifically, revenue is generally non-monotonic in the number of admitted bidders; for any fixed admitted number, it increases with the parameter $\alpha$, which captures the correlation between valuation and signal. Moreover, prescreening tends to yield higher revenue when the prior (marginal) distribution exhibits stronger first-order stochastic dominance, reflecting a more competitive environment. For example, in \Cref{fig:FGM_allpay}(a) and (b), admitting all bidders is optimal regardless of $\alpha$, while in \Cref{fig:FGM_allpay}(d), with prior $F_1(x) = x^{17}$, admitting 5 bidders is optimal for $\alpha = 1$, and 6 bidders for $\alpha = 0.9$.

\begin{figure}[ht]
    \centering
        \begin{subfigure}[b]{0.49\textwidth}
        \centering
        \resizebox{\textwidth}{!}{%
            \input{fig/general_predictor/FGM/allpay_rev_n_c_1.tikz}%
        }
        \caption{Marginal Distribution: $F_1(x)=x$}
     
    \end{subfigure}
    \hfill 
   \begin{subfigure}[b]{0.49\textwidth}
        \centering
        \resizebox{\textwidth}{!}{%
            \input{fig/general_predictor/FGM/allpay_rev_n_c_5.tikz}%
        }
        \caption{Marginal Distribution: $F_1(x)=x^5$}
     
    \end{subfigure}
  \begin{subfigure}[b]{0.49\textwidth}
        \centering
        \resizebox{\textwidth}{!}{%
            \input{fig/general_predictor/FGM/allpay_rev_n_c_10.tikz}%
        }
        \caption{Marginal Distribution: $F_1(x)=x^{10}$}
     
    \end{subfigure}
    \hfill
     \begin{subfigure}[b]{0.49\textwidth}
        \centering
        \resizebox{\textwidth}{!}{%
            \input{fig/general_predictor/FGM/allpay_rev_n_c_17.tikz}%
        }
        \caption{Marginal Distribution: $F_1(x)=x^{17}$}
     
    \end{subfigure}
    
    \caption{Revenue vs. Admitted Number under All-pay Auctions with FMG Copula 
    }
    \medskip    
      Note. \textit{The total number $m=7$. The parameter $\alpha\in[0,1]$ represents the coefficient of the FGM copula, see \Cref{exm:copulas_increasing_expectation}(iii). 
      The marginal distributions are specified in each subfigure.}
    \label{fig:FGM_allpay}
\end{figure}

\begin{figure}[ht]
    \centering
        \begin{subfigure}[b]{0.49\textwidth}
        \centering
        \resizebox{\textwidth}{!}{%
            \input{fig/general_predictor/AMH/allpay_rev_n_c_1.tikz}%
        }
        \caption{Marginal Distribution: $F_1(x)=x$}
     
    \end{subfigure}
    \hfill 
   \begin{subfigure}[b]{0.49\textwidth}
        \centering
        \resizebox{\textwidth}{!}{%
            \input{fig/general_predictor/AMH/allpay_rev_n_c_5.tikz}%
        }
        \caption{Marginal Distribution: $F_1(x)=x^5$}
     
    \end{subfigure}
  \begin{subfigure}[b]{0.49\textwidth}
        \centering
        \resizebox{\textwidth}{!}{%
            \input{fig/general_predictor/AMH/allpay_rev_n_c_10.tikz}%
        }
        \caption{Marginal Distribution: $F_1(x)=x^{10}$}
     
    \end{subfigure}
    \hfill
     \begin{subfigure}[b]{0.49\textwidth}
        \centering
        \resizebox{\textwidth}{!}{%
            \input{fig/general_predictor/AMH/allpay_rev_n_c_17.tikz}%
        }
        \caption{Marginal Distribution: $F_1(x)=x^{17}$}
     
    \end{subfigure}
    
    \caption{Revenue vs. Admitted Number under All-pay Auctions with AMH Copula 
    }
    \medskip    
      Note. \textit{The total number $m=7$. The parameter $\alpha\in[0,1]$ represents the coefficient of the AMH copula, see \Cref{exm:copulas_increasing_expectation}(ii). 
      The marginal distributions are specified in each subfigure.}
    \label{fig:AMH_allpay}
\end{figure}

\noindent\emph{AMH Copulas}.
\Cref{fig:AMH_allpay} presents the results for the AMH copula, as defined in \Cref{exm:copulas_increasing_expectation}(ii), which exhibit a similar pattern to those observed under the FGM copula. A comparison between \Cref{fig:FGM_allpay} and \Cref{fig:AMH_allpay} reveals that, for any given number of admitted bidders and value of $\alpha$, the revenue under the AMH copula is slightly higher. This improvement arises because the AMH copula induces stronger correlations between valuations and signals, enabling more accurate predictions; see the formal statement in \Cref{obs:higher_accurate_common_predictors}.

\subsection{Further Discussions on the Existence of an SSM Equilibrium in First-price Auctions}
\label{app_subsec:condition_firstprice_SSM}

The equilibrium existence condition for first-price auctions stated in \Cref{thm:SSM_firstprice}(i) can be challenging to verify. Below, we provide a simple example, under the hallucinatory predictor with $m = 3$, that satisfies this condition.

\begin{example}\label{exm:hpm=3}
Under any hallucinatory predictor defined in \eqref{eq:predictor_hallucination}, and assuming a power-law prior distribution $F_1(x) = x^c$, the condition in \Cref{thm:SSM_firstprice}(i) holds for any $n \in \{2, m\}$ when $m = 3$ and $c \in (0, 1]$.
\end{example}

\Cref{exm:hpm=3} demonstrates that in less competitive environments (i.e., for small values of $c$), an SSM equilibrium strategy exists for \emph{any} hallucinatory predictor when $n = 2$ and $m = 3$.\footnote{For any predictor and prior distribution, an SSM equilibrium strategy always exists in first-price auctions when $n = m$.} This result is non-trivial, as it violates \emph{all} standard conditions for equilibrium existence known in the literature, yet an SSM equilibrium still emerges. For instance, \citet{milgrom_1982_auctiontheory_competitive_bidding} shows that an SSM equilibrium exists when the joint density $g(\cdot; n, F)$ is affiliated. However, as discussed in \Cref{subsec:equivalent_game}, the joint density is generally \emph{not} affiliated except in the cases of perfect or null predictors. Additionally, \citet{castro_2007_affiliation_positive_dependence} proposes a weaker condition for equilibrium existence, namely, that $\frac{h(x \mid v_i; n, F)}{H(x \mid v_i; n, F)}$ is weakly increasing in $v_i \in [0,1]$ for each fixed $x \in [0,1]$, which also fails in our setting. To prove \Cref{exm:hpm=3}, we rely on the inequality $\exp\left(-\int_{x}^{v_i} \frac{h(t \mid t; n, F)}{H(t \mid t; n, F)} dt\right) \leq 1$. Since this bound is independent of the prior, it naturally leads to a condition on the prior distribution: the condition that works under this bound is $c\leq 1$.

We note that the results in \Cref{thm:SSM_firstprice}(i) and (ii) also extend to general first-price auctions with \emph{arbitrarily} correlated valuation distributions (densities) $g$. The definitions of $h$, $H$, and $\FP(\tilde{v}_i, v_i)$ for a general correlated density $g$ follow analogously to those in our setting, and we omit the parameters $(n, F)$ from these functions for clarity. Below, we summarize the equilibrium existence conditions for first-price auctions with arbitrarily correlated types that have been established in the literature and compare them with our own.
\begin{align}
  \quad & \textrm{Existence of an SSM equilibrium strategy in first-price auctions with \textit{arbitrarily} correlated types}. \nonumber \\
\Longleftarrow \quad &\textrm{The condition in \Cref{thm:SSM_firstprice}(i)}. \nonumber\\
\Longleftarrow \quad & \textrm{For any } v_i^\prime \leq \tilde{v}_i \leq v_i,\quad  \frac{h(\tilde{v}_i\mid v_i^\prime)}{H(\tilde{v}_i\mid v_i^\prime)} 
\leq \frac{h(\tilde{v}_i\mid \tilde{v}_i)}{H(\tilde{v}_i\mid \tilde{v}_i)} 
\leq 
 \frac{h(\tilde{v}_i\mid v_i)}{H(\tilde{v}_i\mid v_i)} .\label{eq:condition_two_inequalities}\\
\Longleftarrow \quad & \textrm{For any given } \tilde{v}_i \in [0,1],\; \frac{h(\tilde{v}_i\mid v_i)}{H(\tilde{v}_i\mid v_i)} \textrm{ is weakly increasing in } v_i \in [0,1] .\nonumber
\\
\Longleftarrow \quad & 
\textrm{$h(\tilde{v}_i\mid v_i)$ is affiliated in $(\tilde{v}_i,v_i) \in [0,1]^2$.} \nonumber
\\
\Longleftarrow \quad & 
\textrm{$g(v)$ is affiliated in $v \in [0,1]^n$.} \tag{\textit{joint affiliation}}  
\end{align}

Joint affiliation was first introduced by \citet{milgrom_1982_auctiontheory_competitive_bidding}, and the condition in \eqref{eq:condition_two_inequalities} was formalized by \citet{castro_2007_affiliation_positive_dependence}. However, these conditions generally fail in our setting, except in special cases (see \Cref{remark:affiliation}). In fact, when $m = 3$ and $n = 2$, we find that the function $\frac{h(\tilde{v}_i \mid v_i)}{H(\tilde{v}_i \mid v_i)}$ is decreasing in $v_i \in [\tilde{v}_i, 1]$. That is, the second inequality in \eqref{eq:condition_two_inequalities} is completely \textit{reversed}. Moreover, no monotonicity holds when $v_i \in [0, \tilde{v}_i]$. These features make \Cref{exm:hpm=3} particularly challenging to prove. We outline the proof sketch below.

We first show that, for any $\tilde{v}_i$, the ratio $\frac{h(\tilde{v}_i \mid v_i)}{H(\tilde{v}_i \mid v_i)}$ initially decreases and then increases as a function of $v_i$ over the interval $v_i \leq \tilde{v}_i$. Moreover, we establish that
\[
\frac{h(\tilde{v}_i \mid 0)}{H(\tilde{v}_i \mid 0)} \leq \frac{h(\tilde{v}_i \mid \tilde{v}_i)}{H(\tilde{v}_i \mid \tilde{v}_i)}.
\]
Together, these two properties confirm the first inequality in \eqref{eq:condition_two_inequalities}. Since the second inequality does not hold in our setting, we instead work directly with $\FP(\tilde{v}_i, v_i)$ by applying the upper bound $\exp\left(-\int_{x}^{v_i} \frac{h(t \mid t; n, F)}{H(t \mid t; n, F)} \, dt\right) \leq 1.$ The complete proof is provided in \Cref{app_sec:proofs}.

 \subsection{Extension: Auctions of Multiple Items}
 \label{subsec:multiple_items}

We now consider an auction with $K$ homogeneous items. The $K$ highest bidders each receive one item, and all winners pay a uniform price equal to the $(K+1)$st-highest bid. Ties, if any, are broken randomly. This type of auction is commonly referred to as a uniform auction \citep{milgrom_1981_rational_expectations,pekevc_2008_revenueranking_random_number}. It is straightforward to verify that a uniform auction reduces to a second-price auction when $K = 1$. To avoid trivial cases, we assume the number of admitted bidders satisfies $n \ge \max\{2, K\}$.

In uniform auctions, truthful bidding constitutes an equilibrium strategy \citep{milgrom_1981_rational_expectations}. The expected revenue is therefore equal to $K$ times the expected $(K+1)$st-highest valuation. By \Cref{lem:FOSD_number}(ii), the $(K+1)$st-highest valuation increases with the admitted number $n$ in the sense of first-order stochastic dominance, implying that admitting all bidders is optimal. Moreover, by \Cref{lem:FOSD_number}(iii), the $(K+1)$st-highest valuation increases with prediction accuracy, meaning that the revenue loss from admitting fewer bidders decreases as prediction accuracy improves. The results are summarized in the proposition below.

\begin{proposition}
\label{prop:uniform_auctions}
Admitting all bidders is revenue-maximizing in uniform auctions.
Moreover, the revenue loss from admitting fewer bidders is (weakly) decreasing in the prediction accuracy. 
\end{proposition}

\section{Proofs}
\label{app_sec:proofs}

Let \( f \) denote the density of the predictor \( F \); analogously, let \( f_{1|2} \) and \( f_{2|1} \) denote the corresponding conditional densities. For convenience, we adopt the following notation throughout the proofs:
\[
\begin{aligned}
    C_1(x,y) &:= \frac{\partial }{\partial x} C(x,y), 
    &\quad 
    C_2(x,y) &:= \frac{\partial }{\partial y} C(x,y), \\
    C_{11}(x,y) &:= \frac{\partial^2 }{\partial x^2} C(x,y), 
    &\quad 
    C_{12}(x,y) &:= \frac{\partial^2 }{\partial x \partial y} C(x,y), 
    &\quad 
    C_{22}(x,y) &:= \frac{\partial^2 }{\partial y^2} C(x,y).
\end{aligned}
\]

\subsection{Proofs for \Cref{sec:model}}
\begin{proof}[\textbf{Proof of \Cref{exm:copulas_increasing_expectation}}] We first prove that the copulas in \Cref{exm:copulas_increasing_expectation} satisfy \Cref{assum:increasing_conditional_expectation}. Observe that
\[
f(v_i, s_i) 
= \frac{\partial^2 C(F_2(v_i), F_2(s_i))}{\partial v_i \, \partial s_i}
= C_{12}\big(F_2(v_i), F_2(s_i)\big) \, f_1(v_i) f_2(s_i),
\]
and hence,
\begin{align}
F_{1 \mid 2}(v_i \mid s_i)
&= \int_0^{v_i} f_{1 \mid 2}(u \mid s_i) \, du
= \int_0^{v_i} \frac{f(u, s_i)}{f_2(s_i)} \, du = \int_0^{v_i} C_{12}\big(F_1(u), F_2(s_i)\big) \, f_1(u) \, du \notag\\
&= C_2\big(F_1(v_i), F_2(s_i)\big) \label{eq:copula formulation of F_{1|2}}. 
\end{align}
Therefore,
\begin{align}\label{eq:E V_i condition s_i}
 \mathbb{E}[V_i \mid S_i = s_i]
= \int_0^1 \big[1 - F_{1 \mid 2}(v_i \mid s_i)\big] \, dv_i
= \int_0^1 \big[1 - C_2(F_1(v_i), F_2(s_i))\big] \, dv_i.   
\end{align}

\noindent\underline{\textit{Proof of (i).}}

For the comonotonic copula \( C(x,y) = \min\{x, y\} \), we have $C_2(x,y) = \1\big\{ y \le x \big\}.$
Therefore, 
\[
\mathbb{E}[V_i \mid S_i = s_i] 
= \int_0^1 \big[1 - \1\big\{ F_2(s_i) \le F_1(v_i) \big\} \big] \, dv_i 
= F_1^{-1}\big( F_2(s_i) \big),
\]
which is non-decreasing in \( s_i \in [0,1] \).

\medskip
\noindent\underline{\textit{Proof of (ii).}}

\medskip
For the AMH copula $C(x,y) = \frac{xy}{1 - \alpha\, (1 - x)(1 - y)},$
we have
\[
C_2(x,y) 
= \frac{1}{\big[ 1 - \alpha (1 - x)(1 - y) \big]^2} 
\Big( x \big[ 1 - \alpha (1 - x)(1 - y) \big] - xy\, \alpha (1 - x) \Big)
= \frac{\, x\, [ 1 - \alpha (1 - x) ] \,}{\big[ 1 - \alpha (1 - x)(1 - y) \big]^2}.
\]
Similarly,
\[
C_{22}(x,y)
= -\, 2x\, [ 1 - \alpha (1 - x) ] \, 
\frac{\, \alpha (1 - x) \,}{\big[ 1 - \alpha (1 - x)(1 - y) \big]^3 }.
\]
Therefore,
\begin{align*}
    \frac{\partial}{\partial s_i} 
\mathbb{E}[ V_i \mid S_i = s_i ]
&= \frac{\partial}{\partial s_i} 
\int_0^1 \big[ 1 - F_{1 \mid 2}(v_i \mid s_i) \big] \, dv_i
= - \int_0^1 \frac{\partial F_{1 \mid 2}(v_i \mid s_i) }{\partial s_i} \, dv_i
\\&= - \int_0^1 C_{22}\big( F_1(v_i), F_2(s_i) \big) \, f_2(s_i) \, dv_i 
\;\ge 0,
\end{align*}
which implies that \( \mathbb{E}[ V_i \mid S_i = s_i ] \) is non-decreasing in \( s_i \in [0,1] \).

\medskip
\noindent\underline{\textit{Proof of (iii).}}

\medskip
For the FGM copula $C(x,y) = xy \big( 1 + \alpha (1 - x)(1 - y) \big),$
we have
\[
C_2(x,y) 
= x \big[ 1 + \alpha (1 - x)(1 - 2y) \big],
\quad
C_{22}(x,y)
= -2 \alpha x (1 - x).
\]
Therefore,
\[
\frac{\partial}{\partial s_i} 
\mathbb{E}\big[ V_i \mid S_i = s_i \big]
= - \int_0^1 
C_{22}\big( F_1(v_i), F_2(s_i) \big) \, f_2(s_i) \, dv_i 
\;\ge 0,
\]
which implies that 
\(
\mathbb{E}\big[ V_i \mid S_i = s_i \big]
\)
is non-decreasing in \( s_i \in [0,1] \).

\medskip
Since the partial derivative of a convex combination equals the corresponding convex combination of the partial derivatives, it follows from \Cref{eq:E V_i condition s_i} that if both predictors 
\( C(F_1(\cdot), F_2(\cdot)) \) and 
\( \tilde{C}(F_1(\cdot), F_2(\cdot)) \) 
satisfy \Cref{assum:increasing_conditional_expectation}, 
then any convex combination of 
\( C(F_1(\cdot), F_2(\cdot)) \) 
and 
\( \tilde{C}(F_1(\cdot), F_2(\cdot)) \) 
also satisfies 
\Cref{assum:increasing_conditional_expectation}.
\end{proof}

\begin{proof}[\textbf{Proof of \Cref{obs:higher_accurate_common_predictors}}]
We only show that
\begin{align*} 
 \text{\emph{hallucinatory predictor \eqref{eq:predictor_hallucination} with parameter $\gamma$}} \pqdorder \text{\emph{AMH predictor with $\alpha=\gamma$}},
\end{align*}
as the other \(\pqdorder\) inequalities can be verified straightforward.

For any fixed \( x, y \in [0,1] \), define  $a(\gamma) 
= [\, \gamma\, x + (1 - \gamma)\, xy\, ] 
\, [\, 1 - \gamma (1 - x)(1 - y) \, ].$
We have
\[
\begin{aligned}
a(\gamma)
&= [\, \gamma\, x (1 - y) + xy\, ]\, [\, 1 - \gamma (1 - x)(1 - y) \, ] \\
&= -\, x (1 - x) (1 - y)^2 \,\gamma^2 
+ x (1 - y)\, [\, 1 - y (1 - x) \, ]\, \gamma + xy.
\end{aligned}
\]
Since \( a(\gamma) \) is quadratic in \(\gamma\) and concave (negative quadratic coefficient), its minimum over \([0,1]\) is attained at either endpoint. Thus,
\[
a(\gamma) \ge \min \big\{ a(0),\, a(1) \big\} 
= \min \big\{ xy,\, x [\, 1 - (1 - x)(1 - y) \, ] \big\} 
= xy.
\]
Similarly,
\[
[\, \gamma\, y + (1 - \gamma)\, xy\, ]\, [\, 1 - \gamma (1 - x)(1 - y) \, ] \ge xy.
\]
Therefore,
\begin{align*}
  \gamma\cdot \min\left\{x, y\right\} + (1-\gamma)\cdot xy &= \1\left\{x\leq y\right\}\cdot [\gamma\cdot x + (1-\gamma)\cdot xy] + \1\left\{y<x\right\}\cdot [\gamma\cdot y + (1-\gamma)\cdot xy] \\
  & \geq \frac{xy}{1-\gamma\cdot (1-x)\cdot (1-y)}.  
\end{align*}
This completes the proof.
\end{proof}

\subsection{Proofs for \Cref{sec:beliefs}}

\begin{proof}[\textbf{Proof of \Cref{thm:beta_posterior}}] 
Let $\widetilde{\I}$ be a random variable representing the set of admitted bidders.
For notation simplicity, we drop the parameters ``$n,F$'' in $\beta(v_{-i}\mid \tilde{\I}=\I,v_i;n,F)$.  
Note that, there is no need to condition separately on the admitted number $n$, as this information is already captured by the admission set $\I$.
By Bayes' formula, we have
\begin{align*}
 \beta(v_{-i}\mid\widetilde{\I}= \I,v_i)  &= \frac{\Pr\left(\widetilde{\I}=\I\mid v_{-i}, v_i\right)\beta\left(v_{-i} \mid v_i\right)}{\int_{[0,1]^{n-1}}\Pr\left(\widetilde{\I}=\I\mid v_{-i}, v_i\right)\beta\left(v_{-i} \mid v_i\right)dv_{-i}} \\
 &= \frac{\Pr\left\{
\max_{j\in \barI} S_j \leq \min_{k\in \I}S_k
\mid v\right\}\beta\left(v_{-i} \mid v_i\right)}{\int_{[0,1]^{n-1}}\Pr\left\{
\max_{j\in \barI} S_j \leq \min_{k\in \I}S_k
\mid v\right\}\beta\left(v_{-i} \mid v_i\right)dv_{-i}}\\
&:=\underbrace{\kappa(v_i;n,F)}_{\textrm{normalizing term}}
\cdot 
\underbrace{\psi(v;n,F)}_{\textrm{admission probability}}
 \cdot 
 \underbrace{\prod_{j\in \I_{-i}} f_1(v_j)}_{\textrm{prior}},
\end{align*}
where $\beta\left(v_{-i} \mid v_i\right)$ is the player $i$'s belief of the other $n-1$ bidders' types conditional on bidder $i$' type being $v_i$
(but \textit{not} conditional on the admission set), $\psi(v;n,F):=  \Pr\left\{
\max_{j\in \barI} S_j \leq \min_{k\in \I}S_k
\mid v\right\}$ is the admission probability, and 
$$\kappa(v_i;n,F) := \frac{1}{\int_{[0,1]^{n-1}}\Pr\left\{
\max_{j\in \barI} S_j \leq \min_{k\in \I}S_k
\mid v\right\}\beta\left(v_{-i} \mid v_i\right)dv_{-i}}$$ 
is the normalizing term. 
Notice that the formula \(\beta(v_{-i} \mid v_i)\) does \emph{not} condition on the event \(\widetilde{\I} = \I\). 
Hence, the bidder \(i\) does \emph{not} know the whether these bidders with valuations $v_{-i}$ are admitted or not, and thus \(\beta(v_{-i} \mid v_i)\) is simply the prior, i.e., \(\prod_{j \in \I_{-i}} f_1(v_j)\).
We acknowledge a slight notational abuse here: previously, we defined \(v_{-i} = (v_j)_{j \neq i,\, j \in \I}\). However, in the term \(\beta(v_{-i} \mid v_i)\), there is no conditioning on the admission set \(\I\). Consequently, one should treat \(v_{-i}\) in \(\beta(v_{-i} \mid v_i)\) as a vector indexed by some set of players, \textit{without} reference to whether those players are admitted or not.

\medskip
 We now show that 
\begin{align}\label{eq:psi_formula}
\psi(v;n,F)=
\int_0^1\prod_{k\in \I}[1-F_{2\mid 1}(x\mid v_k)] dF_2^{m-n}(x),
\end{align}
where $F_{2\mid 1}(\cdot\mid v_k)$ is the condition CDF of signal $S_i$. To see this, by Fubini theorem, we have 
\begin{align*}
    \psi(v;n,F) =&  \Pr\left\{
\max_{j\in \barI} S_j \leq \min_{k\in \I}S_k
\mid v\right\} 
=\int_0^1 \left( \int_x^1 d\Pr\left\{\min_{k\in \I}S_k
\leq y \mid v\right\}\right)d\Pr\left\{\max_{j\in \barI} S_j\leq x \mid v\right\}\\
=&\int_0^1
\Pr\left\{\min_{k\in \I}S_k
\geq x \mid v\right\}
d\Pr\left\{\max_{j\in \barI} S_j\leq x\right\}.
\end{align*}
The second equality holds since $S_{\barI}$, which are signals of eliminated players' valuations, are independent of admitted players' valuations $v$.
Observe that
\[
\Pr\left\{\min_{k\in \I}S_k
\geq x \mid v\right\} = \Pr\left\{\left(\min_{k\in \I}S_k
\mid v\right) \geq x\right\} = \prod_{k\in \I}\Pr\left\{S_k
 \geq x \mid v \right\}= \prod_{k\in \I}[1-F_{2\mid 1}(x\mid v_k)],
\] 
and 
 \begin{align*}
\Pr\left\{\max_{j\in \barI} S_j\leq x\right\} = \prod_{j\in \barI} \Pr\left\{S_j\leq x\right\} =  F_2^{m-n}(x),   
 \end{align*}
hence equation \eqref{eq:psi_formula} follows, as desired.
\end{proof}

\begin{proof}[\textbf{{Proof of \Cref{lem:kappa_vi}}}] We prove each part separately.

\medskip
\noindent\underline{\textit{Proof of (i).}}

\medskip
We have the following equalities:
\begin{align}
  \frac{1}{\kappa(v_i;n,F)} &= \int_{[0,1]^{n-1}}\psi(v;n, F) \prod_{j\in\I_{-i}} f_1(v_j) \, dv_{-i}\nonumber\\
  &\overset{(a)}{=} \int_{[0,1]^{n-1}} \left( \int_0^1\prod_{k\in \I}[1-F_{2\mid 1}(x\mid v_k)] \, d{F}_2^{m-n}(x) \right) \prod_{j\in\I_{-i}} f_1(v_j) \, dv_{-i}\nonumber\\
  & \overset{(b)}{=} \int_0^1 [1-F_{2\mid 1}(x\mid v_i)] \prod_{k\in\I_{-i}}\left(\int_0^1 [1-F_{2\mid 1}(x\mid v_k)] f_1(v_k) \, dv_k \right) d{F}_2^{m-n}(x)\label{eq: 1/kappa}.
\end{align}
Equality $(a)$ holds by \Cref{eq:psi_formula}.
Equality $(b)$ holds by Fubini's theorem.

Hence, if \( \PRD(S_i \mid V_i) \) holds, that is, if \( F_{2 \mid 1}(s_i \mid v_i) \) is non-increasing in \( v_i \in [0,1] \), then it follows that \( \kappa(v_i; n, F) \) is also non-increasing in \( v_i \in [0,1] \).

\medskip
\noindent\underline{\textit{Proof of (ii).}}

\medskip
Similar to the proof of \Cref{eq:copula formulation of F_{1|2}}, we have $F_{2 \mid 1}(s_i \mid v_i) = C_1\big( F_1(v_i),\, F_2(s_i) \big).$
Hence,
\begin{align}
\int_0^1 [\, 1 - F_{2 \mid 1}(x \mid v_k) \,]\, f_1(v_k) \, dv_k
&= 1 - \int_0^1 C_1\big( F_1(v_k),\, F_2(x) \big)\, dF_1(v_k)  = 1 - C\big( F_1(1),\, F_2(x) \big)\nonumber\\
& 
= 1 - F_2(x) \label{eq:A(1)}.
\end{align}
Combining this with the formula for \(\kappa(v_i; n, F)\) in part (i), we obtain
\begin{align*}
\frac{1}{\kappa(v_i; n, F)}
&= \int_0^1 [\, 1 - F_{2 \mid 1}(x \mid v_i) \,]
\prod_{k \in \I_{-i}} 
\left( \int_0^1 [\, 1 - F_{2 \mid 1}(x \mid v_k) \,]\, f_1(v_k) \, dv_k \right)
\, dF_2^{\, m-n}(x) \\
&\le 
\int_0^1 
\prod_{k \in \I_{-i}} 
\left( \int_0^1 [\, 1 - F_{2 \mid 1}(x \mid v_k) \,]\, f_1(v_k) \, dv_k \right)
\, dF_2^{\, m-n}(x) \nonumber\\
&= \int_0^1 [\, 1 - F_2(x) \,]^{n-1} 
\, dF_2^{\, m-n}(x) \nonumber\\
&= \frac{1}{\C^{m-1}_{n-1}}\nonumber.
\end{align*}
This implies that $\kappa(v_i; n, F) \; \ge \; \C^{m-1}_{n-1}.$
\end{proof}

\begin{proof}[\textbf{{Proof of \Cref{lem:joint_dis_g}}}]
We prove each part separately.

\medskip
\noindent\underline{\textit{Proof of (i).}}

\medskip
For this part, for notation simplicity, we drop ``$\I$'' in $g(v_{-i}\mid v_i;\I,n,F)$.
Suppose there exists a symmetric joint density $g(v;n,F)$ such that the conditional density $g(v_{-i}\mid v_i;n,F) = \beta(v_{-i}\mid \I, v_i;n,F)$. Then, by the symmetry assumption, all marginal densities have the same formula, denoted as $g^{\mathsf{mar}}(\cdot;n,F)$. 
For player $i$, we have
\begin{align*}
 g(v;n,F) =& g^{\mathsf{mar}}(v_i;n,F)\cdot \beta(v_{-i}\mid \I, v_i; n, F) =  g^{\mathsf{mar}}(v_i;n,F)\cdot
\kappa(v_i;n,F)
 \cdot  
 \psi(v;n,F) \cdot\left(\prod_{k\in \I_{-i}} f_1(v_k)\right),
\end{align*}
where the last equality holds by \Cref{thm:beta_posterior}.
Similarly, considering player $j$, we have
\begin{align*}
g(v;n,F) =  g^{\mathsf{mar}}(v_j;n,F)\cdot
\kappa(v_j;n,F)
 \cdot 
 \psi(v;n,F) \cdot 
\left(\prod_{k\in \I_{-j}} f_1(v_k)\right) .
\end{align*}
The above two equations imply that for \textit{any} $v_i,v_j\in[0,1]$,
\begin{align*}
\frac{g^{\mathsf{mar}}(v_i;n,F)\cdot
\kappa(v_i;n,F)}{f_1(v_i)} = \frac{g^{\mathsf{mar}}(v_j;n,F)\cdot
\kappa(v_j;n,F)}{f_1(v_j)}.
\end{align*}
This indicates that 
\begin{align}
\label{eq:margindist_kappa}
  g^{\mathsf{mar}}(v_i;n,F)\cdot
\kappa(v_i;n,F) = Cf_1(v_i)  
\end{align}
for some normalizing constant $C$. Therefore, the symmetric joint density $g(\cdot;n,F)$ must admit the form $g(v;n,F) = C \cdot \left( \prod_{i\in\I} f_1(v_i)\right)
 \cdot 
 \psi(v;n,F)$.
The constant $C$ satisfies
 \begin{align*}
    \frac{1}{C} =& {\int_{v\in [0,1]^n}
 \psi(v;n,F)\cdot \left( \prod_{i\in\I} f_1(v_i)\right)dv} \\
 =&{\int_0^1dF_1(v_i) \int_{[0,1]^{n-1}} \prod_{j\in\I_{-i}} f_1(v_j)
 \cdot 
\psi(v;n,F)dv_{-i}}\\
\overset{(a)}{=}& {\int_0^1 \frac{1}{\kappa(v_i;n,F)}dF_1(v_i)}\\
\overset{(b)}{=}&\int_0^1 \left(\int_0^1 [\, 1 - F_{2 \mid 1}(x \mid v_i) \,]
\prod_{k \in \I_{-i}} 
\left( \int_0^1 [\, 1 - F_{2 \mid 1}(x \mid v_k) \,]\, f_1(v_k) \, dv_k \right)
\, dF_2^{\, m-n}(x) \right)dF_1(v_i)\\
\overset{(c)}{=}& \int_0^1 [\, 1 - F_2(x) \,]^{n} 
\, dF_2^{\, m-n}(x) \\
=& \frac{1}{\C^m_n}.
 \end{align*}
The equality $(a)$ holds by the definition of $\kappa(v_i;n,F)$ in \Cref{thm:beta_posterior}. 
The equality $(b)$ holds by \Cref{eq: 1/kappa}.
The equality $(c)$ is by the Fubini theorem and \Cref{eq:A(1)}. This completes the proof of the first part.

\medskip

\noindent\underline{\textit{Proof of (ii).}}

\medskip
Since the bidders’ beliefs coincide by part~(i) and the number of participants is identical, the equilibrium bidding strategy is the same in both settings (more precisely, the set of equilibria is unchanged).  To complete the proof, we are left to show that the seller’s \textit{ex-ante} expected utility is also identical across the two settings.

With slight abuse of notation, let $\beta(\,\cdot \mid \tilde{\I}=\I,S=s; F)$ denote the seller's posterior belief (joint density) about $v\in [0,1]^n$, conditional on the admission set $\tilde{\I}=\I$ and the signal profile $s\in[0,1]^m$. Note that the seller can observe the signals. Thus, it suffices to verify that, for any measurable function $Z:[0,1]^n \to \mathbb{R}$,
\begin{align*}
\mathbb{E}_{S,\,\tilde{\I}}\left[\mathbb{E}_{V\sim \beta(\cdot \mid \tilde{\I},S; F)}[Z(V)]\right]
= \mathbb{E}_{\I}\left[\mathbb{E}_{V\sim g(\cdot;\tilde{\I},F)}[Z(V)]\right].
\end{align*}
Expectations are taken over both the signals and the admission set because the seller determines the admitted number before any realization occurs.
By Fubini's theorem, the left-hand side becomes
\begin{align*}
\mathbb{E}_{S,\,\tilde{\I}} \left[\mathbb{E}_{V\sim \beta(\cdot \mid \tilde{\I},S; F)}[Z(V)]  \right] = \int Z(v) \mathbb{E}_{S,\, \tilde{\I}}\left[\beta\left(v \mid \tilde{\I}, S; F\right)\right] dv,   
\end{align*}
and the right-hand side becomes
\begin{align*}
   \mathbb{E}_{\tilde{\I}}\left[\mathbb{E}_{V\sim g(\cdot;\tilde{\I},F)}[Z(V)]\right] = \int Z(v) \mathbb{E}_{\tilde{\I}} \left[g(v;\tilde{\I},F)\right]dv .
\end{align*}
Hence, it suffices to show $\mathbb{E}_{S,\, \tilde{\I}}\left[\beta\left(v \mid \tilde{\I}, S; F\right)\right]  = \mathbb{E}_{\tilde{\I}} \left[g(v;\tilde{\I},F)\right]$.

Let $\beta(\I,s; F)$ be the joint density of the $\tilde{\I}$ and $S$, and let $\beta(s \mid v; F)$ be the posterior density of $S$ conditional on the admitted bidders’ valuation vector $v$.  Recall that $s \in [0,1]^m$ stacks the signals of all $m$ potential bidders, whereas $v \in [0,1]^n$ records the valuations of the $n$ admitted bidders. Under the across-bidders independence assumption, we have
 \begin{align}
    \beta(s\mid v; F) = \prod_{i\in \I}\beta(s_i\mid v_i; F) \prod_{j\in \bar{\I}}f_1(s_j). \label{eq:beta_expanding} 
 \end{align}
Therefore,
\begin{align*}
\mathbb{E}_{S,\, \tilde{\I}}\left[\beta\left(v \mid \tilde{\I}, S; F\right)\right] 
=& \sum_{\I}\int \beta(v\mid \tilde{\I}=\I,S=s; F) \beta(\I, s; F)d s \\
\overset{(a)}{=}& \sum_{\I} \int \Pr\left(\widetilde{\I}=\I\mid S=s, v; F\right)\beta\left(s \mid v;F\right) \prod_{i\in \I}f_1(v_i)ds \\
\overset{(b)}{=}& \sum_{\I}\int \1\left\{
\max_{j\in \barI} s_j \leq \min_{k\in \I}s_k
\right\}\cdot \prod_{i\in \I}\beta(s_i\mid v_i; F) \prod_{j\in \bar{\I}}f_1(s_j) \cdot \prod_{i\in \I}f_1(v_i)ds\\
\overset{(c)}{=}& \sum_{\I} \Pr \left(\max_{j\in \barI} S_j \leq \min_{k\in \I}S_k
\mid v \right)\prod_{i\in \I}f_1(v_i).
\end{align*}
The equality $(a)$ holds by the Bayes formula.
The equality $(b)$ holds by \eqref{eq:beta_expanding}.
The equality $(c)$ holds by the definition of $\Pr \left(\max_{j\in \barI} S_j \leq \min_{i\in \I}S_i
\mid v \right)$.

By a similar analysis as above and the definition of $g(v;\tilde{\I},n,F)$ in \eqref{eq:joint_dist_g}, we have
\begin{align*}
\Pr\{\widetilde{\I}=\I\} &=  \int \Pr\left(\widetilde{\I}=\I\mid S=s, v;F\right)\beta\left(s \mid v; F\right) \prod_{i\in \I}f_1(v_i)ds\,dv \\
&=\int\Pr \left(\max_{j\in \barI} S_j \leq \min_{k\in \I}S_k
\mid v \right)\prod_{i\in \I}f_1(v_i) dv\\
&=\frac{1}{\C^m_n}.    
\end{align*}
Intuitively, this is because from the ex-ante perspective, the bidders are identical, and thus each has the same probability of being admitted.
Then, we conclude that
\begin{align*}
 \mathbb{E}_{\tilde{\I}} \left[g(v;\tilde{\I},F)\right] =& 
\sum_{\I} \C^m_n\cdot \Pr \left(\max_{j\in \barI} S_j \leq \min_{k\in \I}S_k
\mid v \right) \cdot \left(\prod_{i\in \I}f_1(v_i)\right)\cdot \Pr\{\widetilde{\I}=\I\} \\
=&  \sum_{\I}\Pr \left(\max_{j\in \barI} S_j \leq \min_{k\in \I}S_k
\mid v \right)\cdot \prod_{i\in \I}f_1(v_i).
\end{align*}
Combing the above analysis, $\mathbb{E}_{S,\, \tilde{\I}}\left[\beta\left(v \mid \tilde{\I}, S; F\right)\right]  = \mathbb{E}_{\tilde{\I}} \left[g(v;\tilde{\I},F)\right]$, as desired. 
\end{proof}

\begin{proof}[\textbf{Proof of \Cref{remark:affiliation}}]
In the case of a null predictor or when \( n = m \), the joint density coincides with the prior, that is,
\[
g(v; n, F) = \prod_{i \in \I} f_1(v_i).
\]
By \citet{milgrom_1982_auctiontheory_competitive_bidding}, the joint density in this setting is affiliated. 

Under the perfect predictor, we have $\psi(v; n, \PP) 
= F_1^{\, m-n} \big( \min_{i \in \I} v_i \big)$ by \Cref{eq:gamma_1_admissionprob},
which is log-supermodular in \( v \in [0,1]^n \) by definition. Hence, \(\psi(v; n, \PP)\) is affiliated. Combining this with the fact that the product of affiliated functions remains affiliated (see Theorem~1 in \citet{milgrom_1982_auctiontheory_competitive_bidding}), we conclude that the joint density \( g(v; n, F) \) in \eqref{eq:joint_dist_g} is affiliated.

\end{proof}

\begin{proof}[\textbf{Proof of \Cref{lem:FOSD_number}}] 
We prove each part separately.

\medskip
\noindent\underline{\textit{Proof of (i).}}
\medskip

Define 
\begin{align}
\label{eq:def_A_basemodel}
   A(x,t; F) := \int_0^x [1-F_{2 \mid 1}(t\mid v_k)] dF_1(v_k). 
\end{align}
Then, similar to the derivation of \Cref{eq:A(1)}, we can express this as 
\begin{align}\label{eq:def A}
A(x,t; F) = \int_0^x [1-F_{2 \mid 1}(t\mid v_k)] dF_1(v_k) = F_1(x) - C(F_1(x), F_2(t)).
\end{align}
We have
\begin{align}
G^{\mar}(x;n,F)&=
\int_{[0,x]\times [0,1]^{n-1}}g(v;n,F)dv\notag\\
 &\overset{(a)}{=}  \C^m_n \cdot \int_{[0,x]\times [0,1]^{n-1}}\psi(v;n,F)
 \cdot \prod_{j\in\I} f_1(v_j)  dv \notag\\
 &\overset{(b)}{=} \C^m_n \cdot \int_{[0,x]\times [0,1]^{n-1}} \left(\int_0^1\prod_{k\in \I}[1-F_{2\mid 1}(t\mid v_k)] dF_2^{m-n}(t)\right)
 \cdot \prod_{k\in\I} f_1(v_k)  dv \notag\\
& \overset{(c)}{=} \C^m_n \cdot \int_0^1\left(\int_0^x [1-F_{2\mid 1}(t\mid v_k)] dF_1(v_k)\right)\cdot \left(\int_0^1 [1-F_{2\mid 1}(t\mid v_k)] dF_1(v_k)\right)^{n-1}dF_2^{m-n}(t)\notag\\
& \overset{(d)}{=} \C^m_n \cdot \int_0^1 A(x,t;F) \cdot A^{n-1}(1, x,F) dF_2^{m-n}(t)\notag\\
 &\overset{(e)}{=} \C^m_n \cdot \int_0^1 A(x,t; F)\left[1-F_2(t)\right]^{n-1} dF_2^{m-n}(t)\label{eq:gmarginal}.
\end{align}
Here, equality $(a)$ follows from the definition of $g(v; n, F)$ in \eqref{eq:joint_dist_g}. Equality $(b)$ holds by \Cref{thm:beta_posterior} regarding the admission probability. Equality $(c)$ is a result of Fubini's theorem. Equality $(d)$ uses the definition of $A(x, t; F)$, and equality $(e)$ is true because $A(1, t; F) = 1 - F_2(t)$ by \Cref{eq:A(1)}.

Define the partial derivative $A_2(x, t; F) := \frac{\partial}{\partial t} A(x, t; F)$. It follows that $A_2(x, t; F) = -f_2(t) \cdot C_2(F_1(x), F_2(t))$.
We have
\begin{align*}
&\quad ~G^{\mar}(x;n+1,F)\\
&= \C^m_{n+1} \cdot \int_0^1 A(x, t; F)\left[1-F_2(t)\right]^{n} dF_2^{m-n-1}(t)\\
&\overset{(a)}{=}-\C^m_{n+1} \cdot\int_0^1F_2^{m-n-1}(t)\ d \left(A(x, t; F)\left[1-F_2(t)\right]^{n}\right)\\
&= -\C^m_{n+1}\cdot \int_0^1F_2^{m-n-1}(t) \cdot \left[1-F_2(t)\right]^{n} A_2(x, t; F)dt + \C^m_{n+1} \cdot n\cdot \int_0^1F_2^{m-n-1}(t) \cdot A(x, t; F)\left[1-F_2(t)\right]^{n-1} dF_2(t)\\
&= \C^m_{n+1} \cdot \int_0^1 F_2^{m-n-1}(t)\left(1-F_2(t)\right]^{n-1}\left[nA(x, t; F) - [1-F_2(t)]A_2(x, t; F)\right] dF_2(t)\\
&\overset{(b)}{\geq}  (n+1)\cdot \C^m_{n+1} \cdot \int_0^1 F_2^{m-n-1}(t)\left[1-F_2(t)\right]^{n-1}A(x,t;F)dt\\
&= \C^m_{n} \cdot \int_0^1 A(x, t; F)\left[1-F_2(t)\right]^{n-1}dF_2^{m-n}(t)\\
&=G^{\mar}(x;n,F),
\end{align*}
The equality $(a)$ holds using integration by parts. 
The inequality $(b)$ holds since 
\begin{align*}
 & \textrm{$(b)$ holds}\\
\Longleftrightarrow \quad & 
- [1-F_2(t)]A_2(x, t; F) \geq f_2(t)\cdot A(x,t;F)\\
\Longleftrightarrow\quad & f_2(t)\cdot C_2(F_1(x), F_2(t))\cdot [1-F_2(t)] \geq f_2(t)\cdot [F_1(x) - C(F_1(x), F_2(t))]
\\
\Longleftrightarrow\quad & d(t):=-C_2(F_1(x), F_2(t))\cdot [1-F_2(t)] + F_1(x) - C(F_1(x), F_2(t))\leq 0
\\
\Longleftarrow\quad &
d(1) = F_1(x) - C(F_1(x), F_2(1)) = 0, \, d'(x) = -C_{22}(F_1(x), F_2(t))\cdot f_2(t)\cdot [1-F_2(t)]\geq 0\\
\Longleftrightarrow\quad &
C_{22}(F_1(x), F_2(t)) \leq 0 \tag{\Cref{eq:copula formulation of F_{1|2}}}&\\
\Longleftarrow\quad &
\textrm{$F_{1\mid 2}(x\mid t)$ is non-increasing with $t\in[0,1]$ for all $x\in[0,1]$} \tag{$\PRD(V_i\mid S_i)$}
\end{align*}

\medskip

\noindent\underline{\textit{Proof of (ii).}}

\medskip

Firstly, we will show that
\[
\mathsf{LHS}:=\sum_{i=0}^{k-1}\C^{n+1}_i \cdot \int_{[0,x]^{n+1-i}\times [x,1]^{i}}g(v;n+1, F)dv \leq \sum_{i=0}^{k-1}\C^{n}_i \cdot \int_{[0,x]^{n-i}\times [x,1]^{i}}g(v;n, F)dv:=\mathsf{RHS},
\]
where the left-hand side represents the CDF of the $k^{\mathsf{th}}$ largest valuation sampled from $g(v;n+1, F)$, and the right-hand side reprents the one sampled from $g(v;n, F)$.
We have
\begin{align}
&\C^n_i \cdot \int_{[0,x]^{n-i}\times [x,1]^i}g(v;n, F)dv \notag\\ 
\overset{(a)}{=}& \C^n_i \cdot \C^{m}_n \cdot \int_{[0,x]^{n-i}\times [x,1]^i} \psi(v;n, F)
 \cdot \prod_{k\in\I} f_1(v_k)  dv \notag\\
\overset{(b)}{=}& \C^n_i \cdot \C^{m}_n \cdot \int_{[0,x]^{n-i}\times [x,1]^i} \left(\int_0^1\prod_{k\in \I}[1-F_{2\mid 1}(t \mid v_k)] d{F}_2^{m-n}(t)\right)
 \cdot \prod_{k\in\I} f_1(v_k)  dv \notag\\
\overset{(c)}{=}& \C^n_i \cdot \C^{m}_n \cdot \int_0^1\left(\int_0^x [1-F_{2\mid 1}(t \mid v_k)] dF_1(v_k)\right)^{n-i}\cdot \left(\int_x^1 [1-F_{2\mid 1}(t \mid v_k)] dF_1(v_k)\right)^{i}d{F}_2^{m-n}(t) \notag\\
 \overset{(d)}{=}& \C^n_i \cdot \C^{m}_n \cdot \int_0^1 A^{n-i}(x, t; F)\left[A(1, t; F) - A(x, t; F)\right]^i d{F}_2^{m-n}(t) \label{eq:klargest}.
\end{align}
The equality $(a)$ holds by \eqref{eq:joint_dist_g}, and $(b)$ holds by \eqref{eq:admission_prob}.
The equality $(c)$ holds by Fubini's theorem, and $(d)$ holds by the definition of $A(x, t; F)$ in \eqref{eq:def A}.

Similarly, we have
\begin{align*}
&\C^{n+1}_i \cdot \int_{[0,x]^{n+1-i}\times [x,1]^{i}}g(v;n+1; F)dv \\
=& \C^{n+1}_i \cdot \C^{m}_{n+1} \cdot \int_0^1 A^{n+1-i}(x, t; F)\left[A(1, t; F) - A(x, t; F)\right]^i d{F}_2^{m-n-1}(t)\\
\overset{(a)}{=}&-\C^{n+1}_i \cdot \C^{m}_{n+1} \cdot\int_0^1{F}_2^{m-n-1}(t)\ d \left(A^{n+1-i}(x, t; F)\left[A(1, t; F) - A(x, t; F)\right]^i\right)\\
=& -\C^{n+1}_i \cdot \C^{m}_{n+1} \cdot (n+1-i)\cdot \int_0^1{F}_2^{m-n-1}(t) \cdot A^{n-i}(x, t; F)\left[A(1, t; F) - A(x, t; F)\right]^i A_2(x, t; F)dt \\
& - \C^{n+1}_i \cdot \C^{m}_{n+1} \cdot i\cdot \int_0^1{F}_2^{m-n-1}(t) \cdot A^{n+1-i}(x, t; F)\left[A(1, t; F) - A(x, t; F)\right]^{i-1} \left[A_2(1, t; F) - A_2(x, t; F)\right]dt\\
=& -\frac{n+1}{m-n}\cdot \C^{n}_i \cdot \C^{m}_{n+1} \cdot \int_0^1 A^{n-i}(x, t; F)\left[A(1, t; F) - A(x, t; F)\right]^i A_2(x,t;F)\frac{1}{{f}_2(t)} d{F}_2^{m-n}(t)\\
&-\frac{n+1}{m-n}\cdot \C^{n}_{i-1} \cdot \C^{m}_{n+1} \cdot \int_0^1 A^{n+1-i}(x, t; F)\left[A(1, t; F) - A(x, t; F)\right]^{i-1}\left[A_2(1, t; F) - A_2(x, t; F)\right] \frac{1}{{f}_2(t)}d{F}_2^{m-n}(t).
\end{align*}
The equality $(a)$ holds using integration by parts, and the fact that $A(x, 1; F)=0$ for all $x\in [0,1]$.
Therefore, 
\begin{align*}
&\mathsf{LHS}\\=&\sum_{i=0}^{k-1}\C^{n+1}_i \cdot \int_{[0,x]^{n+1-i}\times [x,1]^{i}}g(v;n+1; F)dv \\
=& -\frac{n+1}{m-n}\cdot\sum_{i=0}^{k-1}  \C^{n}_i \cdot \C^{m}_{n+1} \cdot \int_0^1 A^{n-i}(x, t; F)\left[A(1, t;F) - A(x, t; F)\right]^iA_2(x,t;F) \frac{1}{{f}_2(t)}d{F}_2^{m-n}(t) \\
& -\frac{n+1}{m-n}\cdot \sum_{i=1}^{k-1} \C^{n}_{i-1} \cdot \C^{m}_{n+1} \cdot \int_0^1 \frac{A^{n+1-i}(x, t; F)}{{f}_2(t)}\left[A(1,t;F) - A(x, t; F)\right]^{i-1}\left[A_2(1,t;F) - A_2(x, t; F)\right] d{F}_2^{m-n}(t)\\
\overset{(a)}{=}& -\frac{n+1}{m-n}\cdot\sum_{i=0}^{k-1}  \C^{n}_i \cdot \C^{m}_{n+1} \cdot \int_0^1 A^{n-i}(x, t; F)\left[A(1,t;F) - A(x, t; F)\right]^iA_2(x,t;F) \frac{1}{{f}_2(t)}d{F}_2^{m-n}(t) \\
& -\frac{n+1}{m-n}\cdot \sum_{i=0}^{k-2} \C^{n}_{i} \cdot \C^{m}_{n+1} \cdot \int_0^1 A^{n-i}(x, t; F)\left[A(1,t;F) - A(x, t; F)\right]^{i}\left[A_2(1,t;F) - A_2(x, t; F)\right] \frac{1}{{f}_2(t)}d{F}_2^{m-n}(t)\\
\overset{(b)}{=}& -\frac{n+1}{m-n}\cdot\sum_{i=0}^{k-2}  \C^{n}_i \cdot \C^{m}_{n+1} \cdot \int_0^1 A^{n-i}(x, t; F)\left[A(1,t;F) - A(x, t; F)\right]^iA_2(1,t;F) \frac{1}{{f}_2(t)}d{F}_2^{m-n}(t) \\
&-\frac{n+1}{m-n}\cdot\C^{n}_{k-1} \cdot \C^{m}_{n+1} \cdot \int_0^1 A^{n+1-k}(x, t; F)\left[A(1,t;F) - A(x, t; F)\right]^{k-1}A_2(x,t;F)\frac{1}{{f}_2(t)} d{F}_2^{m-n}(t)\\
\overset{(c)}{\leq} & -\frac{n+1}{m-n}\cdot\sum_{i=0}^{k-1}  \C^{n}_i \cdot \C^{m}_{n+1} \cdot \int_0^1 A^{n-i}(x, t; F)\left[A(1,t;F) - A(x, t; F)\right]^iA_2(1,t;F) \frac{1}{{f}_2(t)}{F}_2^{m-n}(t)\\
\overset{(d)}{=}&\sum_{i=0}^{k-1}\C^n_i \cdot \C^{m}_n \cdot \int_0^1 A^{n-i}(x, t; F)\left[A(1,t;F) - A(x, t; F)\right]^i d{F}_2^{m-n}(t)\\
=& \sum_{i=0}^{k-1}\C^{n}_i \cdot \int_{[0,x]^{n-i}\times [x,1]^{i}}g(v;n,F)dv\\
=& \mathsf{RHS}.
\end{align*}
The equality $(a)$ holds by shifting the index, and the equality $(b)$ holds by collecting the common terms. 
The inequality $(c)$ uses the fact that $A_2(1,t;F) \leq A_2(x,t;F)\leq 0$ for $x\in[0,1]$, since $A_2(x, t; F) = -f_2(t) \cdot C_2(F_1(x), F_2(t)) = -f_2(t)\cdot F_{1\mid 2}(x\mid t)$.
Equality $(d)$ holds since $A_2(1,t;F) = -{f}_2(t)$ and
\begin{align*}
 \frac{n+1}{m-n}\cdot \C^{m}_{n+1}  = \frac{n+1}{m-n}\cdot \C^m_{n+1}   =\C^m_n = \C^{m}_n.
\end{align*}
This completes the proof of the first part. 

\medskip
\noindent \underline{\textit{Proof of (iii) under the perfect predictor.}}
\medskip

Observe that for the perfect predictor,
\[
A(x,t; \PP) = F_1(x) - C(F_1(x), F_2(t))=
\max\left\{F_1(x) - F_2(t),~0\right\}.
\]
Hence,
\begin{align*}
&\C^n_i \cdot \int_{[0,x]^{n-i}\times [x,1]^i}g(v;n, \PP)dv\\ \overset{(a)}{=}& \C^n_i \cdot \C^m_n \cdot \int_{[0,x]^{n-i}\times [x,1]^i} \psi(v;n, \PP)
 \cdot \prod_{j\in\I} f_1(v_j)  dv \\
 \overset{(b)}{=}& \C^n_i \cdot \C^m_n \cdot \int_{[0,x]^{n-i}\times [x,1]^i} \left(\int_0^1\prod_{k\in \I}[1-F_{2\mid1}(t \mid v_k)] dF_2^{m-n}(t)\right)
 \cdot \prod_{k\in\I} f_1(v_k)  dv \\
 \overset{(c)}{=}& \C^n_i \cdot \C^m_n \cdot \int_0^1A^{n-i}(x, t;\PP)\cdot [A(1, t;\PP) - A(x, t;\PP)]^{i}dF_2^{m-n}(t)\\
 =& \C^n_i \cdot \C^m_n \cdot \int_0^{F^{-1}_2(F_1(x))} [F_1(x)-F_2(t)]^{n-i}\left[1 - F_1(x)\right]^i dF_2^{m-n}(t)\\
 =& \C^n_i \cdot \C^m_n \cdot \int_0^{F_1(x)} [F_1(x)-t]^{n-i}\left[1 - F_1(x)\right]^i dt^{m-n}\\
 \overset{(d)}{=}& \C^n_i \cdot \C^m_n \cdot \int_0^{1} F_1^{n-i}(x)\cdot (1-u)^{n-i}\cdot \left[1 - F_1(x)\right]^i\cdot F_1^{m-n}(x) du^{m-n}\\
 =& \C^n_i \cdot \C^m_n \cdot  \frac{1}{\C^{m-i}_{n-i}} \cdot F_1^{m-i}(x)\cdot  [1-F_1(x)]^i\\
 =& \C^m_i \cdot F_1^{m-i}(x)\cdot  [1-F_1(x)]^i.
\end{align*}
Here, the equality $(a)$ follows from the definition of \( g(v; n, \PP) \), and the equality $(b)$ holds by \Cref{eq:admission_prob}. The equality $(c)$ follows from Fubini's theorem, and the equality $(d)$ holds by changing variables $t=u\cdot F_1(x)$.

Therefore,
\[
 \sum_{i=0}^{k-1}\C^{n}_i \cdot \int_{[0,x]^{n-i}\times [x,1]^{i}}g(v;n, \PP)dv = \sum_{i=0}^{k-1}\C^{m}_i \cdot \int_{[0,x]^{n-i}\times [x,1]^{i}}\prod_{k\in\I^0} f_1(v_k)dv,
\]
as desired.

\medskip
\noindent \underline{\textit{Proof of (iii).}} We prove each part separately.
\medskip

\noindent \underline{\textit{Proof of marginal distribution.}} 
\medskip

By \Cref{eq:def A}, we have 
\[
A(x,t; F) = F_1(x) - C\big( F_1(x), F_2(t) \big)
= F_1(x) - F(x,t),
\]
and therefore, if \( F \pqdorder \tilde{F} \), then,  combining this with \Cref{eq:gmarginal}, we obtain 
\[
G^{\mar}(x; n, F) 
\le G^{\mar}(x; n, \tilde{F}),
\]
as desired.

\noindent \underline{\textit{Proof of $k$-th largest valuation.}}
\medskip

We show that whenever \( F \pqdorder \tilde{F} \),
\begin{align}\label{eq:selargest}
    \sum_{i=0}^{k-1}\C^n_i \cdot \int_{[0,x]^{n-i}\times [x,1]^i}g(v;n, F)\;dv \;
\overset{(a)}{\le}\;
     \sum_{i=0}^{k-1}\C^n_i \cdot \int_{[0,x]^{n-i}\times [x,1]^i}g(v;n, \tilde{F})\;dv.
\end{align}
By \Cref{eq:klargest} and the fact that \( A(1, t; F) = 1 - F_2(t) \), we have
\begin{align*}
\sum_{i=0}^{k-1}\C^n_i \cdot \int_{[0,x]^{n-i}\times [x,1]^i}g(v;n, F)\;dv =\C^{m}_n \cdot \sum_{i=0}^{k-1}\C^n_i \cdot  \int_0^1 A^{n-i}(x, t; F)\left[1 - F_2(t) - A(x, t; F)\right]^i d{F}_2^{m-n}(t).
\end{align*}
Define 
\[
e(a) := \sum_{i=0}^{k-1}\C^n_i \cdot a^{n-i}[1-F_2(t) - a]^i, 
\quad a \in [0,\, 1 - F_2(t)].
\]
Then 
\begin{align*}
e'(a) &= \sum_{i=0}^{k-1}(n-i)\C^n_i \cdot a^{n-i-1}[1-F_2(t) - a]^i  - \sum_{i=1}^{k-1}i\C^n_i \cdot a^{n-i}[1-F_2(t) - a]^{i-1}   \\
&= \sum_{i=0}^{k-1}(n-i)\C^n_i \cdot a^{n-i-1}[1-F_2(t) - a]^i  - \sum_{i=0}^{k-2}(i+1)\C^n_{i+1} \cdot a^{n-i-1}[1-F_2(t) - a]^{i}  \\
&\overset{(a)}{=}(n+1-k)\C^n_{k-1}a^{n-k}[1-F_2(t) - a]^{k-1}\\
&\;\ge 0,
\end{align*}
where (a) holds by the fact that $(i+1)\C^n_{i+1} = (n-i)\C^n_i$. Hence we conclude that \( e(a) \) is increasing on \([0,\, 1 - F_2(t)]\).  

Therefore, when \( F \pqdorder \tilde{F} \), we have $A(x, t; F) \le A(x, t; \tilde{F})$, and 
\[
 \sum_{i=0}^{k-1}\C^n_i \cdot A(x, t; F)^{n-i}[1-F_2(t) - A(x, t; F)]^i 
\;\le\;
\sum_{i=0}^{k-1}\C^n_i \cdot A(x, t; \tilde{F})^{n-i}[1-F_2(t) - A(x, t; \tilde{F})]^i.
\]
Integrating over \( t \) under the measure $F^{m-n}_2$ then shows that \Cref{eq:selargest} holds, as claimed.

\end{proof}

\begin{proof}[\textbf{{Proof of \Cref{prop:stochastic_domiance_gamma_1}}}] 
We prove each part separately.

\medskip
\noindent \underline{\textit{Proof of (i).}}
\medskip

Under the perfect predictor, by \Cref{remark:affiliation}, we know that 
the joint distribution $g(v; n, \PP)$ is affiliated in $v \in [0,1]^n$. Since both conditioning and marginalization preserve affiliation \citep{karlin_1980_MTP2}, the conditional density $g(v_{-i} \mid v_i; n, \PP)$ is also affiliated in $v \in [0,1]^{n}$, and the conditional marginal density $\beta^{\mar}(x \mid \I, v_i; n, \PP)$---derived from $g(v_{-i} \mid v_i; n, \PP)$---is affiliated in $(x, v_i) \in [0,1]^2$. Then, the first $\fosd$ follows by Proposition 3.1 in \citet{castro_2007_affiliation_positive_dependence}. 

\medskip
\noindent \underline{\textit{Proof of (ii).}}
\medskip

This result can be directly verified by the definition of first-order stochastic dominance. 
\end{proof}

\subsection{Proofs for \Cref{sec:firstprice_second_auctions}}

\begin{proof}[\textbf{Proof of \Cref{thm:opt_admittednumber_secondprice}}]

We prove each case separately.

\medskip
\noindent\underline{\textit{Proof of (i).}}
\medskip

By \Cref{prop:equilibrium_secondprice}, bidding truthfully is an equilibrium. Consequently, the seller’s revenue in a second-price auction equals the expected second-highest valuation sampled from $g(v;n,F)$. By \Cref{lem:FOSD_number}(ii), this second-highest valuation is increasing in the number of admitted bidders in terms of first-order stochastic dominance. Therefore, admitting all bidders is optimal.

\medskip

\noindent\underline{\textit{Proof of (ii).}}
\medskip

When \( n = m \), trivially, admitting all players results in the same expected revenue for any predictor. Furthermore, for any given admitted number, by \Cref{lem:FOSD_number}(iii), the second-highest valuation increases with the prediction accuracy. In other words, the revenue loss due to admitting fewer bidders is (weakly) decreasing in the prediction accuracy.

Under a perfect predictor, for any admitted number $n\in[2,m]$, by \Cref{lem:FOSD_number}(ii), the distribution of the second largest valuation is independent of $n$. Therefore, admitting any number of bidders yields the same expected revenue in this case.
\end{proof}

\begin{proof}[\textbf{Proof of \Cref{thm:SSM_firstprice}}]
The proof includes four steps:
\begin{enumerate}
    \item  derive the unique SSM equilibrium $\sigma^{\mathsf{FP}}(v_i;n,F)$ if it exists;
    \item show that $\sigma^{\mathsf{FP}}(v_i;n,F)$ is indeed an equilibrium under the condition in \Cref{thm:SSM_firstprice}(i);
    \item show that under the perfect predictor, the condition in \Cref{thm:SSM_firstprice}(i) holds, and the equilibrium strategy does not change with the admitted number.
\end{enumerate}

\medskip

\noindent\textit{\underline{Step 1: SSM equilibrium if it exists.}}
\medskip

When all other $n-1$ admitted bidders follow a strictly increasing strategy $\sigma:[0,1]\to \mathbb{R}_+$, the best response of the bidder $i$ is
\begin{align}
\label{eq:utility function first}
    \argmax_{\tilde{v}_i\in [0,1]}~  \Pr\left\{\sigma({V}_j)\leq \sigma(\tilde{v}_i),\forall j\in \I_{-i}\right\} \cdot [v_i - \sigma(\tilde{v}_i)],
\end{align}
where $\Pr\left\{\sigma({V}_j)\leq \sigma(\tilde{v}_i),\forall j\in \I_{-i}\right\}$ is the winning probability under bidder $i$'s belief when she bids $\sigma(\tilde{v}_i)$.  
Observe that 
\begin{align*}
   \Pr\left\{\sigma({V}_j)\leq \sigma(\tilde{v}_i),\forall j\in \I_{-i}\right\} =\Pr\left\{V_j\leq \tilde{v}_i,\forall j\in \I_{-i}\right\} = H(\tilde{v}_i\mid v_i;n,F),
\end{align*}
where the first equality holds since $\sigma(\cdot)$ is strictly increasing, and the last equality holds because $H(\cdot \mid v_i;n,F)$ is the CDF of $\max_{j\in \I_{-i}}V_j$ by \Cref{eq:H_CDF}. Hence, \Cref{eq:utility function first} is equivalent to 
\begin{equation}
\label{eq:utility function first used}
\max_{\tilde{v}_i\in[0,1]}~U^{\mathsf{FP}}(\tilde{v}_i,v_i;\sigma(\cdot)):=  H(\tilde{v}_i\mid v_i; n,F)\cdot [v_i - \sigma(\tilde{v}_i)].
\end{equation}
If there exists an SSM equilibrium strategy, it must satisfy the first-order condition:
\begin{align*}
    &\frac{\partial U^{\mathsf{FP}}(\tilde{v}_i,v_i;\sigma(\cdot))}{\partial \tilde{v}_i}\bigg |_{\tilde{v}_i = v_i}
    = \left(v_ih(\tilde{v}_i\mid v_i; n,F) - H(\tilde{v}_i\mid v_i; n,F) \frac{d \sigma(\tilde{v}_i)}{d\tilde{v}_i} - h(\tilde{v}_i\mid v_i; n,F)\sigma(\tilde{v}_i)\right)\bigg |_{\tilde{v}_i = v_i} = 0\\
    \Longleftrightarrow\quad & \frac{d \sigma(\tilde{v}_i)}{d\tilde{v}_i} + \frac{h({v}_i\mid v_i; n,F)}{H({v}_i\mid v_i; n,F)}\sigma(\tilde{v}_i)-v_i\frac{h(v_i\mid v_i;n,F)}{H(v_i\mid v_i;n,F)} = 0.
\end{align*}
Define
\begin{align}
\label{eq:L}
    L(x\mid v_i;n,F)
 :=\exp\left(-\int_{x}^{v_i} \frac{h(t\mid t;n,F)}{H(t\mid t;n,F)}dt\right)
 =\exp\left(-\int_{x}^{v_i} 
 \RHR(t;n,F)dt\right)
 ,~\forall x\in [0,v_i].
\end{align}
The above differential equation yields the unique SSM equilibrium strategy if it exists as follows:
\begin{align*}
 \sigma(v_i)=&  \exp\left(-\int_{0}^{v_i} \frac{h(t\mid t;n,F)}{H(t\mid t;n,F)}dt\right) \int_{0}^{v_i}x \frac{h(x\mid x;n,F)}{H(x\mid x;n,F)}\exp\left(\int_{0}^{x} \frac{h(t\mid t;n,F)}{H(t\mid t;n,F)}dt\right) dx \\
=& \int_{0}^{v_i}x \frac{h(x\mid x;n,F)}{H(x\mid x;n,F)} \exp\left(-\int_{x}^{v_i} \frac{h(t\mid t;n,F)}{H(t\mid t;n,F)}dt\right) dx\\
\overset{(a)}{=}&\int_{0}^{v_i}x  dL(x\mid v_i;n,F) \\
=& v_i L(v_i\mid v_i;n,F) - \int_{0}^{v_i}L(x\mid v_i;n,F) dx\\
=& v_i - \int_{0}^{v_i}L(x\mid v_i;n,F) dx := \sigma^{\mathsf{FP}}(v_i;n,F).
\end{align*}
The equality $(a)$ holds by the definition of $L(x\mid v_i;n,F)$ in \Cref{eq:L}.
One can easily check that $\sigma^{\FP}(v_i;n,F)$ is indeed strictly increasing in $v_i\in[0,1]$.

 \medskip

\noindent\textit{\underline{Step 2: $\sigma^{\mathsf{FP}}(v_i;n,F)$ is indeed an equilibrium under the condition in \Cref{thm:SSM_firstprice}(i)}.}

\medskip

To show this, it's equivalent to show that
\begin{align*}
    U^{\mathsf{FP}}(\tilde{v}_i,v_i;\sigma^{\mathsf{FP}}(\cdot;n,F))= & H(\tilde{v}_i\mid v_i; n, F)  \left( v_i - \sigma^{\mathsf{AP}}(\tilde{v}_i;n, F)\right)\\
    = & H(\tilde{v}_i\mid v_i; n, F)\left( v_i - \tilde{v}_i + \int_0^{\tilde{v}_i} L(x\mid \tilde{v}_i;n, F)dx\right)
\end{align*}
indeed obtains the maximum at $\tilde{v}_i = v_i$. It's sufficient to show that $U^{\mathsf{FP}}(\tilde{v}_i,v_i;\sigma^{\mathsf{AP}}(\cdot;n,F))$ non-decreases with $\tilde{v}_i \in [0, v_i]$, and non-increases with $\tilde{v}_i \in [v_i, 1]$.
Observe that
\begin{align}
&\frac{\partial U^{\mathsf{FP}}(\tilde{v}_i,v_i;\sigma^{\mathsf{FP}}(\cdot;n,F))}{\partial \tilde{v}_i} \notag\\ = & h(\tilde{v}_i\mid v_i;n,F)  \left(v_i - \tilde{v}_i  +\int_0^{\tilde{v}_i}L(x\mid \tilde{v}_i;n,F)dx\right) \notag\\
 &+ H(\tilde{v}_i\mid v_i;n,F)  \left(-1+L(\tilde{v}_i\mid \tilde{v}_i;n,F)
 + \int_{0}^{\tilde{v}_i}\frac{\partial L(x\mid \tilde{v}_i;n,F)}{\partial \tilde{v}_i}dx
 \right)
  \notag\\
=&  h(\tilde{v}_i\mid v_i;n,F)  \left(v_i - \tilde{v}_i  +\int_0^{\tilde{v}_i}L(x\mid \tilde{v}_i;n,F)dx\right)
 + H(\tilde{v}_i\mid v_i;n,F)   \int_{0}^{\tilde{v}_i}\frac{\partial L(x\mid \tilde{v}_i;n,F)}{\partial \tilde{v}_i}dx \notag\\
 = & h(\tilde{v}_i\mid v_i;n,F)  \left(v_i - \tilde{v}_i\right)
 +
  h(\tilde{v}_i\mid v_i;n,F)  \int_0^{\tilde{v}_i}L(x\mid \tilde{v}_i;n,F)dx \notag\\
  &
  -  H(\tilde{v}_i\mid v_i;n,F)  \frac{h(\tilde{v}_i\mid \tilde{v}_i;n,F)}{H(\tilde{v}_i\mid \tilde{v}_i;n,F)}  \int_0^{\tilde{v}_i}L(x\mid \tilde{v}_i;n,F)dx \notag\\
  = & h(\tilde{v}_i\mid v_i;n,F)  \left(v_i - \tilde{v}_i\right)
 + 
   \left(
 h(\tilde{v}_i\mid v_i;n,F)  - H(\tilde{v}_i\mid v_i;n,F)  \frac{h(\tilde{v}_i\mid \tilde{v}_i;n,F)}{H(\tilde{v}_i\mid \tilde{v}_i;n,F)}
 \right)\int_0^{\tilde{v}_i}L(x\mid \tilde{v}_i;n,F)dx \notag\\
 =&  h(\tilde{v}_i\mid v_i;n,F) \left[v_i - \tilde{v}_i\right]
 +H(\tilde{v}_i\mid v_i;n,F) 
 \left[\int_0^{\tilde{v}_i}L(x\mid \tilde{v}_i;n)dx\right]
\left[
 \frac{h(\tilde{v}_i\mid v_i;n,F)}{H(\tilde{v}_i\mid v_i;n,F)}  - \frac{h(\tilde{v}_i\mid \tilde{v}_i;n,F)}{H(\tilde{v}_i\mid \tilde{v}_i;n,F)}
 \right]\notag \\
 :=& \FP(\tilde{v}_i,v_i;n,F) \label{eq:firstprice_condition} . 
\end{align}
Hence, if $\FP(\tilde{v}_i,v_i;n,F) $ is non-negative for $\tilde{v}_i\in [0,v_i]$ and non-positive for $\tilde{v}_i\in [v_i,1]$ for any given $v_i\in [0,1]$, we have $\frac{\partial U^{\mathsf{FP}}(\tilde{v}_i,v_i;\sigma^{\mathsf{FP}}( ;n,F))}{\partial \tilde{v}_i} \geq 0$ for $\tilde{v}_i 
\in [0, v_i]$, and $\frac{U^{\mathsf{FP}}(\tilde{v}_i,v_i;\sigma^{\mathsf{FP}}(\cdot;n,F))}{\partial \tilde{v}_i} \leq 0$ for $\tilde{v}_i \in  [v_i,1]$. Therefore, $U^{\mathsf{FP}}(\tilde{v}_i,v_i;\sigma^{\mathsf{FP}}(\cdot;n,F))$ non-decreases with $\tilde{v}_i \in [0, v_i]$, and non-increases with $\tilde{v}_i \in [v_i, 1]$, as desired.
Besides, one can easily verify that under the null predictor or when admitting all bidders, this condition holds.

 \medskip
\noindent\emph{\underline{Step 3: Under the perfect predictor, the condition in \Cref{thm:SSM_firstprice}(i) holds, and $\sigma^{\FP}(v_i;n,\PP) = $}}

\noindent\textit{\underline{$v_i - \int_0^{v_i} \frac{F_1^{m-1}(x)}{F_1^{m-1}(v_i)}\, dx$.}}

 \medskip
 
For the perfect predictor, by \Cref{remark:affiliation}, we have $g(\cdot;n,\PP)$ is affiliated in $v\in [0,1]^n$. Then by the discussion in \Cref{app_subsec:condition_firstprice_SSM}, we conclude that 
\begin{align*}  \frac{h(\tilde{v}_i\mid v_i^\prime;n,\PP)}{H(\tilde{v}_i\mid v_i^\prime;n,\PP)} 
\leq \frac{h(\tilde{v}_i\mid \tilde{v}_i;n,\PP)}{H(\tilde{v}_i\mid \tilde{v}_i;n,\PP)} 
\leq 
 \frac{h(\tilde{v}_i\mid v_i;n,\PP)}{H(\tilde{v}_i\mid v_i;n,\PP)}
\end{align*}
for any $v_i^\prime\leq \tilde{v}_i\leq v_i$. This implies that $\FP(\tilde{v}_i,v_i;n,\PP) $ is non-negative for $\tilde{v}_i\in [0,v_i]$ and non-positive for $\tilde{v}_i\in [v_i,1]$ for any given $v_i\in [0,1]$.
Thus,  the condition in \Cref{thm:SSM_firstprice}(i) holds.

To obtain $\sigma^{\FP}(v_i;n,\PP) $, by  \Cref{cor:special_case_H_h}, we have
\begin{align*}
  &  H(x\mid x;n,\PP) = \frac{\kappa(x; n, \PP)}{\C^{m-1}_{n-1}}\cdot F_1^{m-1}(x),\\
  \textrm{and}\quad &h(x \mid x;n,\PP) = \frac{m-1}{\C^{m-1}_{n-1}}\kappa(x; n, \PP)\cdot F_1^{m-2}(x)f_1(x).
\end{align*}
This leads to 
\begin{align*}
\RHR(x; n, \PP)=   \frac{h(x \mid x;n,\PP)}{H(x \mid x;n,\PP)} = \frac{(m-1)f_1(x)}{F_1(x)}.
\end{align*}
Hence, we have $L(x\mid \tilde{v}_i;n,\PP) = \frac{F_1^{m-1}(x)}{F_1^{m-1}(v_i)},$ independent of the admitted number.
As a result, $\sigma^{\FP}(v_i;n,\PP) = v_i - \int_0^{v_i} \frac{F_1^{m-1}(x)}{F_1^{m-1}(v_i)}  dx$.   
\end{proof}

\begin{proof}[\textbf{Proof of \Cref{thm:opt_admittednumber_firstprice}}]

By \Cref{eq:SSM_equilibrium_firstprice} and \Cref{eq:L}, together with the definition of the reverse hazard rate \(\RHR(x; n, F)\), we have
\begin{align}
\sigma^{\FP}(v_i; n, F)
&= v_i - \int_0^{v_i} L(x \mid v_i; n, F)\, dx\notag \\
&= v_i - \int_0^{v_i} 
\exp\!\Big( -\int_x^{v_i} 
\frac{h(t \mid t; n, F)}{H(t \mid t; n, F)}\, dt 
\Big)\, dx\notag \\
&= v_i - \int_0^{v_i} 
\exp\!\Big( -\int_x^{v_i} \RHR(t; n, F)\, dt \Big)\, dx \label{eq:equilfirst_RHR}.
\end{align}
Hence, \(\sigma^{\FP}(v_i; n, F)\) is increasing in the reverse hazard rate.  When $n=m$, the posterior remains independent, hence 
\[
\RHR(x; n = m, F) = \frac{(m - 1)\, f_1(x)}{F_1(x)}.
\]
Therefore, when \Cref{eq:condition_firstprice} holds, it follows that 
\(\RHR(x; n, F) \le \RHR(x; n = m, F)\)
for any \( n \in [2, m] \). Thus,
\begin{align}\label{eq:oederofequil}
 \sigma^{\FP}(v_i; n, F) \;\le\; \sigma^{\FP}(v_i; n = m, F).   
\end{align}

Furthermore, by \Cref{lem:FOSD_number}~(ii), the CDF \( G^{\lar}(\cdot; n, F) \) of the largest valuation is increasing in \( n \) in the sense of first-order stochastic dominance. Therefore, for any \( n \in [2, m] \),
\begin{align*}
\R^{\FP}(n; F)
&\overset{(a)}{=} 
\mathbb{E}_{X \sim G^{\lar}(\cdot; n, F)}
\big[\, \sigma^{\FP}(X; n, F) \big] \overset{(b)}{\le} 
\mathbb{E}_{X \sim G^{\lar}(\cdot; n = m, F)}
\big[\, \sigma^{\FP}(X; n, F) \big] \\[4pt]
&\overset{(c)}{\le} 
\mathbb{E}_{X \sim G^{\lar}(\cdot; n = m, F)}
\big[\, \sigma^{\FP}(X; n = m, F) \big] \\[4pt]
&\overset{(d)}{=} 
\R^{\FP}(m; F).
\end{align*}
Here, (a) and (d) follow from \Cref{eq:rev_n_firstprice}; (b) holds by the first-order stochastic dominance of \( G^{\lar}(\cdot; n, F) \) for different $n$, see \Cref{lem:FOSD_number}(ii); and (c) holds by \Cref{eq:oederofequil}.
\end{proof}

\begin{proof}[\textbf{Proof of \Cref{lem:condition_example_firstprice}}]
Observe that
\begin{align*}
& \RHR(x; n, F) \leq
 \frac{(m-1)f_1(x)}{F_1(x)} \quad
\Longleftrightarrow\quad  \underbrace{(m-1)f_1(x)\cdot H(x\mid x;n,F) - F_1(x)\cdot h(x\mid x;n,F)}_{:=\Delta}\geq 0 .
\end{align*}
We have
\begin{align}
H(x\mid v_i; n, F) &= \int_{[0,x]^{n-1}}g(v_{-i}\mid v_i; n, F)dv_{-i} \notag\\
&=\kappa(v_i;n, F)\cdot \int_{[0,x]^{n-1}}\psi(v;n, F) \cdot \left(\prod_{j\in \I_{-i}}f_1(v_j)\right)dv_{-i} \notag\\
&= \kappa(v_i;n, F)\cdot \int_{[0,x]^{n-1}} \left(\int_0^1\prod_{k\in \I}[1-F_{2\mid 1}(t \mid v_k)] d{F}_2^{m-n}(t)\right)\cdot \left(\prod_{j\in \I_{-i}}f_1(v_j)\right)dv_{-i} \notag\\
&\overset{(a)}{=} \kappa(v_i;n,F)\cdot \int_0^1\left(\int_0^x [1-F_{2\mid 1}(t \mid v_k)] dF_1(v_k)\right)^{n-1}\cdot [1-F_{2\mid 1}(t\mid v_i)] d{F}_2^{m-n}(t) \notag\\
&\overset{(b)}{=} \kappa(v_i;n, F)\cdot \int_0^1A^{n-1}(x, t; F)\cdot [1-F_{2\mid 1}(t\mid v_i)] d{F}_2^{m-n}(t) \label{eq:H_general_predictor}.
\end{align}
The equality $(a)$ holds by the Fubini theorem, and the equality $(b)$ holds by the definition of $A(x,t;F)$ in \Cref{eq:def A}.
We then have
\begin{align}
h(x\mid v_i; n, F) = &\frac{\partial}{\partial x}H(x\mid v_i; n, F)\notag  \\
=& (n-1)\cdot \kappa(v_i;n, F)\cdot \int_0^1A^{n-2}(x,t;F) \cdot A_1(x, t;F) \cdot [1-F_{2\mid 1}(t\mid v_i)] d{F}_2^{m-n}(t)\notag\\
=& (n-1)f_1(x)\cdot \kappa(v_i;n, F)\cdot \int_0^1A^{n-2}(x,t;F) \cdot [1-C_1(F_1(x), F_2(t))] \cdot [1-F_{2\mid 1}(t\mid v_i)] d{F}_2^{m-n}(t)\label{eq:h},
\end{align}
where the last equality holds by 
\[
A_1(x, t; F) := \frac{\partial}{\partial x} A(x, t; F) = f_1(x)\cdot [1-C_1(F_1(x), F_2(t))].
\]
Based on the above analysis, we have
\begin{align*}
  \Delta
  =f_1(x)\cdot \kappa(x;n,F) \cdot \int_0^1A^{n-2}(x, t;F) \cdot [1-F_{2\mid 1}(t\mid x)] \cdot  \diamondsuit \,d{F}_2^{m-n}(t),
\end{align*}
where 
\begin{align*}
    \diamondsuit &= {(m-1)A( x, t;F) - (n-1)F_1(x)[1-C_1(F_1(x), F_2(t))]}\\
    & = (m-n)A( x, t;F) + (n-1)[F_1(x)C_1(F_1(x), F_2(t)) - C(F_1(x), F_2(t))].
\end{align*}
We have
\begin{align*}
&\frac{\Delta}{f_1(x)\cdot \kappa(x;n,F)} 
\\=&  \int_0^1A^{n-2}(x,t;F) \cdot [1-F_{2\mid 1}(t\mid x)] \cdot\diamondsuit \,d{F}_2^{m-n}(t)\\
=& \int_0^1A^{n-2}(x, t;F) [1-F_{2\mid 1}(t\mid x)][(m-n)A( x, t;F) + (n-1)[F_1(x)C_1(F_1(x), F_2(t)) - C(F_1(x), F_2(t))]] d{F}_2^{m-n}(t)\\
=& \underbrace{(m-n)\int_0^1A^{n-1}(x,t;F) \cdot [1-F_{2\mid 1}(t\mid x)]  d{F}_2^{m-n}(t)}_{:=\clubsuit}\\
&+ (m-n)(n-1) \int_0^1 {F}_2^{m-n-1}(t)A^{n-2}(x,t;F)[1-F_{2\mid 1}(t\mid x)] \cdot [F_1(x)C_1(F_1(x), F_2(t)) - C(F_1(x), F_2(t))] d{F}_2(t).
\end{align*}
Integration by parts yields
\begin{align*}
\clubsuit=& - (m-n)\int_0^1 {F}_2^{m-n}(t) ~  d \left(A^{n-1}( x,t;F) \cdot [1-F_{2\mid 1}(t\mid x)]\right)\\
=& -(m-n)(n-1) \int_0^1 {F}_2^{m-n}(t) \cdot A^{n-2}(x, t;F) \cdot [1-F_{2\mid 1}(t\mid x)] \cdot A_2(x, t;F)dt \\
&+ (m-n) \int_0^1 {F}_2^{m-n}(t) \cdot A^{n-1}(x,t;F)\cdot f_{2\mid 1}(t\mid x) dt.\\
=& (m-n)(n-1) \int_0^1 {F}_2^{m-n}(t) \cdot A^{n-2}(x, t;F) \cdot [1-F_{2\mid 1}(t\mid x)] \cdot C_2(F_1(x), F_2(t))dt \\
&+ (m-n) \int_0^1 {F}_2^{m-n}(t) \cdot A^{n-1}(x,t;F)\cdot f_{2\mid 1}(t\mid x) dt.
\end{align*}
Therefore,
\begin{align*}
&\frac{\Delta}{f_1(x)\cdot \kappa(x;n,F)}\\
=&\clubsuit + (m-n)(n-1) \int_0^1 {F}_2^{m-n-1}(t)A^{n-2}(x,t;F)[1-F_{2\mid 1}(t\mid x)] \cdot [F_1(x)C_1(F_1(x), F_2(t)) - C(F_1(x), F_2(t))] d{F}_2(t)\\
=&(m-n)(n-1)\int_0^1 {F}_2^{m-n-1}(t) \cdot A^{n-2}(x, t;F) \cdot [1-F_{2\mid 1}(t\mid x)]\cdot \spadesuit\, dt \\
& + (m-n) \int_0^1 {F}_2^{m-n}(t) \cdot A^{n-1}(x, t;F)\cdot f_{2\mid 1}(t\mid x) dt,
\end{align*}
where
\[
\spadesuit = F_1(x)\cdot C_{1}(F_1(x), F_2(x)) + F_2(t)\cdot C_{2}(F_1(x), F_2(x)) - C(F_1(x), F_2(x))
\]
Since $A(x,t;F)\geq 0$ by definition \eqref{eq:def A}, we have
\begin{align*}
 & \Delta \geq 0 \quad
\Longleftrightarrow \quad \frac{\Delta}{f(x)\cdot \kappa(x;n,F)}\geq 0\quad
 \Longleftarrow \quad \spadesuit\geq 0 \quad
\Longleftrightarrow \quad \text{\Cref{eq:sufficient_condition_firstprice} holds}.
\end{align*}
This completes the proof of the first part. We now prove (i) and (ii) separately.

\medskip
\noindent\textit{\underline{Proof of (i).}}

 \medskip

For the comonotonic copula, we have
\[
x\cdot C_1(x,y) + y\cdot C_2(x,y) - C(x, y) = x\cdot \1\left\{x\leq y\right\} + y\cdot \1\left\{y\leq x\right\} - \min\left\{x,y\right\} \geq 0.
\]
For the AMH copula, we have
\begin{align*}
    & xC_1(x, y) + yC_2(x, y) - C(x, y) \\
    = ~&\frac{xy\cdot [1- \alpha \cdot (1-y)] }{[1 - \alpha\cdot (1-x) (1-y)]^2} + \frac{xy\cdot [1- \alpha \cdot (1-x)] }{[1 - \alpha\cdot (1-x) (1-y)]^2} - \frac{xy}{1 - \alpha\cdot (1-x) (1-y)}\\
    = ~& \frac{xy}{[1 - \alpha\cdot (1-x) (1-y)]^2} \left[ 2 - \alpha \cdot (1-x) -  \alpha \cdot (1-y) - 1 + \alpha\cdot (1-x) (1-y)\right]\\
    = ~& \frac{xy}{[1 - \alpha\cdot (1-x) (1-y)]^2} \left[ 1 + \alpha \cdot (xy -1)\right]\\
    \geq ~& 0.
\end{align*}
For the FGM copula, we have
\begin{align*}
    & xC_1(x, y) + yC_2(x, y) - C(x, y) \\
    = ~&xy\cdot [1+\alpha(1-2x)(1-y)] + xy\cdot[1+\alpha(1-x)(1-2y)] - xy\cdot[1+\alpha(1-x)(1-y)]\\
    = ~& xy\cdot [1+ \alpha \cdot (1-2x-2y+3xy)]\\
    \geq ~& 0,
\end{align*}
where the final inequality follows from the fact that
    \[
    3xy - 2x - 2y + 1 \geq -1 \quad \forall x, y \in [0, 1].
    \]    
    
This completes the proof of (i). 

\medskip
\noindent\textit{\underline{Proof of (ii).}}

\medskip
The result follows directly from the fact that the partial derivative of any convex combination of two functions is equal to the corresponding convex combination of their partial derivatives.
\end{proof}

\begin{proof}[\textbf{{Proof of \Cref{thm:rev_prediction_firstprice}}}] We prove each part separately.

\medskip
\noindent\textit{\underline{Proof of (i).}}
\medskip

By \Cref{eq:equilfirst_RHR}, we have 
\[
    \sigma^{\FP}(v_i; n, F)
= v_i - \int_0^{v_i} 
\exp\!\Big( -\int_x^{v_i} \RHR(t; n, F)\, dt \Big) dx.
\]
Therefore, if $\mathsf{RHR}(x; n, F) \ge \mathsf{RHR}(x; n, \tilde{F})$ for all $x$, it follows immediately that 
\begin{equation}\label{eq:equilorder}
 \sigma^{\FP}(v_i; n, F) \ge \sigma^{\FP}(v_i; n, \tilde{F}),~\forall v_i\in[0,1].
\end{equation}
Furthermore, by \Cref{lem:FOSD_number}(iii), if $F$ is more accurate than $\tilde{F}$, then 
$G^{\lar}(x; n, F)$ first order stochastic dominance $G^{\lar}(x; n, \tilde{F})$.
Combining these observations, we obtain
\begin{align*}
 \R^{\FP}(n; F)
&\overset{(a)}{=} 
\mathbb{E}_{X \sim G^{\lar}(\cdot; n, F)}
\big[\, \sigma^{\FP}(X; n, F) \big] 
\overset{(b)}{\ge} 
\mathbb{E}_{X \sim G^{\lar}(\cdot; n, F)}
\big[\, \sigma^{\FP}(X; n, \tilde{F}) \big] 
\\
&\overset{(c)}{\ge} 
\mathbb{E}_{X \sim G^{\lar}(\cdot; n, \tilde{F})}
\big[\, \sigma^{\FP}(X; n, \tilde{F}) \big] 
\\
&\overset{(d)}{=} 
\R^{\FP}(n; \tilde{F}).   
\end{align*}
Here, (a) and (d) follow from \Cref{eq:rev_n_firstprice}; (b) holds by \eqref{eq:equilorder}; and (c) holds because $G^{\lar}(\cdot; n, F)$ first-order stochastically dominates $G^{\lar}(\cdot; n, \tilde{F})$ and the equilibrium strategy $\sigma^{\FP}(\cdot ;n,\tilde{F})$ is an increasing function.
This completes the proof of (i).

\medskip
\noindent\textit{\underline{Proof of (ii).}}
\medskip

By \Cref{thm:SSM_firstprice}, admitting any number of bidders yields the same expected revenue under the perfect predictor. When $n=m$, the posterior remains independent, hence $\RHR(x; n = m, F)$ for any predictor $F$ coincides with $\RHR(x; n, \PP)$ as given in \Cref{eq:reverse_hazard_rate_perfect_predictor}, i.e.,
\[
\RHR(x; n = m, F) 
= \frac{(m - 1)\, f_1(x)}{F_1(x)}
= \RHR(x; n, \PP).
\]
Furthermore, for any predictor $F$ satisfying \eqref{eq:condition_firstprice}, it follows from \Cref{thm:opt_admittednumber_firstprice} that
\[
\R^{\FP}(n; F) 
\le \R^{\FP}(m; F) 
= \R^{\FP}(n; \PP).
\]

This completes the proof of (ii).
\end{proof}

\begin{proof}[\textbf{{Proof of \Cref{exm:hazard_prediction}}}] When $n=m$, the posterior remains independent. 
Thus,
for any predictor, we have $\RHR(x; n = m, F) 
= \frac{(m - 1)\, f_1(x)}{F_1(x)}$. Therefore, 
the reverse hazard rate preserves the same ordering as the prediction accuracy for predictors in \Cref{obs:higher_accurate_common_predictors}. We now focus on the case of $n = 2$ in what follows.

By \Cref{eq:H_general_predictor} and \Cref{eq:h}, together with the identities 
\[
F_{2 \mid 1}(t \mid v_i) = C_1\big(F_1(v_i), F_2(t)\big)
\quad \text{and} \quad 
A(x, t; F) = F_1(x) - C\big(F_1(x), F_2(t)\big),
\]
we obtain the following expression for the reverse hazard rate:
\[
\begin{aligned}
\RHR(x; n, F)
&= \frac{h(x \mid x; n, F)}{H(x \mid x; n, F)} \\[6pt]
&= (n - 1) f_1(x) \cdot 
\frac{
\int_0^1 [F_1(x) - C(F_1(x), F_2(t)) ]^{n - 2} 
[ 1 - C_1(F_1(x), F_2(t)) ]^2 \, dF_2^{m - n}(t)
}{
\int_0^1 [F_1(x) - C(F_1(x), F_2(t)) ]^{n - 1} 
[ 1 - C_1(F_1(x), F_2(t)) ] \, dF_2^{m - n}(t)
} \\[6pt]
&= (n - 1) f_1(x) \cdot 
\frac{
\int_0^1 [F_1(x) - C(F_1(x), t) ]^{n - 2} 
[ 1 - C_1(F_1(x), t) ]^2 \, dt^{m - n}
}{
\int_0^1 [F_1(x) - C(F_1(x), t) ]^{n - 1} 
[ 1 - C_1(F_1(x), t) ] \, dt^{m - n}
}.
\end{aligned}
\]
To emphasize the dependence on the predictor $F$, we write $C(\cdot, \cdot; F)$. Accordingly, we can express $\RHR(x; n, F) = (n - 1) f_1(x) \, E(F_1(x), n; F),$ where 
\[
E(x, n; F)
= \frac{
\int_0^1 [ x - C(x, t; F) ]^{n - 2} [ 1 - C_1(x, t; F) ]^2 \, dt^{m - n}
}{
\int_0^1 [ x - C(x, t; F) ]^{n - 1} [ 1 - C_1(x, t; F) ] \, dt^{m - n}
}.
\]
We now show that \( E(x, n=2; \HB(\gamma)) \) and \( E(x, n=2; \FGM(\gamma)) \) are increasing in \( \gamma \), and that $E(x, n=2; \HB(\gamma)) \geq E(x, n=2; \FGM(\gamma)),$
which will complete the proof.

\medskip
\noindent\textit{\underline{(a). $E(x, n=2; \HB(\gamma))$ is increasing with $\gamma \in [0,1]$.}}
\medskip

By \Cref{cor:special_case_H_h}, we have
\begin{align*}
E(x, n=2; \HB(\gamma)) &= \frac{1}{x}\cdot\frac{\frac{(1-\gamma)^2}{3} + \gamma(2-\gamma)x-\gamma(1-\gamma)x^2}{\frac{(1-\gamma)^2}{3} + \frac{\gamma(3-2\gamma)}{2}x - \frac{2}{3}\gamma(1-\gamma)x^2} 
= \frac{1}{x}\cdot \left(1 + \frac{\frac{1}{2}\gamma x - \frac{1}{3}\gamma(1-\gamma)x^2}{\frac{(1-\gamma)^2}{3} + \frac{\gamma(3-2\gamma)}{2}x - \frac{2}{3}\gamma(1-\gamma)x^2}\right)\\
& = \frac{1}{x}\cdot \left(1 + \frac{\left(\frac{1}{2}x-\frac{1}{3}x^2\right)\gamma + \frac{1}{3}x^2\gamma^2}{\left(\frac{1}{3}-x + \frac{2}{3}x^2\right)\gamma^2 + \left(\frac{3}{2}x - \frac{2}{3}x^2 -\frac{2}{3}\right)\gamma + \frac{1}{3}}\right).    
\end{align*}
Hence
\begin{align*}
\frac{\partial}{\partial \gamma}  E(x, n=2; \HB(\gamma)) = \frac{1}{18x}\cdot \frac{(3x^3 - 8x^2 + 9x -3)\gamma^2+ 4x^2\gamma + 3-2x^2}{\left[\frac{(1-\gamma)^2}{3} + \frac{\gamma(3-2\gamma)}{2}x - \frac{2}{3}\gamma(1-\gamma)x^2\right]^2}.   
\end{align*}
Since 
\[
(3x^3 - 8x^2 + 9x - 3)\,\gamma^2 + 4x^2 \gamma + (3 - 2x^2)
\]
is a quadratic function of $\gamma$ with non-negative linear coefficient (because $4x^2 \ge 0$), it attains its minimum at either $\gamma = 0$ or $\gamma = 1$. Evaluating at these endpoints, we have
\[
\big[(3x^3 - 8x^2 + 9x - 3)\,\gamma^2 + 4x^2 \gamma + (3 - 2x^2)\big]_{\gamma = 0}
= 3 - 2x^2 > 0,
\]
and
\[
\big[(3x^3 - 8x^2 + 9x - 3)\,\gamma^2 + 4x^2 \gamma + (3 - 2x^2)\big]_{\gamma = 1}
= 3x^3 - 6x^2 + 9x \ge 0.
\]
Therefore, we conclude that $\frac{\partial}{\partial \gamma} E(x, n = 2; \HB(\gamma)) \ge 0,$ as desired.

\medskip
\noindent\textit{\underline{(b). $E\left(x, n=2; \FGM(\gamma)\right)$ is increasing with $\gamma \in [0,1]$.}}
\medskip

Observe that, by the proof of \Cref{exm:copulas_increasing_expectation}(iii), we have 
\[
C(x, y; \FGM(\gamma)) 
= xy \big[ 1 + \gamma (1 - x)(1 - y) \big], 
\quad 
C_1(x, y; \FGM(\gamma)) 
= y \big[ 1 + \gamma (1 - 2x)(1 - y) \big].
\]
Hence, it follows that 
\begin{align*}
E(x, n=2; \FGM(\gamma))
&= \frac{
\int_0^1 ( 1 - C_1(x, t; \FGM(\gamma) )^2 \, dt
}{
\int_0^1 [ x - C(x, t; \FGM(\gamma)) ] [ 1 - C_1(x, t; \FGM(\gamma)) ] \, dt
}\\
&=\frac{
\int_0^1 ( 1 - t[1+\gamma(1-2x)(1-t)])^2 \, dt
}{
\int_0^1 ( x - xt[1+\gamma (1-x)(1-t)] ) ( 1 - t[1+\gamma(1-2x)(1-t)]) \, dt
}\\
& = \frac{2\bigl(4\gamma^{2}x^{2}-4\gamma^{2}x+\gamma^{2}+10\gamma x-5\gamma+10\bigr)}%
     {x\bigl(4\gamma^{2}x^{2}-6\gamma^{2}x+2\gamma^{2}+15\gamma x-10\gamma+20\bigr)}\\
     & = \frac{2}{x}\cdot \left(1+ \frac{(2x-1)\gamma^2 + 5(1-x)\gamma -10}{(4x^2 -6x + 2)\gamma^2 + 5(3x -2)\gamma +20}\right).
\end{align*}
Furthermore, the derivative with respect to $\gamma$ can be computed explicitly:
\begin{align*}
\frac{\partial}{\partial \gamma}  E(x, n=2; \FGM(\gamma)) 
&= 10\cdot \frac{(2x-1)^2\gamma^2 + 8(2x-1)\gamma + 10}{[(4x^2 -6x + 2)\gamma^2 + 5(3x -2)\gamma +20]^2}\\
&\geq 10\cdot \frac{-8\gamma + 10}{[(4x^2 -6x + 2)\gamma^2 + 5(3x -2)\gamma +20]^2}\\
&\geq 0,    
\end{align*}
as desired.

\medskip
\noindent\textit{\underline{(c). $E(x, n=2; \HB(\gamma)) \geq E(x, n=2; \FGM(\gamma))$.}}
\medskip

By the proofs of parts (a) and (b), it suffices to show that for any $x$ and $\gamma \in [0,1]$,
\begin{align*}
&\frac{1}{x} \cdot \frac{\frac{(1-\gamma)^2}{3} + \gamma(2-\gamma)x - \gamma(1-\gamma)x^2}{\frac{(1-\gamma)^2}{3} + \frac{\gamma(3-2\gamma)}{2}x - \frac{2}{3}\gamma(1-\gamma)x^2} 
\;\geq\; 
\frac{2\left(4\gamma^{2}x^{2} - 4\gamma^{2}x + \gamma^{2} + 10\gamma x - 5\gamma + 10\right)}{x\left(4\gamma^{2}x^{2} - 6\gamma^{2}x + 2\gamma^{2} + 15\gamma x - 10\gamma + 20\right)} \\
\Longleftrightarrow\;& 
\left[(1-\gamma)^2 + 3\gamma(2-\gamma)x - 3\gamma(1-\gamma)x^2\right] \cdot \left(4\gamma^{2}x^{2} - 6\gamma^{2}x + 2\gamma^{2} + 15\gamma x - 10\gamma + 20\right)\geq \\
&  
\left[2(1-\gamma)^2 + 3\gamma(3 - 2\gamma)x - 4\gamma(1-\gamma)x^2\right] \cdot \left(4\gamma^{2}x^{2} - 4\gamma^{2}x + \gamma^{2} + 10\gamma x - 5\gamma + 10\right) \\
\Longleftrightarrow\;& 
4\gamma^3(1-\gamma)x^4 - 5\gamma^2(1-\gamma)(2\gamma + 1)x^3 
- \gamma\left(8\gamma^3 - 11\gamma^2 - 26\gamma + 20\right)x^2 
+ \gamma\left(2\gamma^3 - 16\gamma^2 - 3\gamma + 25\right)x \;\geq\; 0 \\
\Longleftrightarrow\;& 
4\gamma^2(1-\gamma)x^3 - 5\gamma(1-\gamma)(2\gamma + 1)x^2 
- \left(8\gamma^3 - 11\gamma^2 - 26\gamma + 20\right)x 
+ \left(2\gamma^3 - 16\gamma^2 - 3\gamma + 25\right) \;\geq\; 0.
\end{align*}
Define the function
\[
\Gamma(x;\gamma) = 4\gamma^2(1-\gamma)x^3 - 5\gamma(1-\gamma)(2\gamma + 1)x^2 
- \left(8\gamma^3 - 11\gamma^2 - 26\gamma + 20\right)x 
+ \left(2\gamma^3 - 16\gamma^2 - 3\gamma + 25\right),
\]
where \( x, \gamma \in [0,1] \). The derivative of \( \Gamma \) with respect to \( x \) is given by
\[
\frac{\partial}{\partial x}\Gamma(x;\gamma) = 12\gamma^2(1-\gamma)x^2 - 10\gamma(1-\gamma)(2\gamma + 1)x - \left(8\gamma^3 - 11\gamma^2 - 26\gamma + 20\right).
\]
We analyze the monotonicity of \( \Gamma(x;\gamma) \) over the interval \( x \in [0,1] \). Note that the quadratic coefficient in the derivative is non-negative for all \( \gamma \in [0,1] \), since \( \gamma^2(1 - \gamma) \geq 0 \). To determine the behavior of the derivative, observe that the turning point of the quadratic is located at $\frac{10\gamma(1-\gamma)(2\gamma + 1)}{24\gamma^2(1-\gamma)} = \frac{5(2\gamma + 1)}{12\gamma}\geq 1 $.
This implies that the derivative \( \frac{\partial}{\partial x}\Gamma(x;\gamma) \) is decreasing on the interval \( [0,1] \). Therefore, the function \( \Gamma(x;\gamma) \) is either monotonic or unimodal (i.e., first increasing then decreasing) on \( [0,1] \). In all such cases, the minimum of \( \Gamma(x;\gamma) \) over \( x \in [0,1] \) must occur at the boundary points.

We now evaluate \( \Gamma \) at the boundaries. Since for all $x,\gamma \in [0,1]$,
\[
\Gamma(0;\gamma) = 2\gamma^3 - 16\gamma^2 - 3\gamma + 25 \geq 2\gamma^3 + 6 > 0,
\]
and
\[
\Gamma(1;\gamma) = -6\gamma^2 + 18\gamma + 5 > 0, 
\]
we conclude that
\[
\Gamma(x;\gamma) \geq \min\{\Gamma(0;\gamma), \Gamma(1;\gamma)\} > 0 \quad \text{for all } x, \gamma \in [0,1],
\]
as desired.
\end{proof}

\subsection{Proofs for \Cref{sec:all-pay_auctions}}

\begin{proof}[\textbf{{Proof of \Cref{thm:strictlymonotone_equilibrium_allpay}}}]
The proof includes four steps: 
\begin{enumerate}
    \item  derive the (unique) SSM equilibrium $\sigma^{\mathsf{AP}}(v_i;n,F)$ if exists;
    \item  show that $\sigma^{\mathsf{AP}}(v_i;n,F)$ is indeed an equilibrium under the condition in \Cref{thm:strictlymonotone_equilibrium_allpay}(i);
    \item show that when \(\PRD(S_i \mid V_i)\) holds, if the inflated type 
\(\theta(v_i; n, F)\) is non-decreasing in \(v_i\), then the condition in \Cref{thm:strictlymonotone_equilibrium_allpay}(i) is satisfied.
\item show that if there exists an SSM equilibrium under the perfect predictor, then the inflated type \(\theta(v_i; n, \PP)\) must be non-decreasing in \(v_i\).
\end{enumerate}

 \medskip

\noindent\textit{\underline{Step 1: SSM equilibrium if it exists.}}

 \medskip

When all other $n-1$ admitted bidders follow a strictly increasing strategy $\sigma:[0,1]\to \mathbb{R}_+$, the best response of the bidder $i$ is
\begin{align}
\label{eq:utility function all pay}
    \argmax_{\tilde{v}_i\in [0,1]}~  \Pr\left\{\sigma({V}_j)\leq \sigma(\tilde{v}_i),\forall j\in \I_{-i}\right\} \cdot v_i - \sigma(\tilde{v}_i),
\end{align}
where $\Pr\left\{\sigma({V}_j)\leq \sigma(\tilde{v}_i),\forall j\in \I_{-i}\right\}$ is the winning probability under bidder $i$'s belief when she bids $\sigma(\tilde{v}_i)$.  
Observe that 
\begin{align*}
   \Pr\left\{\sigma({V}_j)\leq \sigma(\tilde{v}_i),\forall j\in \I_{-i}\right\} =\Pr\left\{V_j\leq \tilde{v}_i,\forall j\in \I_{-i}\right\} = H(\tilde{v}_i\mid v_i;n,F),
\end{align*}
where the first equality holds since $\sigma(\cdot)$ is strictly increasing, and the last equality holds because $H(\cdot \mid v_i;n,F)$ is the CDF of $\max_{j\in \I_{-i}}V_j$ by \Cref{eq:H_CDF}. Hence, \Cref{eq:utility function all pay} is equivalent to 
\begin{equation}
\label{eq:utility function all pay used}
\max_{\tilde{v}_i\in[0,1]}~U^{\mathsf{AP}}(\tilde{v}_i,v_i;\sigma(\cdot)):= v_i H(\tilde{v}_i\mid v_i; n,F) - \sigma(\tilde{v}_i).
\end{equation}
If there exists an SSM equilibrium strategy, it must satisfy the first-order condition:
\begin{align*}
    \frac{\partial U^{\mathsf{AP}}(\tilde{v}_i,v_i;\sigma(\cdot))}{\partial \tilde{v}_i}\bigg |_{\tilde{v}_i = v_i} = \left(v_ih(\tilde{v}_i\mid v_i; n,F) - \frac{d \sigma(\tilde{v}_i)}{d\tilde{v}_i}\right)\bigg |_{\tilde{v}_i = v_i} = 0.
\end{align*}
The above differential equation yields the unique SSM equilibrium strategy if it exists as follows:
\begin{align*}
   \sigma(v_i) =\int_0^{v_i} x h(x\mid x;n,F)dx
   := \sigma^{\AP}(v_i;n,F),~ \forall v_i\in [0,1].
\end{align*}
One can easily check that $\sigma^{\AP}(v_i;n,F)$ is indeed strictly increasing in $v_i\in[0,1]$.

 \medskip
\noindent\textit{\underline{Step 2: $\sigma^{\mathsf{AP}}(v_i;n, F)$ is indeed an equilibrium under the condition in \Cref{thm:strictlymonotone_equilibrium_allpay}(i).}}
\medskip

It is equivalent to showing that
\begin{align*}
    U^{\mathsf{AP}}(\tilde{v}_i,v_i;\sigma^{\mathsf{AP}}(\cdot;n,F))= v_i H(\tilde{v}_i\mid v_i; n, F) - \sigma^{\mathsf{AP}}(\tilde{v}_i;n, F) = v_i H(\tilde{v}_i\mid v_i; n, F) - \int_0^{\tilde{v}_i} x h(x\mid x;n, F)dx.
\end{align*}
indeed obtains the maximum at $\tilde{v}_i = v_i$, i.e., when all other bidders follow the strategy $\sigma^{\AP}(\cdot;n,F)$, it is optimal for bidder $i$ also to use this strategy.
It's sufficient to show that $U^{\mathsf{AP}}(\tilde{v}_i,v_i;\sigma^{\mathsf{AP}}(\cdot;n,F))$ non-decreases with $\tilde{v}_i \in [0, v_i]$, and non-increases with $\tilde{v}_i \in [v_i, 1]$.

Observe that
$$\frac{\partial U^{\mathsf{AP}}(\tilde{v}_i,v_i;\sigma^{\mathsf{AP}}(\cdot;n,F))}{\partial \tilde{v}_i} = v_ih(\tilde{v}_i\mid {v}_i;n, F) -  \tilde{v}_ih(\tilde{v}_i \mid \tilde{v}_i;n, F),$$ 
hence if $v_i h(\tilde{v}_i\mid v_i;n,F) - \tilde{v}_i h(\tilde{v}_i\mid \tilde{v}_i;n,F)$ is non-negative for $\tilde{v}_i\in [0,v_i]$ and non-positive for $\tilde{v}_i\in [v_i,1]$ for any given $v_i\in [0,1]$, we have $\frac{\partial U^{\mathsf{AP}}(\tilde{v}_i,v_i;\sigma^{\mathsf{AP}}(\cdot;n,F))}{\partial \tilde{v}_i} \geq 0$ for $\tilde{v}_i 
\in [0, v_i]$, and $\frac{U^{\mathsf{AP}}(\tilde{v}_i,v_i;\sigma^{\mathsf{AP}}(\cdot;n,F))}{\partial \tilde{v}_i} \leq 0$ for $\tilde{v}_i \in  [v_i,1]$. Therefore, $U^{\mathsf{AP}}(\tilde{v}_i,v_i;\sigma^{\mathsf{AP}}(\cdot;n,F))$ non-decreases with $\tilde{v}_i \in [0, v_i]$, and non-increases with $\tilde{v}_i \in [v_i, 1]$, as desired.

\medskip 

\noindent\textit{\underline{Step 3: When $\PRD(S_i\mid V_i)$ holds, the condition in \Cref{thm:strictlymonotone_equilibrium_allpay}(i) holds if the inflated type
}}

\noindent\textit{\underline{
$\theta\left(v_i; n, F\right)$ is non-decreasing.}}
 \medskip

To verify the condition in \Cref{thm:strictlymonotone_equilibrium_allpay}(i), it suffices to ensure that 
$v_ih(\tilde{v}_i\mid {v}_i;n, F)$ is non-decreasing in $v_i\in[0,1]$ for any $\tilde{v}_i \in [0,1]$.
By \Cref{eq:h}, we have
\[
h(\tilde{v}_i \mid v_i; n, F)
= (n-1)\, f_1(\tilde{v}_i)\, \kappa(v_i; n, F)
\int_0^1 A^{n-2}(\tilde{v}_i, t; F)\, 
\big[ 1 - C_1(F_1(\tilde{v}_i), F_2(t)) \big]\,
\big[ 1 - F_{2 \mid 1}(t \mid v_i) \big]\,
dF_2^{m-n}(t).
\]
Hence,
\[
v_i\, h(\tilde{v}_i \mid v_i; n, F)
= (n-1)\, f_1(\tilde{v}_i)\, \theta(v_i; n, F)
\int_0^1 A^{n-2}(\tilde{v}_i, t; F)\,
\big[ 1 - C_1(F_1(\tilde{v}_i), F_2(t)) \big]\,
\big[ 1 - F_{2 \mid 1}(t \mid v_i) \big]\,
dF_2^{m-n}(t),
\]
which is non-decreasing in \( v_i \in [0,1] \) when \(\PRD(S_i \mid V_i)\) holds and \(\theta(v_i; n, F)\) is non-decreasing in \( v_i \).

\medskip 

\noindent\textit{\underline{Step 4: If there exists an SSM equilibrium under the perfect predictor, then the inflated type \(\theta(v_i; n, \PP)\)}}

\noindent\textit{\underline{is non-decreasing in \(v_i\).}}
 \medskip

This follows from Theorem 1 in \cite{sun_2024_contests}. Step 1 to Step 4 complete the proof.
\end{proof}

\begin{proof}[\textbf{{Proof of \Cref{prop:larger_gamma_stricter_condition}}}]

By the proof of \Cref{lem:closedforms_admission_prob}, see \Cref{eq:1_over_kappa}, we have
\begin{align*}
    \frac{1}{\kappa(v_i;n,\mathsf{HP}(\gamma))} 
    &= \frac{1-\gamma}{\C^m_n} + \gamma \int_0^{v_i} (1-F_1(x))^{n-1} dF_1^{m-n}(x),
\end{align*}
which implies
\begin{align*}
    \theta(v_i;n,\mathsf{HP}(\gamma)) = \kappa(v_i;n,\mathsf{HP}(\gamma))v_i = \frac{v_i}{\frac{1-\gamma}{\C^m_n} + \gamma \int_0^{v_i} (1-F_1(x))^{n-1} dF_1^{m-n}(x)}.
\end{align*}
Differentiating $\theta(v_i;n,\mathsf{HP}(\gamma))$ with respect to $v_i$, we obtain
\begin{align*}
    \theta'(v_i;n,\mathsf{HP}(\gamma)) &= \frac{W(v_i;\gamma)}{\left(\frac{1-\gamma}{\C^m_n} + \gamma \int_0^{v_i} (1-F_1(x))^{n-1} dF_1^{m-n}(x)\right)^2},
\end{align*}
where
\begin{align*}
    W(v_i;\gamma) &= (m-n)\gamma \int_0^{F_1(v_i)} t^{m-n-1}(1-t)^{n-1} dt + \frac{1-\gamma}{\C^m_n}  - v_i (m-n)\gamma F_1^{m-n-1}(v_i)[1-F_1(v_i)]^{n-1} f_1(v_i).
\end{align*}
For $F_1(x) = x^c$, $W(v_i;\gamma)$ is differentiable, hence we have 
\begin{align*}
    \frac{\partial W(v_i;\gamma)}{\partial v_i} 
    &=-(m-n)\cdot \gamma \cdot v_i \cdot F_1^{m-n-2}(v_i) (1-F_1(v_i))^{n-2} [m-n-1-(m-2)F_1(v_i)] f_1^2(v_i) \\
    &\quad -(m-n)\cdot\gamma \cdot v_i \cdot F_1^{m-n-1}(v_i) (1-F_1(v_i))^{n-1} f_1'(v_i) \\
    &= (m-n)\cdot \gamma \cdot c \cdot v_i^{(m-n)c-1} (1-v_i^c)^{n-2} \left([(m-1)c-1] v_i^c - (m-n)c+1\right).
\end{align*}

\begin{enumerate}[(a)]
    \item If $(m - n)c \leq 1$, then
    \begin{align*}
        \left([(m - 1)c - 1]v_i^c - (m - n)c + 1 \right)\Big|_{v_i=0} &= 1 - (m - n)c \geq 0, \\
        \left([(m - 1)c - 1]v_i^c - (m - n)c + 1 \right)\Big|_{v_i=1} &= (n - 1)c > 0.
    \end{align*}
    Since $((m - 1)c - 1)v_i^c - (m - n)c + 1$ is monotone with $v_i$ in $[0, 1]$, we conclude that $\frac{\partial W(v_i;\gamma)}{\partial v_i} \geq 0$ in $[0, 1]$. Observe that $W(0,\gamma) = \frac{1-\gamma}{\C_n^m} \geq 0$, it follows that $W(v_i;\gamma) \geq 0$ and $\theta(v_i; n,\mathsf{HP}(\gamma))$ is increasing in $[0,1]$.
    
    \item If $(m - n)c > 1$, then $W(v_i;\gamma)$ is decreasing on $\left[0, k^{\frac{1}{c}}(c)\right]$ and increasing on $\left[k^{\frac{1}{c}}(c), 1\right]$, where
    \begin{align*}
        k(c) = \frac{(m - n)c - 1}{(m - 1)c - 1} = \frac{m-n}{m-1} - \frac{n-1}{(m-1)[(m-1)c - 1]}.
    \end{align*}
    Hence $W(v_i;\gamma) \geq 0$ for all $v_i \in [0,1]$ if and only if $W\left(k^{\frac{1}{c}}(c); \gamma\right) \geq 0$.
\end{enumerate}
From (a) and (b), we conclude that $\theta(v_i;n,\gamma)$ is non-decreasing in $[0,1]$ if $(m-n)c\leq 1$ or $(m-n)c >1$ and $W\left(k^{\frac{1}{c}}(c); \gamma\right)\geq 0$. Observe that $W\left(k^{\frac{1}{c}}(c);\gamma\right)$ can be written as
\begin{align*}
W\left(k^{\frac{1}{c}}(c);\gamma\right) =& (m-n)\gamma \int_0^{k(c)} t^{m-n-1}(1-t)^{n-1}dt - (m-n)\gamma \cdot c \cdot k^{m-n}(c)(1-k(c))^{n-1} + \frac{1-\gamma}{\C^m_n} \\
=& \gamma \cdot \widetilde{W}(c) + \frac{1}{\C^m_n},
\end{align*}
where $\widetilde{W}(c) = (m-n) \int_0^{k(c)} t^{m-n-1}(1-t)^{n-1}dt - (m-n)c \cdot k^{m-n}(c)(1-k(c))^{n-1} - \frac{1}{\C^m_n}$. We now show that $\widetilde{W}(c)$ is a strictly decreasing and negative function with $\widetilde{W}(\frac{1}{m-n}) = - \frac{1}{\C^m_n}$.
Then, $\theta(v_i;n,\gamma)$ is non-decreasing in $[0,1]$ is equivalent to $c \in \left(0, \widetilde{W}^{-1} \left(-\frac{1}{\gamma C^m_n}\right) \right]$, which completes the proof.

\medskip
\noindent To see this, by calculation, we have
\begin{align*}
 \widetilde{W}'(c) =& (m-n)k'(c) \left(k^{m-n-1}(c)\left[1-k(c)\right]^{n-1} - c\cdot k^{m-n-1}(c)\left[1-k(c)\right]^{n-2} \left[m-n- (m-1)k(c)\right]\right)\\
 & - (m-n)\cdot k^{m-n}(c)\cdot (1-k(c))^{n-1} \\
 \overset{(a)}{\leq } & (m-n)k'(c) \left(k^{m-n-1}(c)\left[1-k(c)\right]^{n-1} - c\cdot k^{m-n-1}(c)\left[1-k(c)\right]^{n-2} \left[m-n- (m-1)k(c)\right]\right) \\
 =& (m-n)k'(c) \cdot k^{m-n-1}(c)\left[1-k(c)\right]^{n-2} \cdot \left( 1- k(c) - c \left[m-n- (m-1)k(c)\right]\right)\\
 =& (m-n)k'(c) \cdot k^{m-n-1}(c)\left[1-k(c)\right]^{n-2} \cdot \left( [(m-1)c-1]k(c) + 1-(m-n)c\right)\\
 =&0,
\end{align*}
The inequality $(a)$ holds since $k(c)\in(0,1)$.
Hence, $\widetilde{W}(c)$ is strictly decreasing for $c \geq \frac{1}{m-n}$.  Consequently, $\widetilde{W}(c) < \widetilde{W}(\frac{1}{m-n}) = - \frac{1}{\C^m_n} <0$, and the inverse function $\widetilde{W}^{-1}(\cdot)$ is also strictly decreasing.
Therefore,
\begin{align*}
    W\left(k^{\frac{1}{c}}(c); \gamma\right)\geq 0,\ c\geq \frac{1}{m-n} \iff \frac{1}{m-n}\leq c \leq \widetilde{W}^{-1}\left(-\frac{1}{\gamma \C^m_n}\right),
\end{align*}
and we can conclude that $\theta(v_i;n,\mathsf{HP}(\gamma))$ is non-decreasing in $[0,1]$ is equivalent to $c \in \left(0, \widetilde{W}^{-1} \left(-\frac{1}{\gamma C^m_n}\right) \right]$. This implies \Cref{prop:larger_gamma_stricter_condition}(i) and \Cref{prop:larger_gamma_stricter_condition}(ii).
\end{proof}

\begin{proof}[\textbf{{Proof of \Cref{thm:opt_admittednumber_allpay}}}]

We prove each case separately. 

\medskip
\noindent\underline{(i) \textit{Proof of the perfect predictor.}}

\medskip
By \Cref{cor:special_case_H_h} and \Cref{eq:SSM_equilibrium_all_pay}, the expected revenue is given by
\begin{align}
\R^{\AP}(n; \PP) 
&= n\, \mathbb{E}_{V_i \sim G^{\mar}(\cdot; n, \PP)}
\big[\, \sigma^{\AP}(V_i; n, \PP) \big] \notag\\[4pt]
&= n \int_0^1 \int_0^{v_i} 
\frac{x}{J(F_1(x), n, m)}\, dF_1^{m-1}(x)\, dG^{\mar}(v_i; n, \PP) \notag\\[4pt]
&= \int_0^1 
\frac{n}{J(F_1(x), n, m)}\, 
\big[\, 1 - G^{\mar}(x; n, \PP) \big]\, 
x\, dF_1^{m-1}(x),\label{eq:allpay_rev}
\end{align}
where the last equality follows by Fubini's theorem. We will show that $\frac{n}{J(F_1(x),n,m)}\cdot \left(1-G^{\mar}(x;n,\PP)\right)$ is strictly increasing in $n\in [2, m]$ for any $x\in (0,1)$, hence the unique optimal solution is $n^*=2$. 

Observe that 
\begin{align}
\frac{n}{J(F_1(x),n,m)}\cdot (1-G^{\mar}(x;n,\PP)) 
& = \frac{n}{J(F_1(x),n,m)}\cdot \left[1- \frac{m}{n} \cdot \int_0^{F_1(x)}J(t,n,m)dt\right] \nonumber\\
& = \frac{n}{J(F_1(x),n,m)} - \frac{m}{J(F_1(x),n,m)}\cdot \int_0^{F_1(x)} J(t,n,m)dt. \label{eq:allpay_rev_gamma=1_intergrand}
\end{align}
The first equality holds because 
\begin{align*}
  G^{\mar}(v_i; n,\PP) &\overset{(a)}{=}  \C_n^m  \int_0^{v_i} \frac{1}{\kappa(t;n,\PP)}f_1(t)dt\overset{(b)}{=}
\C_n^m  \int_0^{v_i} \frac{J(F_1(t),n,m)}{\C^{m-1}_{n-1}}  dF_1(t) = \frac{m}{n} \cdot \int_0^{F_1(x)}J(t,n,m)dt.
\end{align*}
The equality $(a)$ holds by \Cref{eq:margindist_kappa}. 
The equality $(b)$ follows from the definition of $J(F_1(t),n,m)$.

By \Cref{cor:special_gamma} about $\kappa(t,n,\PP)$, we have 
\begin{align}
 J(F_1(x),n,m) & = \C^{m-1}_{n-1} / \kappa(x; n, \PP) \nonumber \\
 &= \C^{m-1}_{n-1}   \cdot \left[
 (n-1)\int_0^{F_1(x)}t^{m-n}(1-t)^{n-2}dt + F_1^{m-n}(x)(1-F_1(x))^{n-1}
 \right]  \label{eq:def_J_integral}\\
 &= \1\left\{n<m\right\}\cdot \C^{m-1}_{n-1}\cdot (m-n)\cdot \int_0^{F_1(x)}t^{m-n-1}(1-t)^{n-1}dt + \1\left\{n=m\right\}\nonumber.
\end{align}
Consider $n<m$:
\begin{align}
\int_0^x J(t,n,m)dt =& (m-n)\C_{n-1}^{m-1}  \cdot\int_{0}^xdt\int_0^t u^{m-n-1}(1-u)^{n-1}du \nonumber\\
=& (m-n)\C_{n-1}^{m-1} \cdot\int_{0}^xdu\int_u^x u^{m-n-1}(1-u)^{n-1}dt \nonumber\\
=& (m-n)\C_{n-1}^{m-1} \cdot \left(x\int_0^xu^{m-n-1}(1-u)^{n-1}du - \int_0^xu^{m-n}(1-u)^{n-1}du\right) \nonumber\\
=& xJ(x,n,m) - \frac{m-n}{m}\cdot J(x,n,m+1) . \label{eq:integral_J}
\end{align}
It is easy to verify that, when $n=m$, \Cref{eq:integral_J} still holds.
Plugging \Cref{eq:integral_J} into \Cref{eq:allpay_rev_gamma=1_intergrand}, we have
\begin{align}
\label{eq:allpay_rev_gamma=1_intergrand2}
    \frac{n}{J(F_1(x),n,m)}\cdot (1-G^{\mar}(x;n,\PP))   = 
 \frac{n + (m-n)J(F_1(x),n,m+1)}{J(F_1(x),n,m)} -mF_1(x).
\end{align}
We will show that for any integer $n\in [3,m]$,
\begin{align*}
 &  \frac{n-1 + (m-n +1)J(x,n-1,m+1)}{J(x,n-1,m)} \geq  \frac{n + (m-n)J(x,n,m+1)}{J(x,n,m)},\ \forall x\in [0,1] \\
   \Longleftrightarrow \quad &l(x)\geq 0, \  \forall x\in[0,1],
\end{align*}
where 
\begin{align*}
   l(x) := J(x,n,m) \left[n-1 + (m-n +1)J(x,n-1,m+1) \right] - J(x,n-1,m) \left[n + (m-n )J(x,n,m+1) \right]. 
\end{align*}
Besides, to show \Cref{eq:allpay_rev_gamma=1_intergrand2} is strictly decreasing in the admitted number, we will prove that $l(x)>0$ for any $x\in (0,1)$.
Notice that $J(0,n,m)=0$ and $J(1,n,m)=1$ for all $n\leq m$, and thus we have $l(0)=l(1)=0$. 

To conclude the proof, we will show that $l(x)$ firstly strictly increases and then strictly decreases with $x\in (0,1)$. Consider the derivative as follows:
\begin{align*}
 l'(x) = & J'(x,n,m) \left[n-1 + (m-n +1)J(x,n-1,m+1) \right]  + (m-n+1)J(x,n,m)J'(x,n-1,m+1) \\
 & - J'(x,n-1,m) \left[n + (m-n )J(x,n,m+1)\right] - (m-n)J(x,n-1,m)J'(x,n,m+1)\\
\overset{(a)}{=}& (n-1)\C_{n-1}^{m-1} x^{m-n-1}(1-x)^{n-2}l_1(x),
\end{align*}
where 
\begin{align*}
l_1(x) =& m-n + mx^2J(x,n,m) + \frac{(m-n+1)(m-n)}{n-1}(1-x) J(x,n-1,m+1) \\
&-  \frac{m(m-n)}{n-1}x(1-x) J(x,n-1,m) - (m-n)x J(x,n,m+1) -mx.
\end{align*}
Notice that, in the equality (a), we use the following key identity:
\begin{align*}
 J^\prime(x,n,m)
 = \C^{m-1}_{n-1}(m-n)x^{m-n-1}(1-x)^{n-1}.
\end{align*}

Taking the derivative of $l_1(x)$, we have 
\begin{align*}
    l'_1(x) =& 2mxJ(x,n,m) - (m-n) J(x,n,m+1) - \frac{m(m-n)}{n-1}(1-2x)J(x,n-1,m)\\
    &- \frac{(m-n)(m-n+1)}{n-1}J(x,n-1,m+1) - m,
\end{align*}
and $l''_1(x) = 2mJ(x,n,m) + \frac{2m(m-n)}{n-1}J(x,n-1,m)\geq 0.$ Hence $l'_1(x) $ is nondecreasing in $x\in[0,1]$. Since $l'_1(0) = - m<0$ and $l'_1(1) = m>0$, 
we conclude that $l_1(x)$ first strictly decreases and then strictly increases. Combining the fact $l_1(0) = m-n + \frac{(m-n+1)(m-n)}{n-1}>0$ and $l_1(1) = 0$, we have $l(x)$ first strictly increases and then strictly decreases, which implies that $l(x)> l(0)=l(1)= 0$ for $x\in (0,1)$, as desired.

\medskip

\noindent\underline{(ii) \textit{Proof of the null predictor.}}
\medskip

By \Cref{lem:H_h} and \Cref{lem:marginals_joint_g}, we know that $h(x\mid x;n,\NP) = (n-1)F_1^{n-2}(x)f_1(x)$, $G^{\mar}(\cdot; n,\NP) = F_1(x)$ and $G^{\lar}(\cdot; n,\NP) = F_1^n(x)$. 
That is, for the null predictor, the distribution environment remains the same as the prior.
Besides, by \Cref{thm:strictlymonotone_equilibrium_allpay}(iii), we have
\begin{align*}
\sigma^{\AP}(v_i;n,\NP)
= \int_0^{v_i} xdF_1^{n-1}(x),~\forall v_i\in [0,1].
\end{align*}
The expected revenue is
\begin{align*}
 \R^{\AP}(n;\NP) & = n\cdot \mathbb{E}_{V_i\sim G^{\mar}(\cdot; n,\NP)}[\sigmaAP(V_i;n,\NP)]\\
 &= n\int_0^1 \int_0^{v_i} x\, dF_1^{n-1}(x)\, dG^{\mar}(v_i;n,\NP)\\[1mm]
 &= \int_0^1 nx \left(1-G^{\mar}(x; n,\NP)\right)\, dF_1^{n-1}(x) \\
 &= \int_0^1 nx \left(1-F_1(x)\right)\, dF_1^{n-1}(x)\\
 & = n \int_0^1 xdF_1^{n-1}(x) - (n-1) \int_0^1 xdF_1^{n}(x)\\
 & = (n-1)\int_0^1 F_1^{n}(x)dx - n\int_0^1 F_1^{n-1}(x)dx + 1,
\end{align*}
where the third equality is obtained by the Fubini theorem.
Thus, we have 
\begin{align*}
 \R^{\AP}(n+1;\NP) - \R^{\AP}(n;\NP) = n\int_0^1 F_1^{n-1}(x) (1- F_1(x))^2  dx > 0,
\end{align*}
which implies that $\R^{\AP}(n;\NP,F)$ is strictly increasing with the admitted number $n$. Therefore, $n^* = m$.
\end{proof}

\subsection{Proofs for \Cref{sec:rev_ranking}}
\begin{proof}[\textbf{Proof of \Cref{thm:rev_ranking}}] We prove each result separately.

\medskip
\noindent\underline{\textit{Proof of (i).}}

\medskip

\Cref{thm:opt_admittednumber_secondprice} establishes that admitting all bidders is optimal for second-price auctions. Since admitting all bidders is always feasible in all-pay auctions, it follows that an all-pay auction with optimal prescreening achieves (weakly) higher revenue than a second-price auction. Furthermore, if the condition \eqref{eq:condition_firstprice} holds, \Cref{thm:opt_admittednumber_firstprice} establishes that admitting all bidders is also optimal for first-price auctions. Consequently, the revenue equivalence theorem \citep{myerson_1981_optimal_auction} implies that 
\[
\R^{\FP}_\ast(F, m) = \R^{\SP}_\ast( F, m).
\]
Therefore, all-pay auctions also dominate first-price auctions.

\medskip
\noindent\underline{\textit{Proof of (ii).}}

\medskip

By \Cref{thm:opt_admittednumber_allpay}(i), we know that admitting only two bidders is optimal in all-pay auctions for the perfect predictor. Combined with the previous argument, we have
\begin{align*}
    \R^{\AP}_\ast(F,m) > \R^{\FP}_\ast(F,m) = \R^{\SP}_\ast(F,m).
\end{align*}
In this case, we have
\begin{align*}
 \R_\ast^{\AP}(F,m)
 &= \R^{\AP}(n=2;F,m) = \int_0^1 x \cdot\frac{2}{J(F_1(x),n=2,m)} \cdot \left(1-G^{\mar}(x; n=2,\PP)\right)  \, dF_1^{m-1}(x),
\end{align*}
and
\begin{align*}
\R^{\FP}_\ast(F,m+1) = \R^{\SP}_\ast(F,m+1)= 
 \R^{\SP}(n=m+1;F,m+1)= \int_0^1 (m+1) \cdot x \cdot [1-F_1(x)]\, dF_1^m(x).
\end{align*}
Define
\[
Q(x;m) := (m-1) F_1^{m-2}(x) f_1(x)\cdot\frac{2x}{J(F_1(x),n=2,m)} \cdot \left(1-G^{\mar}(x; n=2,\PP)\right),
\]
and
\[
\Omega(x;m) := m(m+1)x\cdot F_1^{m-1}(x)[1-F_1(x)]f_1(x).
\]
We aim to prove the inequality
\[
\lim_{m\rightarrow\infty}\int_0^1 Q(x;m)\,dx \geq \lim_{m\rightarrow\infty}\int_0^1 \Omega(x;m)\,dx.
\]
To that end, observe that for any \( x \in (0,1) \), the ratio of the integrands can be expressed as
\begin{align*}
\frac{Q(x;m)}{\Omega(x;m)} 
&= \frac{(m-1) F_1^{m-2}(x) f_1(x)\cdot \frac{2x}{J(F_1(x),n=2,m)} \cdot \left[1 - G^{\text{mar}}(x; n=2, \PP)\right]}{m(m+1)x \cdot F_1^{m-1}(x)[1 - F_1(x)] f_1(x)} \\
&= \frac{(m-1)\cdot \frac{2}{J(F_1(x),n=2,m)} \cdot \left[1 - G^{\text{mar}}(x; n=2, \PP)\right]}{m(m+1) \cdot F_1(x)[1 - F_1(x)]} \\
&= \frac{\frac{2}{J(F_1(x),n=2,m)} \cdot \left[1 - G^{\text{mar}}(x; n=2, \PP)\right]}{m [1 - F_1(x)]} \cdot \frac{m - 1}{(m + 1)F_1(x)} \\
&\overset{(a)}{>} \frac{m - 1}{(m + 1)\cdot F_1(x)}.
\end{align*}
Inequality \( (a) \) follows from the proof of \Cref{thm:opt_admittednumber_allpay}(i), which establishes that the term
\begin{align}
\label{eq:pf_rev_ranking_large_enough_allpay}
 \frac{n}{J(F_1(x), n, m)} \cdot \left[1 - G^{\text{mar}}(x; n, \PP)\right]   
\end{align}
is strictly decreasing in \( n \) for all \( x \in (0,1) \). In particular, for \( n = 2 \), \eqref{eq:pf_rev_ranking_large_enough_allpay} is strictly larger than the value it attains at \( n = m \), where it equals \( m[1 - F_1(x)] \).

Assume that \( \Omega(x;m) \) satisfies the {dominated convergence condition}, i.e., there exists an integrable function \( \Omega(x) \) such that 
\begin{align}
\label{eq:dominated_ convergence_condition}
\Omega(x;m) \leq \Omega(x), \quad \forall m \geq 2,\ x \in [0,1].
\end{align}
Then the product $\Omega(x;m)\cdot \frac{m-1}{(m+1)\cdot F_1(x)}$ also satisfies dominated convergence condition. Therefore, we obtain
\begin{align*}
\lim_{m\rightarrow\infty}\int_0^1 Q(x;m)\,dx 
&\overset{(b)}{\geq} \lim_{m\rightarrow\infty}\int_0^1 \Omega(x;m)\cdot \frac{m-1}{(m+1)\cdot F_1(x)}\,dx \\
&\overset{(c)}{=} \int_0^1 \lim_{m\rightarrow\infty} \left[\Omega(x;m)\cdot \frac{m-1}{(m+1)\cdot F_1(x)}\right]\,dx \\
&= \int_0^1 \lim_{m\rightarrow\infty} \left[\Omega(x;m)\cdot \frac{1}{F_1(x)}\right]\,dx \\
&\geq \int_0^1 \lim_{m\rightarrow\infty} \Omega(x;m)\,dx \\
&\overset{(d)}{=} \lim_{m\rightarrow\infty}\int_0^1 \Omega(x;m)\,dx,
\end{align*}
as desired. Here, step \( (b) \) follows from inequality \( (a) \), and steps \( (c) \) and \( (d) \) follow by the dominated convergence theorem.
\end{proof}

\begin{proof}[\textbf{Proof of \Cref{cor:Predictions vs. Negotiations}}]
This directly follows from \Cref{thm:rev_ranking}(ii) and \cite{bulow_1994_auctions_negotiations}, see the discussions before \Cref{cor:Predictions vs. Negotiations}.
\end{proof}

\begin{proof}[\textbf{Proof of \Cref{prop:rev_comparison_given_admitted_number}}] We prove each result separately.

\medskip
\noindent\underline{\textit{Proof of (i).}}
We first show that $\R^{\FP}(n;\PP,m)=\R^{\SP}(n;\PP,m)$ for every $n\in[2,m]$.
By \Cref{thm:opt_admittednumber_secondprice}(ii) and \Cref{thm:rev_prediction_firstprice}(ii), the seller’s expected revenue under either a second-price or a first-price auction is independent of the admitted number of bidders $n\in[2,m]$. In particular, within each format, the revenue with \textit{any} $n$ equals the revenue when all $m$ bidders are admitted. When $n=m$, bidders’ valuations remain independent, so the standard revenue-equivalence theorem applies and yields $\R^{\FP}(m;\PP,m)=\R^{\SP}(m;\PP,m)$ \citep{myerson_1981_optimal_auction}. Combining these observations shows that $\R^{\FP}(n;\PP,m)=\R^{\SP}(n;\PP,m)$ for all $n\in[2,m]$.

We now show that $\R^{\AP}(n; \PP,m) \geq \R^{\FP}(n; \PP,m)$. 
From the proof of \Cref{thm:opt_admittednumber_allpay}(i), $\R^{\AP}(n; \PP,m)$ is non-increasing in $n$ over $[2,m]$. 
Thus,
\[
    \R^{\AP}(n; \PP,m) \ge \R^{\AP}(m; \PP,m) = \R^{\FP}(m; \PP,m),
\]
where the equality again follows from the revenue-equivalence theorem \citep{myerson_1981_optimal_auction}.

\medskip
\noindent\underline{\textit{Proof of (ii).}} 
For $v_i,\tilde{v}_i \in [0,1]$, define
\[
\Xi^{\SP}(\tilde{v}_i,v_i) 
= \mathbb{E}\left[\sigma^{\SP}
     \left(\max_{j\in \I_{-i}}V_j; n, F\right) \; \Big| \; v_i, \max_{j\in \I_{-i}}V_j \leq \tilde{v}_i\right],
\]
and let
\[
\Xi_2^{\SP}(\tilde{v}_i,v_i) 
= \frac{\partial}{\partial v_i} \, \Xi^{\SP}(\tilde{v}_i,v_i).
\]
By Proposition~7.1 of \cite{krishna2009auction}, if $\Xi_2^{\SP}(\tilde{v}_i,v_i) \geq 0$, then the second-price auction yields (weakly) higher revenue than the first-price auction; conversely, if $\Xi_2^{\SP}(\tilde{v}_i,v_i) \leq 0$, the first-price auction yields (weakly) higher revenue.  
Hence, it suffices to establish that
\begin{align*}
&\frac{\partial}{\partial x}\left(\frac{H(x \mid v_i; n, F)}{H_2(x \mid v_i; n, F)}\right) \leq 0, \quad x \in [0,v_i] 
\quad\Longrightarrow\quad \Xi_2^{\SP}(v_i,v_i) \geq 0,\\
&\frac{\partial}{\partial x}\left(\frac{H(x \mid v_i; n, F)}{H_2(x \mid v_i; n, F)}\right) \geq 0, \quad x \in [0,v_i] 
\quad\Longrightarrow\quad \Xi_2^{\SP}(\tilde{v}_i,v_i) \leq 0.
\end{align*}
By Bayes’ formula,
\begin{align*}
\Pr\left(\max_{j\in \I_{-i}}V_j \leq x\;\Big|\; v_i, \max_{j\in \I_{-i}}V_j \leq \tilde{v}_i\right)
&= \frac{\Pr\left(\max_{j\in \I_{-i}}V_j \leq \min\{x,\tilde{v}_i\} \;\Big|\; v_i\right)}{\Pr\left(\max_{j\in \I_{-i}}V_j \leq \tilde{v}_i\;\Big|\; v_i\right)} \\
&\quad= \frac{H(\min\{x,\tilde{v}_i\} \mid v_i; n, F)}{H(\tilde{v}_i \mid v_i; n, F)} \\
&\quad= \1\left\{x < \tilde{v}_i\right\} \cdot \frac{H(x \mid v_i; n, F)}{H(\tilde{v}_i \mid v_i; n, F)} \;+\; \1\left\{x \geq \tilde{v}_i\right\}.
\end{align*}
Therefore,
\begin{align*}
\Xi^{\SP}(\tilde{v}_i,v_i) 
&= \mathbb{E}\left[\sigma^{\SP}\left(\max_{j\in \I_{-i}}V_j; n,F\right) \; \Big| \; v_i, \max_{j\in \I_{-i}}V_j \leq \tilde{v}_i\right] \\
&= \int_0^{\tilde{v}_i} \sigma^{\SP}(x; n,F) \cdot \frac{h(x \mid v_i; n, F)}{H(\tilde{v}_i \mid v_i; n, F)} \; dx \\
&= \int_0^{\tilde{v}_i} x \cdot \frac{h(x \mid v_i; n, F)}{H(\tilde{v}_i \mid v_i; n, F)} \; dx \\
&\overset{(a)}{=} \tilde{v}_i - \int_0^{\tilde{v}_i} \frac{H(x \mid v_i; n, F)}{H(\tilde{v}_i \mid v_i; n, F)} \; dx,
\end{align*}
where (a) uses the integral by parts. Differentiating with respect to $v_i$ yields
\begin{align*}
\Xi^{\SP}_2(\tilde{v}_i,v_i) 
&= - \int_0^{\tilde{v}_i} \frac{\partial}{\partial v_i} \left( \frac{H(x \mid v_i; n, F)}{H(\tilde{v}_i \mid v_i; n, F)} \right) dx \\
&= - \int_0^{\tilde{v}_i} \frac{H_2(x \mid v_i; n, F) \cdot H(\tilde{v}_i \mid v_i; n, F) - H(x \mid v_i; n, F) \cdot H_2(\tilde{v}_i \mid v_i; n, F)}{H^2(\tilde{v}_i \mid v_i; n, F)} \; dx.
\end{align*}
It follows that
\begin{align*}
\frac{\partial}{\partial x}\left(\frac{H(x \mid v_i; n, F)}{H_2(x \mid v_i; n, F)}\right) \leq 0, \quad x\in [0,v_i] &\quad \Longleftrightarrow\quad \frac{H(v_i \mid v_i; n, F)}{H_2(v_i \mid v_i; n, F)} \leq \frac{H(x \mid v_i; n, F)}{H_2(x \mid v_i; n, F)}, \quad \forall\; 0 \leq x \leq v_i \leq 1 \\
&\quad \;\Longrightarrow\quad \Xi^{\SP}_2(v_i,v_i) \geq 0,
\end{align*}
and
\begin{align*}
\frac{\partial}{\partial x}\left(\frac{H(x \mid v_i; n, F)}{H_2(x \mid v_i; n, F)}\right) \geq 0, \quad x\in [0,v_i] 
&\quad \Longleftrightarrow\quad \frac{H(v_i \mid v_i; n, F)}{H_2(v_i \mid v_i; n, F)} \geq \frac{H(x \mid v_i; n, F)}{H_2(x \mid v_i; n, F)}, \quad \forall\; 0 \leq x \leq v_i \leq 1 \\
&\quad \;\Longrightarrow\quad \Xi^{\SP}_2(v_i,v_i) \leq 0.
\end{align*}

\medskip
\noindent\underline{\textit{Proof of (iii).}} The result follows directly from Proposition 7.2 of \cite{krishna2009auction}.
\end{proof}

\subsection{Proofs for \Cref{sec:extension}}

\begin{proof}[\textbf{Proof of \Cref{prop:reserve_prices}}]
The results for second-price and first-price auctions follow from the discussions after \Cref{prop:reserve_prices}. We now provide the proof for all-pay auctions; the null predictor case is straightforward and thus omitted.

The cutoff valuation $\tau(n,F)$, making a bidder with valuation $\tau(n,F)$ indifferent between bidding zero and bidding the reserve price $r$, solves the equation (about $x\in[0,1]$):
\begin{align*}
x \cdot H(x\mid x;n,F) - r = 0,
\end{align*}
where the first term represents the expected utility from winning, and the second term is the required payment by a bidder with valuation $\tau(n,F)$ (who bids the reserve price $r$).

Clearly, $\tau(n,F)$ decreases as the winning probability $H(x\mid x;n,F)$ increases (element-wise). Under the perfect predictor, we have: for any $x\in[0,1]$,
\begin{align*}
H(x\mid x;n,\PP) = \frac{F_1^{m-1}(x)}{J(F_1(x),n,m)} \geq \frac{F_1^{m-1}(x)}{J(F_1(x),n+1,m)} = H(x\mid x;n+1,\PP).
\end{align*}
The inequality holds since $J(F_1(x),n,m)$ is increasing in $n \in [2,m]$ for all $x$. Thus, under the perfect predictor, the cutoff valuation $\tau(n,F)$ is increasing with the admitted number.

Based on the proof of \Cref{thm:opt_admittednumber_allpay}, even without a reserve price (where the cutoff value remains constant), the equilibrium strategy still decreases with the admitted number under the perfect predictor. The increasing property of the cutoff value further strengthens this monotonicity. Combining this argument with \Cref{thm:opt_admittednumber_allpay}(i), we conclude that admitting only two bidders maximizes revenue in all-pay auctions under the perfect predictor for any reserve price $r \geq 0$.
\end{proof}

\begin{proof}[\textbf{Proof of \Cref{prop:joint_optimality}}]
See the discussions after \Cref{prop:joint_optimality}.  
\end{proof}

\begin{proof}[\textbf{Proof of \Cref{prop:uniform_auctions}}]
In uniform-price auctions, truthful bidding constitutes an equilibrium strategy. Consequently, the seller’s expected revenue equals \( K \) times the expected \((K+1)\)st-highest valuation among the bidders. By \Cref{lem:FOSD_number}(ii), the distribution of the \((K+1)\)st-highest valuation increases in the number of admitted bidders \( n \) in the sense of first-order stochastic dominance. Therefore, the expected revenue is increasing in \( n \), and it is optimal to admit all bidders.

Since different predictors yield the same marginal CDFs and, when \( n = m \), the posterior reduces to the independent case, admitting all players results in the same expected revenue for different predictors. Furthermore, for any given admitted number, by \Cref{lem:FOSD_number}(iii), we conclude that the $(K+1)$st-highest valuation increases with prediction accuracy. Therefore, the revenue loss due to admitting fewer bidders is (weakly) decreasing in the prediction accuracy.
\end{proof}

\begin{proof}[\textbf{Proof of \Cref{prop:all-pay_highestbid}}] We prove each result separately.

\medskip
\noindent\underline{\textit{Proof of (i).}}

\medskip

We use the following lemma to facilitate our proof.
\begin{lemma}
\label{lem:J_increasing}
For any $x\in [0,1]$, $J(x,n,m)$ is increasing with the admitted number $n\in[2,m]$.
\end{lemma}
\begin{proof}
This result originally appears in Lemma 1(ii) of \cite{sun_2024_contests}. Here, we provide a self-contained and simplified proof. As noted in the literature review, \cite{sun_2024_contests} only consider all-pay auctions with a perfect predictor and focus solely on the highest bid. In contrast, we study general predictors, examine various auction formats, and also consider revenue maximization.

 By \eqref{eq:def_J_integral}, we have for $3\leq n\leq m$, 
\begin{align}
     J(x,n-1, m) &= \C^{m-1}_{n-2} (m-n+1)\int_0^{x}t^{m-n}(1-t)^{n-2}dt =\C^{m-1}_{n-1}(n-1)\int_0^{x}t^{m-n}(1-t)^{n-2}dt. \nonumber 
\end{align}
Hence $\forall x \in [0,1]$,
\begin{align}
     J(x,n-1, m)- J(x,n, m)&=-\C^{m-1}_{n-1}x^{m-n}(1-x)^{n-1}\leq 0\nonumber, 
\end{align}
and the inequality is strict when $x\in (0,1)$. This implies that $J(x,n, m)$ is weakly increasing with $n\in [2,m]$ for $x\in \{0,1\}$, and strictly increasing for $x\in (0,1)$.
\end{proof}

For a predictor $F$, define the expected highest bid with $n$ admitted bidders as
\begin{align*}
 \HB^{\AP}(n;F) = \mathbb{E}_{V_i\sim G^{\lar}(\cdot ; n,F)}[\sigmaAP(V_i;n,F)].
\end{align*}
By \Cref{lem:FOSD_number}(ii)(a), we know that  $G^{\lar}(\cdot; n,\PP) =F^m(\cdot)$.
Combining \eqref{eq:SSM_equilibrium_gamma_1_all_pay} about the equilibrium strategy $\sigmaAP(V_i;n,\PP)$, the expected highest bid becomes
\begin{align*}
 \HB^{\AP}(n;\PP) & = \mathbb{E}_{V_i\sim G^{\lar}(\cdot; n,\PP)}[\sigmaAP(V_i;n,\PP)]\\
 &= \int_0^1 \int_0^{v_i} \frac{x}{J(F_1(x),n,m)}\, dF_1^{m-1}(x)\, dG^{\lar}(v_i;n,\PP)\\[1mm]
 &= \int_0^1 \int_0^{v_i} \frac{x}{J(F_1(x),n,m)}\, dF_1^{m-1}(x)\, dF_1^{m}(v_i).
\end{align*}
By \Cref{lem:J_increasing}, $J(F_1(x),n,m)$ is strictly increasing with the admitted number $n\in [2,m]$ when $x\in (0,1)$, and non-decreasing when $x\in \{0,1\}$. Hence, we conclude that $\HB^{\AP}(n;\PP)$ is strictly decreasing with $n$.
As a result, the unique optimal solution is $n^\ast=2$ in terms of the expected highest bid. 

\medskip
\noindent\underline{\textit{Proof of (ii).}}

\medskip

For the null predictor, the expected highest bid is
\begin{align*}
 \HB^{\AP}(n;\PP) & = \mathbb{E}_{V_i\sim G^{\lar}(\cdot; n,\PP)}[\sigmaAP(V_i;n,\PP)]\\
 &= \int_0^1 \int_0^{v_i} x\, dF^{n-1}(x)\, dG^{\lar}(v_i;n,\PP)\\[1mm]
 &\overset{(a)}{=} \int_0^1 x \left(1-G^{\lar}(x; n,\PP)\right)\, dF_1^{n-1}(x)\\
 &= \int_0^1 x \left(1-F_1^n(x)\right)\, dF_1^{n-1}(x)\\
 & = \int_0^1 xdF_1^{n-1}(x) - \frac{(n-1)}{2n-1} \int_0^1 xdF_1^{2n-1}(x)\\
 & = \frac{(n-1)}{2n-1}\int_0^1 F_1^{2n-1}(x)dx - \int_0^1 F_1^{n-1}(x)dx + \frac{n}{2n-1},
\end{align*}
where (a) holds by the Fubini theorem. 
For the power-law prior distribution $F_1(x) = x^c$, we have
\begin{align*}
    \HB^{\AP}(n;\NP) = &\frac{n-1}{(2n-1)[(2n-1)c+1]} - \frac{1}{(n-1)c+1} + \frac{n}{2n-1} \\
    =& \frac{n(n-1)c^2}{[(2n-1)c+1][(n-1)c+1]}.
\end{align*}
Hence 
\begin{align*}
    \frac{ \HB^{\AP}(n+1;\NP)}{ \HB^{\AP}(n;\NP)} = \frac{(n+1)[(2n-1)c+1][(n-1)c+1]}{(n-1)(nc+1)[(2n+1)c+1]}.
\end{align*}
Notice that
\begin{align*}
 &\frac{ \HB^{\AP}(n+1;\NP)}{ \HB^{\AP}(n;\NP)} \geq 1 \quad
 \Longleftrightarrow
 \quad -(n-1)c^2 + (3n-1)c + 2 \geq 0 \quad
 \Longleftrightarrow\quad  (3c - c^2)n + c^2-c+2\geq 0.
\end{align*}
We have three different situations about the monotonicity of $\HB^{\AP}(n;\NP)$ with respect to the admitted number $n$.
\begin{enumerate}[(a)]
    \item For $0<c\leq 3$, we have $(3c - c^2)n + c^2-c+2\geq 0 $  for any integer $n\geq 2$. 
    Thus, $\HB^{\AP}(n;\NP)$ is increasing with the admitted number $n\in [2,m]$, which implies that $n^* = m$.
    \item For $3 \leq c \leq \frac{5+\sqrt{33}}{2}$, we have $(3c - c^2)n + c^2-c+2\geq 0$ if $2 \leq n \leq \frac{c^2-c+2}{c^2-3c}$, and $(3c - c^2)n + c^2-c+2\leq 0$ if $ n \geq \frac{c^2-c+2}{c^2-3c}$.
    Thus, $\HB^{\AP}(n;\NP)$ firstly increases with $n\in \left[2,m\wedge \frac{c^2-c+2}{c^2-3c}\right]$, and then decreases with $n\in \left(m\wedge \frac{c^2-c+2}{c^2-3c},m\right]$ if non-empty.
    As a result, 
    \begin{align*}
     \textrm{$n^* = m \wedge \bigg\lceil{\frac{c^2-c+2}{c^2-3c}}\bigg\rceil$ or $n^* = m \wedge \bigg\lfloor{\frac{c^2-c+2}{c^2-3c}}\bigg\rfloor$.}   
    \end{align*}
    \item For $c \geq \frac{5+\sqrt{33}}{2}$, we have $(3c - c^2)n + c^2-c+2\leq 0$ for any admitted number $n\geq 2$.
    In this situation,  $\HB^{\AP}(n;\NP)$ is decreasing with the admitted number $n\in [2,m]$.
    As a result, $n^\ast=2$.
\end{enumerate}
These two parts complete the proof.
\end{proof}

\subsection{Proof of \Cref{app_sec:auxiliary_results}}

\begin{proof}[\textbf{Proof of \Cref{exm:hpm=3}}]
When $n=m$, the posterior remains independent. Thus, the condition in \Cref{thm:SSM_firstprice}(i) holds.
We focus on the case of $n = 2$ in what follows.
By (a) of \Cref{cor:special_case_H_h}(ii), for any $v_i^\prime\leq \tilde{v}_i\leq v_i$, we have
\begin{align*}  \frac{h(\tilde{v}_i\mid v_i^\prime;n=2,\HB(\gamma))}{H(\tilde{v}_i\mid v_i^\prime;n=2,\HB(\gamma))} 
\leq \frac{h(\tilde{v}_i\mid \tilde{v}_i;n=2,\HB(\gamma))}{H(\tilde{v}_i\mid \tilde{v}_i;n=2,\HB(\gamma))} 
\geq 
 \frac{h(\tilde{v}_i\mid v_i;n=2,\HB(\gamma))}{H(\tilde{v}_i\mid v_i;n=2,\HB(\gamma))}.
\end{align*}
Hence when $\tilde{v}_i \geq v_i$, $\FP(\tilde{v}_i,v_i;n,\HB(\gamma)) $ is non-positive. We are left to show that for $F_1(x) = x^c$ $(c\leq 1)$, $\FP(\tilde{v}_i,v_i;n,\HB(\gamma)) $ is non-negative when $\tilde{v}_i \leq v_i$. To see this, 
observe that 
\begin{align*}
  \int_0^{\tilde{v}_i}L(x\mid \tilde{v}_i;n,\HB(\gamma))dx \leq \tilde{v}_i,
\end{align*}
where the inequality holds since $L(x\mid \tilde{v}_i;n,\HB(\gamma))\leq 1$ for all $x,\tilde{v}_i\in [0,1]$ by the definition \eqref{eq:L}.
Hence, for $n=2$, we have
\begin{align*}
  &\FP(\tilde{v}_i,v_i;n,\HB(\gamma))\\  =&h(\tilde{v}_i\mid v_i;n,\HB(\gamma))  \left(v_i - \tilde{v}_i\right)\\
  & +H(\tilde{v}_i\mid v_i;n,\HB(\gamma))  \int_0^{\tilde{v}_i}L(x\mid \tilde{v}_i;n,\HB(\gamma))dx
  \cdot \left(
 \frac{h(\tilde{v}_i\mid v_i;n,\HB(\gamma))}{H(\tilde{v}_i\mid v_i;n,\HB(\gamma))}  - \frac{h(\tilde{v}_i\mid \tilde{v}_i;n,\HB(\gamma))}{H(\tilde{v}_i\mid \tilde{v}_i;n,\HB(\gamma))}
 \right) \\
 \geq & h(\tilde{v}_i\mid v_i;n,\HB(\gamma))  \left(v_i - \tilde{v}_i\right)
 +H(\tilde{v}_i\mid v_i;n,\HB(\gamma)) \cdot \tilde{v}_i \cdot \left(
 \frac{h(\tilde{v}_i\mid v_i;n,\HB(\gamma))}{H(\tilde{v}_i\mid v_i;n,\HB(\gamma))}  - \frac{h(\tilde{v}_i\mid \tilde{v}_i;n,\HB(\gamma))}{H(\tilde{v}_i\mid \tilde{v}_i;n,\HB(\gamma))}
 \right) \\
 \geq & H(\tilde{v}_i\mid v_i;n,\HB(\gamma)) \left( \frac{h(\tilde{v}_i\mid {v}_i;n,\HB(\gamma))}{H(\tilde{v}_i\mid {v}_i;n,\HB(\gamma))}{v}_i -  \frac{h(\tilde{v}_i\mid \tilde{v}_i;n,\HB(\gamma))}{H(\tilde{v}_i\mid \tilde{v}_i;n,\HB(\gamma))}\tilde{v}_i\right)\\
 \geq & 0,
\end{align*}
where the last inequality is by (b) of \Cref{cor:special_case_H_h}(ii). This completes the proof.

\end{proof}

\section{Further Results about the Hallucinatory Predictor}
\label{app_subsec:beliefs}
For simplicity, we use \(\gamma\) in place of \(\HB(\gamma)\) in relevant functions, e.g., 
\(\beta\), \(\psi\), \(\kappa\), \(g\), \(H\), \(h\), \(\sigma^{\mathrm{FP}}\), \(\sigma^{\mathrm{AP}}\), 
\(G^{\mathrm{lar}}\), and \(G^{\mathrm{mar}}\). 
Here, \(\gamma = 0\) represents the null predictor, while \(\gamma = 1\) corresponds to the perfect predictor.

\begin{lemma} 
\label{lem:closedforms_admission_prob}
The admission probability has the additive property, i.e.,  
\begin{align*}
 \psi(v;n,\gamma)  & = \hat{\psi}_0(n,\gamma)+\sum_{k=1}^n \hat{\psi}_k(v_{(k)};n,\gamma),
\end{align*}
where $v_{(k)}$ is the $k^{\mathsf{th}}$ smallest element among all elements in the vector $v$, $\hat{\psi}_0(n,\gamma) =  \frac{(1-\gamma)^n}{\C^{m}_n}$, and
\begin{align*}
 \hat{\psi}_k(v_{(k)};n,\gamma) =& (m-n) F_1^{m-n} \left(v_{(k)}\right)\mathds{1}\{k\leq n-1\}\sum_{j=1}^{n-k} \C^{n-k}_{j-1}\gamma^j
 (1-\gamma)^{n-j}\sum_{i=0}^{n-j}\frac{(-1)^{i}\C^{n-j}_i}{m-n+i}F_1^{i}\left(v_{(k)}\right)\\
 &+ (m-n) \gamma^{n-k+1}
 (1-\gamma)^{k-1}F_1^{m-n}\left(v_{(k)}\right)\sum_{j=0}^{k-1}\frac{(-1)^{j}\C^{k-1}_j}{m-n+j}F_1^{j}\left(v_{(k)}\right).  
\end{align*}
The normalizing term $\kappa(v_i;n,\gamma)$ is given by
\begin{align*}
 \kappa(v_i;n,\gamma)  = 1/\left( \frac{1-\gamma}{\C^m_n} + (m-n)\gamma \cdot \sum_{j=0}^{n-1}\frac{(-1)^j\C^{n-1}_j}{m-n+j} F_1^{m-n+j}(v_i)\right) . 
\end{align*}
\end{lemma}
Given \Cref{lem:closedforms_admission_prob}, we can directly obtain the cases $\gamma = 0$ and $\gamma = 1$.

\begin{corollary}\label{cor:special_gamma}
 When $\gamma=1$, for any $v\in[0,1]^n$, $\psi(v;n,\gamma=1) = F_1^{m-n}(v_{(1)})$; for any $x\in[0,1]$,  
 \begin{align*}
     \frac{1}{\kappa(x; n, \gamma = 1)} =& (m-n) \sum_{i=0}^{n-1}\frac{(-1)^i\C^{n-1}_i}{m-n+i} F_1^{m-n+i}(x)\\
     =&  (n-1)\int_0^{F_1(x)}t^{m-n}(1-t)^{n-2}dt + F_1^{m-n}(x)(1-F_1(x))^{n-1}.
 \end{align*}
 When $\gamma=0$, for any $v\in[0,1]^n$ and for any $x\in [0,1]$, we have
 \begin{align*}
   \psi(v;n,\gamma=0) = \frac{1}{\C^m_n},\quad \textrm{and} \quad
   \kappa(x;n,\gamma=0) = \C^m_n.
 \end{align*}
\end{corollary}
The following proposition provides further properties of the admission probability.

\begin{proposition}
\label{prop:admissionprob_property}
 The admission probability $\psi(v;n,\gamma)$ has the following properties: 
 \begin{enumerate}[(i)]
    \item $\psi(v;n,\gamma)$ is symmetric in $v\in [0,1]^n$ and is differentiable almost everywhere.

     \item For any $i, j\in \I$ and $i\neq j$, $\frac{\partial^2  \psi}{\partial v_i\partial v_j}$ exists almost everywhere and equals zero whenever it exists.
    
     \item Local Supermodularity: $\psi$ is supermodular in any domain $\mathcal{V}\subset 
     [0,1]^n$ in which for any $v,v^\prime \in \mathcal{V}$, the order of elements in $v$ is the same as the order of elements in $v^\prime$.

     \item When $\gamma=1$, $\psi$ is both supermodular and log-supermodular in $v\in [0,1]^n$.

 \end{enumerate}
\end{proposition}

For numerical efficiency of calculating the expected revenue and expected highest bid, we now give the closed form of $H(\cdot\mid v_i;n,\gamma)$, its partial derivative $h(\cdot\mid v_i;n,\gamma)$, marginal CDF $G^{\mar}(\cdot; n,\gamma)$ of the joint distribution $g(\cdot;n,\gamma)$ and the CDF $G^{\lar}(\cdot;n,\gamma)$ of $\max_{i\in\I}V_i$, where $V\sim g(\cdot;n,\gamma)$.

\begin{lemma}
 \label{lem:H_h}
(i). When $x \in [0,v_i]$,
\begin{align*}
\frac{H(x\mid v_i;n,\gamma)}{\kappa(v_i;n,\gamma)} &=   
 (m-n)F_1^{m-1}(x)\sum_{k=1}^{n-1}  \sum_{j=1}^{n-k}\C^{n-k}_{j-1}\gamma^j
 (1-\gamma)^{n-j}\sum_{i=0}^{n-j}\frac{(-1)^{i}\C^{n-j}_i\C^{n-1}_{n-k}}{(m-n+i)\C^{m+i-1}_{n-k}}F_1^i\left(x\right)\\
 & \quad + (m-n)F_1^{m-1}(x)\sum_{k=1}^{n-1} \gamma^{n-k+1}
 (1-\gamma)^{k-1}\sum_{j=0}^{k-1}\frac{(-1)^{j}\C^{k-1}_j\C^{n-1}_{n-k}}{(m-n+j)\C^{m+j-1}_{n-k}}F_1^{j}\left(x\right) \\
 &\quad + F_1^{n-1}(x)\left((m-n)\gamma(1-\gamma)^{n-1}F_1^{m-n}(v_i)\sum_{j=0}^{n-1}\frac{(-1)^j\C^{n-1}_j}{m-n+j}F_1^j(v_i) + \frac{(1-\gamma)^n}{\C^m_n}\right).
\end{align*}
When $x \in [v_i, 1]$, 
{\footnotesize{
\begin{align*}
&\frac{H(x\mid v_i;n,\gamma)}{\kappa(v_i;n,\gamma)}\\
=& \frac{(1-\gamma)^n}{\C^m_n}F_1^{n-1}(x) + (m-n)\sum_{k=1}^{n-1} \C^{n-1}_{k-1}F_1^{m-n+k-1}(v_i)(F_1(x) - F_1(v_i))^{n-k} \sum_{j=1}^{n-k} \C^{n-k}_{j-1}\gamma^j
 (1-\gamma)^{n-j}\sum_{t=0}^{n-j}\frac{(-1)^{t}\C^{n-j}_t}{m-n+t}F_1^{t}\left(v_{i}\right)\\
 & + (m-n)\sum_{k=1}^{n}\C^{n-1}_{k-1} \gamma^{n-k+1}
 (1-\gamma)^{k-1}F_1^{m-n+k-1}\left(v_{i}\right)(F_1(x) - F_1(v_i))^{n-k}\sum_{j=0}^{k-1}\frac{(-1)^{j}\C^{k-1}_j}{m-n+j}F_1^{j}\left(v_{i}\right)\\
 & + (m-n) \sum_{k=2}^{n}\C^{n-1}_{k-1}\sum_{s=1}^{k-1}\sum_{j=1}^{n-s} \C^{n-s}_{j-1}\gamma^j
 (1-\gamma)^{n-j}\sum_{t=0}^{n-j}\frac{(-1)^{t}\C^{n-j}_t\C^{k-1}_{k-s}}{(m-n+t)\C^{m-n+t+k-1}_{k-s}}F_1^{m-n+t+k-1}(v_i)(F_1(x) - F_1(v_i))^{n-k}\\
 & + (m-n)\sum_{k=2}^{n} \C^{n-1}_{k-1}\sum_{s=1}^{k-1}\gamma^{n-s+1}
 (1-\gamma)^{s-1}\sum_{j=0}^{s-1}\frac{(-1)^{j}\C^{s-1}_j\C^{k-1}_{k-s}}{(m-n+j)\C^{m-n+j+k-1}_{k-s}}F_1^{m-n+j+k-1}(v_i)(F_1(x) - F_1(v_i))^{n-k}\\
 & + (m-n)\sum_{k=1}^{n-2} \C^{n-1}_{k-1}\sum_{s=k+1}^{n-1}\sum_{j=1}^{n-s} \C^{n-s}_{j-1}\gamma^j
 (1-\gamma)^{n-j}\sum_{t=0}^{n-j}\frac{(-1)^{t}\C^{n-j}_t}{m-n+t}\sum_{r=0}^{m-n+t}\frac{\C^{n-k}_{n-s+1}\C^{m-n+t}_r}{\C^{n+r-k}_{n-s+1}}F_1^{m-n+t+k-r-1}(v_i)(F_1(x) - F_1(v_i))^{n+r-k}\\
 & + (m-n)\sum_{k=1}^{n-1}\C^{n-1}_{k-1} \sum_{s=k+1}^{n}\gamma^{n-s+1}
 (1-\gamma)^{s-1}\sum_{j=0}^{s-1}\frac{(-1)^{j}\C^{s-1}_j}{m-n+j}\sum_{r=0}^{m-n+j}\frac{\C^{n-k}_{n-s+1}\C^{m-n+j}_r}{\C^{n+r-k}_{n-s+1}}F_1^{m-n+j+k-r-1}(v_i)(F_1(x) - F_1(v_i))^{n+r-k}.
\end{align*}}
}
(ii). When $x \in [0,v_i]$,
 \begin{align*}
 \frac{h(x\mid v_i;n,\gamma)}{\kappa(v_i;n,\gamma)} &= (m-n)F_1^{m-2}(x)f_1(x)\sum_{k=1}^{n-1}  \sum_{j=1}^{n-k}\C^{n-k}_{j-1}\gamma^j
 (1-\gamma)^{n-j}\sum_{i=0}^{n-j}\frac{(-1)^{i}(m+i-1)\C^{n-j}_i\C^{n-1}_{n-k}}{(m-n+i)\C^{m+i-1}_{n-k}}F_1^i\left(x\right)\\
 & \quad + (m-n)F_1^{m-2}(x)f_1(x)\sum_{k=1}^{n-1} \gamma^{n-k+1}
 (1-\gamma)^{k-1}\sum_{j=0}^{k-1}\frac{(-1)^{j}(m+j-1)\C^{k-1}_j\C^{n-1}_{n-k}}{(m-n+j)\C^{m+j-1}_{n-k}}F_1^{j}\left(x\right) \\
 &\quad + (n-1)F_1^{n-2}(x)f_1(x)\left((m-n)\gamma(1-\gamma)^{n-1}F_1^{m-n}(v_i)\sum_{j=0}^{n-1}\frac{(-1)^j\C^{n-1}_j}{m-n+j}F_1^j(v_i) + \frac{(1-\gamma)^n}{\C^m_n}\right). 
 \end{align*}
When $x \in [v_i, 1]$, 
{\tiny {
\begin{align*}
&\frac{h(x\mid v_i;n,\gamma)}{\kappa(v_i;n,\gamma)}\\
=& \frac{(1-\gamma)^n}{\C^m_n}(n-1)F_1^{n-2}(x)f_1(x)\\
&+ (m-n)f_1(x)\sum_{k=1}^{n-1}\C^{n-1}_{k-1} F_1^{m-n+k-2}(v_i)(F_1(x) - F_1(v_i))^{n-k-1}\sum_{j=1}^{n-k} \C^{n-k}_{j-1}\gamma^j
 (1-\gamma)^{n-j}\sum_{t=0}^{n-j}\frac{(-1)^{t}(m-n+k+t-1)\C^{n-j}_t}{m-n+t}F_1^{t}\left(v_{i}\right)\\
 & + (m-n)f_1(x)\sum_{k=1}^{n}\C^{n-1}_{k-1} F_1^{m-n+k-2}(v_i)(F_1(x) - F_1(v_i))^{n-k-1}\gamma^{n-k+1}
 (1-\gamma)^{k-1}\sum_{j=0}^{k-1}\frac{(-1)^{j}(m-n+k+j-1)\C^{k-1}_j}{m-n+j}F_1^{j}\left(v_{i}\right)\\
 & + (m-n) f_1(x)\sum_{k=2}^{n}(n-k)\C^{n-1}_{k-1}\sum_{s=1}^{k-1}\sum_{j=1}^{n-s} \C^{n-s}_{j-1}\gamma^j
 (1-\gamma)^{n-j}\sum_{t=0}^{n-j}\frac{(-1)^{t}\C^{n-j}_t\C^{k-1}_{k-s}}{(m-n+t)\C^{m-n+t+k-1}_{k-s}}F_1^{m-n+t+k-1}(v_i)(F_1(x) - F_1(v_i))^{n-k-1}\\
 & + (m-n)f_1(x)\sum_{k=2}^{n} (n-k)\C^{n-1}_{k-1}\sum_{s=1}^{k-1}\gamma^{n-s+1}
 (1-\gamma)^{s-1}\sum_{j=0}^{s-1}\frac{(-1)^{j}\C^{s-1}_j\C^{k-1}_{k-s}}{(m-n+j)\C^{m-n+j+k-1}_{k-s}}F_1^{m-n+j+k-1}(v_i)(F_1(x) - F_1(v_i))^{n-k-1}\\
 & + (m-n)f_1(x)\sum_{k=1}^{n-2}\C^{n-1}_{k-1} \sum_{s=k+1}^{n-1}\sum_{j=1}^{n-s} \C^{n-s}_{j-1}\gamma^j
 (1-\gamma)^{n-j}\sum_{t=0}^{n-j}\frac{(-1)^{t}\C^{n-j}_t}{m-n+t}\sum_{r=0}^{m-n+t}\frac{(n+r-k)\C^{n-k}_{n-s+1}\C^{m-n+t}_r}{\C^{n+r-k}_{n-s+1}}F_1^{m-n+t+k-r-1}(v_i)(F_1(x) - F_1(v_i))^{n+r-k-1}\\
 & + (m-n)f_1(x)\sum_{k=1}^{n-1} \C^{n-1}_{k-1}\sum_{s=k+1}^{n}\gamma^{n-s+1}
 (1-\gamma)^{s-1}\sum_{j=0}^{s-1}\frac{(-1)^{j}\C^{s-1}_j}{m-n+j}\sum_{r=0}^{m-n+j}\frac{(n+r-k)\C^{n-k}_{n-s+1}\C^{m-n+j}_r}{\C^{n+r-k}_{n-s+1}}F_1^{m-n+j+k-r-1}(v_i)(F_1(x) - F_1(v_i))^{n+r-k-1}.
\end{align*}}
} 
\end{lemma}
Specifically, the closed forms and properties of 
$H(\cdot\mid v_i;n,\gamma = 1)$ and $h(\cdot\mid v_i;n,\gamma = 1)$
for some special cases are given as follows.

\begin{corollary}[Special Cases]\label{cor:special_case_H_h}
(i). When $\gamma = 1$, $H(x\mid v_i;n,\gamma = 1)$ and $h(x\mid v_i;n,\gamma = 1)$ are given by
    \begin{align*}
      &H(x\mid v_i;n,\gamma = 1) \\
      = & \mathds{1}\{x\leq v_i\}\kappa(v_i; n, \gamma = 1)\frac{1}{\C^{m-1}_{n-1}}F_1^{m-1}(x)\\
      &+ \mathds{1}\{x> v_i\} \kappa(v_i; n, \gamma = 1) \left((n - 1)\int_0^{F_1(v_i)}t^{n_1-n_2}\left(F_1(x)-t \right)^{n - 2}dt + F_1^{m-n}(v_i)(F_1(x) - F_1(v_i))^{n - 1} \right),
    \end{align*}
and
{\small {
    \begin{align*}
      &h(x\mid v_i;n,\gamma = 1) \\
      = & \mathds{1}\{x\leq v_i\}\kappa(v_i; n, \gamma = 1)\frac{m-1}{\C^{m-1}_{n-1}}F_1^{m-2}(x)f_1(x)\\
      &+ \mathds{1}\{x> v_i\} \kappa(v_i; n, \gamma = 1)(n - 1) \left((n-2)\int_0^{F_1(v_i)}t^{n_1-n_2}\left(F_1(x)-t \right)^{n - 3}dt + F_1^{m-n}(v_i)(F_1(x) - F_1(v_i))^{n - 2} \right) f_1(x).
    \end{align*}}}
\noindent (ii). When $m=3$ and $n=2$, $H\left(x\mid v_i;n = 2,\gamma\right)$ and $h\left(x\mid v_i;n = 2,\gamma\right)$ are given by
{\scriptsize{
\begin{align*} 
 &H\left(x\mid v_i;n = 2,\gamma\right) \\
 =& \mathds{1} \{x\leq v_i\}\kappa(v_i;n = 2,\gamma)
 \left(\frac{ (1-\gamma)^2}{3} F_1\left(x\right) + \gamma (1-\gamma)F_1\left(v_i\right)F_1\left(x\right)
 - \frac{\gamma(1-\gamma)}{2} F_1^2\left(v_i\right)F_1\left(x\right) + \frac{\gamma}{2}  F_1^2\left(x\right)  - \frac{\gamma(1-\gamma)}{6} F_1^3\left(x\right)\right)
 \\
 & + \mathds{1}\{x\leq v_i\}\kappa(v_i;n = 2,\gamma)\left(\frac{(1-\gamma)^2}{3}F_1\left(x\right) + \frac{\gamma(1-\gamma)}{2} F_1^2\left(x\right) - \frac{\gamma(1-\gamma)}{6} F_1^3\left(x\right) + \gamma F_1\left(v_i\right)F_1\left(x\right) - \frac{\gamma^2}{2} F_1^2\left(v_i\right) - \frac{\gamma (1-\gamma)}{2} F_1^2\left(v_i\right) F_1\left(x\right)\right),
\end{align*}
}
}
and
{\small { 
\begin{align*}
 &h(x,v_i;n=2)\\
 =& \mathds{1}\{x\leq v_i\} \kappa(v_i;n = 2,\gamma) f_1(x) \left(\frac{ (1-\gamma)^2}{3}  + \gamma (1-\gamma)F_1\left(v_i\right) - \frac{\gamma(1-\gamma)}{2} F_1^2\left(v_i\right) + \gamma F_1\left(x\right) - \frac{\gamma(1-\gamma)}{2} F_1^2\left(x\right)\right)\\
 & + \mathds{1}\{x > v_i\} \kappa(v_i;n = 2,\gamma) f_1(x)\left(\frac{(1-\gamma)^2}{3} + \gamma(1-\gamma) F_1\left(x\right) - \frac{\gamma(1-\gamma)}{2} F_1^2\left(x\right) + \gamma F_1\left(v_i\right) - \frac{\gamma (1-\gamma)}{2} F_1^2\left(v_i\right) \right). 
\end{align*}}
}
Furthermore,
\begin{enumerate}[(a)]
     \item For any $v_i^\prime\leq \tilde{v}_i\leq v_i$, we have
\begin{align*}  \frac{h(\tilde{v}_i\mid v_i^\prime;n=2,\gamma)}{H(\tilde{v}_i\mid v_i^\prime;n=2,\gamma)} 
\leq \frac{h(\tilde{v}_i\mid \tilde{v}_i;n=2,\gamma)}{H(\tilde{v}_i\mid \tilde{v}_i;n=2,\gamma)} 
\geq 
 \frac{h(\tilde{v}_i\mid v_i;n=2,\gamma)}{H(\tilde{v}_i\mid v_i;n=2,\gamma)}.
\end{align*}
\item $\frac{h(x\mid v_i;n=2,\gamma)}{H(x\mid v_i;n=2,\gamma)}v_i$ is non-decreasing with $v_i \in [x,1]$ for $F_1(x) = x^c$ with $0< c\leq 1$.
\end{enumerate}

\end{corollary}

\begin{lemma}
\label{lem:marginals_joint_g} (i). The marginal density and CDF of the joint distribution $g(\cdot;n,\gamma)$ in equation \eqref{eq:joint_dist_g} are 
\begin{align*}
    g^{\mar}(x ;n,\gamma)
 = (1-\gamma)f_1(x) + (m-n)\C^m_n\gamma \cdot \sum_{i=0}^{n-1}\frac{(-1)^i\C^{n-1}_i}{m-n+i} F_1^{m-n+i}(x)f_1(x),
\end{align*}
and 
  \begin{align*}
   G^{\mar}(x; n,\gamma)  = (1-\gamma)F_1(x) + (m-n)\C^m_n\gamma \cdot \sum_{i=0}^{n-1}\frac{(-1)^i\C^{n-1}_i}{(m-n+i)(m-n+i+1)} F_1^{m-n+i+1}(x).
 \end{align*}
(ii). The density and CDF of $\max_{i\in\I}V_i$
are  
\begin{align*}
    g^{\lar}(x;n,\gamma) =& n(1-\gamma)^n F_1^{n-1}(x)f_1(x)\\
    &
   +(m-n)\C^m_nF_1^m(x)f_1(x)\sum_{i=0}^n\C^n_i\frac{(-1)^i}{m-n+i}\left[i (1-\gamma)[\gamma + (1-\gamma)F_1(x)]^{i-1}-i(1-\gamma)^nF_1^{i-1}(x) \right],
\end{align*}
and 
 \begin{align*}
  & G^{\lar}(x; n,\gamma) =(1-\gamma)^nF_1^n(x) + (m-n)\C^m_nF_1^m(x)\sum_{i=0}^n\C^n_i\frac{(-1)^i}{m-n+i}\left[[\gamma + (1-\gamma)F_1(x)]^i -(1-\gamma)^nF_1^i(x)\right].
 \end{align*}
\end{lemma}

\subsection{An Identity}
\label{app_subec:identity}
The following identity plays a key role in the proofs of \Cref{app_subsec:beliefs}.
\begin{lemma}
\label{lem:useful_identity} Let $v_i = v_{(k)}$ be the $k^{\mathsf{th}}$ smallest element among all elements in the vector $v$. For $j\in \I_{-i}$, $x\in [v_i, 1]$, and positive integer $s$, 
{\small
\begin{align*}
\int_{[0,v_i]^{k-1}\times[v_i,x]^{n-k}}F_1^s\left(v_{(j)}\right) \prod_{\ell \in \I_{-i}}f_1(v_{\ell})dv_{-i}
    =&\mathds{1}\{j\leq k-1\}\frac{\C^{k-1}_{k-j}}{\C^{s+k-1}_{k-j}} F_1^{s+k-1}(v_i)\left(F_1(x)-F_1(v_i)\right)^{n-k} \\
    & + \mathds{1}\{k+1\leq j \leq n\} \sum_{r=0}^s\frac{\C_{n-j+1}^{n-k}\C_r^s}{\C_{n-j+1}^{n+r-k}}F_1^{s+k-r-1}(v_i)(F_1(x)-F_1(v_i))^{n+r-k}.
\end{align*}}
Specifically, when $k=n$, 
\begin{align*}
\int_{ [0,v_i]^{n-1}}F_1^s\left(v_{(j)}\right) \prod_{\ell \in \I_{-i}}f_1(v_{\ell})dv_{-i} = \mathds{1}\{j\leq n-1\}\frac{\C^{n-1}_{n-j}}{\C^{s+n-1}_{n-j}} F_1^{s+n-1}(v_i).
\end{align*}
\end{lemma}

\subsection{Proofs for \Cref{app_subsec:beliefs}}

\begin{proof}[\textbf{Proof of \Cref{lem:closedforms_admission_prob}}] Define $v_{(0)} = 0$ and $v_{(n+1)} = 1$. 
When $x \in \left[F_2^{-1}\left(F_1\left(v_{(k)}\right)\right), F_2^{-1}\left(F_1\left(v_{(k+1)}\right)\right)\right)$,  we have
\begin{align*}
 \prod_{i\in \I}\left[1-F_{2\mid 1}(x\mid v_i)\right]=&\prod_{i\in \I}\left[\gamma\cdot \mathds{1}\{F_2^{-1}\left(F_1(v_i)\right)> x\} + (1-\gamma)\cdot (1-F_2(x))\right]\\
 =& \left[(1-\gamma)\cdot (1-F_2(x))\right]^k
 \cdot 
 \left[\gamma + (1-\gamma)\cdot (1 -F_2(x))\right]^{n-k}.   
\end{align*}
In addition, define
\begin{align*}
    Q_j\left(x;n,m\right) 
:= \int_{0}^{F_1(x)} t^{m-n-1}(1-t)^{n-j}
 dt = \sum_{i=0}^{n-j}(-1)^{i}\C^{n-j}_i\int_{0}^{F_1(x)}t^{m-n+i-1}dt = \sum_{i=0}^{n-j}\frac{(-1)^{i}\C^{n-j}_i}{m-n+i}F^{m-n+i}_1(x).
\end{align*}
By \eqref{eq:admission_prob}, we have
\begin{align*}
\psi(v;n,\gamma)&=  \Pr\left\{
\max_{j\in \barI} \zeta_j \leq \min_{i\in \I}\zeta_i
\mid v\right\} =
\int_0^1\prod_{i\in \I}[1-F_{2\mid 1}(x\mid v_i)] dF_2^{m-n}(x)\\
 &= \sum_{k=0}^n \int_{F_2^{-1}\left(F_1\left(v_{(k)}\right)\right)}^{F_2^{-1}\left(F_1\left(v_{(k+1)}\right)\right)}
 \left[(1-\gamma)\cdot (1-F_2(x))\right]^k
 \cdot 
 \left[\gamma + (1-\gamma)\cdot (1 -F_2(x))\right]^{n-k}dF_2^{m-n}(x)\\
 & = \sum_{k=0}^n \int_{F_1\left(v_{(k)}\right)}^{F_1\left(v_{(k+1)}\right)}
 (1-\gamma)^k (1-x)^k
 \cdot 
 \left[\gamma + (1-\gamma)\cdot (1 -x)\right]^{n-k}dx^{m-n}\\
 & \overset{(a)}{=} \sum_{k=0}^n \sum_{j=0}^{n-k}\int_{F_1\left(v_{(k)}\right)}^{F_1\left(v_{(k+1)}\right)}\C^{n-k}_j\gamma^j
 (1-\gamma)^{n-j} (1-x)^{n-j}
 dx^{m-n}\\
 & = (m-n)\sum_{j=0}^n \gamma^j
 (1-\gamma)^{n-j} \sum_{k=0}^{n-j}\C^{n-k}_j\int_{F_1\left(v_{(k)}\right)}^{F_1\left(v_{(k+1)}\right)} x^{m-n-1}(1-x)^{n-j}
 dx\\
 & \overset{(b)}{=} (m-n)\sum_{j=0}^n \gamma^j
 (1-\gamma)^{n-j} \sum_{k=0}^{n-j}\C^{n-k}_j\left(Q_j\left(v_{(k+1)};n,m\right) - Q_j\left(v_{(k)};n,m\right)\right)\\
 & = (m-n)\sum_{j=0}^n \gamma^j
 (1-\gamma)^{n-j}\left( \sum_{k=1}^{n-j+1}\C^{n-k+1}_jQ_j\left(v_{(k)};n,m\right) - \sum_{k=0}^{n-j}\C^{n-k}_jQ_j\left(v_{(k)};n,m\right)\right)\\
 & = (m-n)\sum_{j=1}^n \gamma^j
 (1-\gamma)^{n-j} \sum_{k=1}^{n-j}\C^{n-k}_{j-1}Q_j\left(v_{(k)};n,m\right) + (m-n)\sum_{j=0}^n \gamma^j
 (1-\gamma)^{n-j} Q_j\left(v_{(n-j+1)};n,m\right)\\
 & = (m-n)\sum_{k=1}^{n-1}  \sum_{j=1}^{n-k}\C^{n-k}_{j-1}\gamma^j
 (1-\gamma)^{n-j}Q_j\left(v_{(k)};n,m\right) + (m-n)\sum_{j=0}^n \gamma^j
 (1-\gamma)^{n-j}Q_j\left(v_{(n-j+1)};n,m\right)\\
 & = (m-n)\sum_{k=1}^{n-1}  \sum_{j=1}^{n-k}\C^{n-k}_{j-1}\gamma^j
 (1-\gamma)^{n-j}\sum_{i=0}^{n-j}\frac{(-1)^{i}\C^{n-j}_i}{m-n+i}F_1^{m-n+i}\left(v_{(k)}\right)\\
 & \quad + (m-n)\sum_{j=0}^n \gamma^j
 (1-\gamma)^{n-j}\sum_{i=0}^{n-j}\frac{(-1)^{i}\C^{n-j}_i}{m-n+i}F_1^{m-n+i}\left(v_{(n-j+1)}\right)\\
 & = (m-n)\sum_{k=1}^{n-1} F_1^{m-n}\left(v_{(k)}\right) \sum_{j=1}^{n-k}\C^{n-k}_{j-1}\gamma^j
 (1-\gamma)^{n-j}\sum_{i=0}^{n-j}\frac{(-1)^{i}\C^{n-j}_i}{m-n+i}F_1^{i}\left(v_{(k)}\right)\\
 & \quad + (m-n)\sum_{j=0}^{n-1} \gamma^{n-j}
 (1-\gamma)^{j}F_1^{m-n}\left(v_{(j+1)}\right)\sum_{i=0}^{j}\frac{(-1)^{i}\C^{j}_i}{m-n+i}F_1^{i}\left(v_{(j+1)}\right) + \frac{(1-\gamma)^n}{\C^{m}_n}\\
  & = (m-n)\sum_{k=1}^{n-1} F_1^{m-n} \left(v_{(k)}\right)\sum_{j=1}^{n-k}\C^{n-k}_{j-1}\gamma^j
 (1-\gamma)^{n-j}\sum_{i=0}^{n-j}\frac{(-1)^{i}\C^{n-j}_i}{m-n+i}F_1^{i}\left(v_{(k)}\right)\\
 & \quad + (m-n)\sum_{k=1}^{n} \gamma^{n-k+1}
 (1-\gamma)^{k-1}F_1^{m-n}\left(v_{(k)}\right)\sum_{j=0}^{k-1}\frac{(-1)^{j}\C^{k-1}_j}{m-n+j}F_1^{j}\left(v_{(k)}\right) + \frac{(1-\gamma)^n}{\C^{m}_n}\\
 &= \hat{\psi}_0(n,\gamma)+\sum_{k=1}^n \hat{\psi}_k(v_{(k)};n,\gamma).
  \end{align*}
The equality $(a)$ holds by using the binomial expansion, and the equality $(b)$ holds by the definition of $Q_j(x;n,m)$.

We now proceed to derive the normalizing term $\kappa(v_i;n,\gamma)$. By definition, we have
\begin{align*}
 \frac{1}{\kappa(v_i;n,\gamma)} =& \int_{ [0,1]^{n-1}} \psi(v;n,\gamma) \prod_{j\in \I_{-i}} f_1(v_j) dv_{-i} \\
 =& \int_{ [0,1]^{n-1}}  \left(\int_0^1\prod_{k\in \I}\left[1-F_{2\mid 1}(x\mid v_k)\right] dF_2^{m-n}(x)\right)\prod_{j\in \I_{-i}} f_1(v_j) dv_{-i}\\
 =& \int_0^1 [1-F_{2\mid 1}(x\mid v_i)]\prod_{j\in \I_{-i}}\left(\int_0^1 [1-F_{2\mid 1}(x\mid v_j)]dF_1(v_j)\right)dF_2^{m-n}(x),
\end{align*}
where the last equality holds by the Fubini theorem. Observe that 
\begin{align*}
\int_0^1 [1-F_{2\mid 1}(x;v_j)]dF_1(v_j) = &\int_0^1 \gamma \cdot \mathds{1}\{F_2^{-1}(F_1(v_j))>x\} + (1-\gamma)\cdot (1-F_2(x))dF_1(v_j) \\
=&(1-\gamma)(1-F_2(x)) F_2(x) +\left(\gamma + (1-\gamma)(1-F_2(x))\right)(1-F_2(x))\\
=&1-F_2(x).
\end{align*}
Hence, we have
\begin{align}
\frac{1}{\kappa(v_i;n,\gamma)} =&   \int_0^1 \left(\gamma \cdot \mathds{1}\{F_2^{-1}(F_1(v_j))>x\} + (1-\gamma)\cdot (1-F_2(x))\right)(1-F_2(x))^{n-1}dF_2^{m-n}(x) \nonumber\\
=& (1-\gamma)\cdot \int_0^1(1-F_2(x))^ndF_2^{m-n}(x) + \gamma \cdot \int_0^{v_i} (1-F_1(x))^{n-1}dF_1^{m-n}(x)\nonumber\\
=& \frac{1-\gamma}{\C^m_n}   + \gamma \cdot \int_0^{F_2^{-1}(F_1(v_i))} (1-F_2(x))^{n-1}dF_2^{m-n}(x) \label{eq:1_over_kappa} \\
=& \frac{1-\gamma}{\C^m_n} + (m-n)\gamma \cdot \sum_{j=0}^{n-1}\frac{(-1)^j\C^{n-1}_j}{m-n+j} F_1^{m-n+j}(v_i). \nonumber
\end{align}
This completes the proof.
\end{proof}

\begin{proof}[\textbf{Proof of \Cref{cor:special_gamma}}]
    This directly follows from \Cref{lem:closedforms_admission_prob}.
\end{proof}

\begin{proof}[\textbf{Proof of \Cref{prop:admissionprob_property}}]
From \Cref{lem:closedforms_admission_prob} regarding the expression of the admission probability, we know that the admission probability depends only on the order of the elements in $v$ and is \textit{independent} of the specific dimension. Thus, the admission probability is a symmetric function. Moreover, in any domain $\mathcal{V}$ where the order of elements in $v\in \mathcal{V}$ remains fixed, the expression for the admission probability is invariant with respect to each dimension $i$. This immediately implies differentiability almost everywhere and \Cref{prop:admissionprob_property} (ii) and (iii).
When $\gamma=1$, by \eqref{eq:gamma_1_admissionprob}, we have $\psi(v;n,\gamma)=F^{m-n}\left(\min_{i\in \I}v_i\right)$. Then, supermodularity and log-supermodularity can be directly verified from this expression. 
\end{proof}

\begin{proof}[\textbf{Proof of \Cref{lem:H_h}}] 
We prove each result separately.

\medskip
\noindent\underline{\textit{Proof of (i).}}
\medskip

When $x \in [0,v_i]$, by the definition \eqref{eq:H_CDF} and \Cref{lem:joint_dis_g} about the jonity density $g(v;n,\gamma)$, we have
\begin{align*}
\frac{H(x\mid v_i;n,\gamma)}{\kappa(v_i;n,\gamma)} &=   
 \int_{[0,x]^{n-1}}
\psi(v;n,\gamma)\prod_{j\in \I_{-i}} f_1(v_j) 
dv_{-i}\\
& \overset{(a)}{=} \int_{[0,x]^{n-1}}
\left(\hat{\psi}_0(n,\gamma)+\sum_{k=1}^n \hat{\psi}_k(v_{(k)};n,\gamma)\right)\prod_{j\in \I_{-i}} f_1(v_j)dv_{-i}\\
& = \sum_{k=1}^{n-1} \int_{[0,x]^{n-1}} \hat{\psi}_k(v_{(k)};n,\gamma) \prod_{j\in \I_{-i}} f_1(v_j)
dv_{-i} + F_1^{n-1}(x)\left(\hat{\psi}_0(n,\gamma)+\hat{\psi}_n(v_{i};n,\gamma)\right)\\
&\overset{(b)}{=}  (m-n)\sum_{k=1}^{n-1}  \sum_{j=1}^{n-k}\C^{n-k}_{j-1}\gamma^j
 (1-\gamma)^{n-j}\sum_{i=0}^{n-j}\frac{(-1)^{i}\C^{n-j}_i}{m-n+i}\int_{[0,x]^{n-1}}F_1^{m-n+i}\left(v_{(k)}\right)\left(\prod_{j\in \I_{-i}} f_1(v_j)\right)dv_{-i}\\
 & \quad + (m-n)\sum_{k=1}^{n-1} \gamma^{n-k+1}
 (1-\gamma)^{k-1}\sum_{j=0}^{k-1}\frac{(-1)^{j}\C^{k-1}_j}{m-n+j}\int_{[0,x]^{n-1}}F_1^{m-n+j}\left(v_{(k)}\right)\left(\prod_{j\in \I_{-i}} f_1(v_j)\right) dv_{-i}\\
 &\quad + F_1^{n-1}(x)\left((m-n)\gamma(1-\gamma)^{n-1}F_1^{m-n}(v_i)\sum_{j=0}^{n-1}\frac{(-1)^j\C^{n-1}_j}{m-n+j}F_1^j(v_i) + \frac{(1-\gamma)^n}{\C^m_n}\right)\\
& \overset{(c)}{=}  (m-n)F_1^{m-1}(x)\sum_{k=1}^{n-1}  \sum_{j=1}^{n-k}\C^{n-k}_{j-1}\gamma^j
 (1-\gamma)^{n-j}\sum_{i=0}^{n-j}\frac{(-1)^{i}\C^{n-j}_i\C^{n-1}_{n-k}}{(m-n+i)\C^{m+i-1}_{n-k}}F_1^i\left(x\right)\\
 & \quad + (m-n)F_1^{m-1}(x)\sum_{k=1}^{n-1} \gamma^{n-k+1}
 (1-\gamma)^{k-1}\sum_{j=0}^{k-1}\frac{(-1)^{j}\C^{k-1}_j\C^{n-1}_{n-k}}{(m-n+j)\C^{m+j-1}_{n-k}}F_1^{j}\left(x\right) \\
 &\quad + F_1^{n-1}(x)\left((m-n)\gamma(1-\gamma)^{n-1}F_1^{m-n}(v_i)\sum_{j=0}^{n-1}\frac{(-1)^j\C^{n-1}_j}{m-n+j}F_1^j(v_i) + \frac{(1-\gamma)^n}{\C^m_n}\right).
\end{align*}
The equalities $(a)$ and $(b)$ hold by \Cref{lem:closedforms_admission_prob}, while the equality $(c)$ uses \Cref{lem:useful_identity}.

When $x\in [v_i,1]$, we have
\begin{align*}
    \frac{H(x\mid v_i;n,\gamma)}{\kappa(v_i;n,\gamma)} &=   
 \int_{[0,x]^{n-1}}
\psi(v;n,\gamma)\prod_{j\in \I_{-i}} f_1(v_j) 
dv_{-i}\\
& = \int_{[0,x]^{n-1}}
\left(\hat{\psi}_0(n,\gamma)+\sum_{k=1}^n \hat{\psi}_k(v_{(k)};n,\gamma)\right)\prod_{j\in \I_{-i}} f_1(v_j)dv_{-i}\\
& = \hat{\psi}_0(n,\gamma)F_1^{n-1}(x) + \sum_{k=1}^{n}\C^{n-1}_{k-1}\int_{[0,v_i]^{k-1}\times [v_i,x]^{n-k}}
\sum_{s=1}^n \hat{\psi}_s(v_{(s)};n,\gamma)\prod_{\ell\in \I_{-i}} f_1(v_{\ell})dv_{-i}\\
& = \hat{\psi}_0(n,\gamma)F_1^{n-1}(x) + \sum_{k=1}^n \C^{n-1}_{k-1}\hat{\psi}_k(v_{i};n,\gamma)F_1^{k-1}(v_i)(F_1(x) - F_1(v_i))^{n-k}\\
& \quad + \sum_{k=1}^{n}\C^{n-1}_{k-1}\int_{[0,v_i]^{k-1}\times [v_i,x]^{n-k}}
\sum_{s\neq k} \hat{\psi}_s(v_{(s)};n,\gamma)\prod_{\ell\in \I_{-i}} f_1(v_{\ell})dv_{-i}.
\end{align*}
By \Cref{lem:useful_identity} and \Cref{lem:closedforms_admission_prob}, we have 
\begin{align*}
&\sum_{s= 1}^{k-1}\int_{[0,v_i]^{k-1}\times [v_i,x]^{n-k}} \hat{\psi}_s(v_{(s)};n,\gamma)\prod_{\ell\in \I_{-i}} f_1(v_{\ell})dv_{-i}  \\
=&(m-n) \sum_{s=1}^{k-1}\sum_{j=1}^{n-s} \C^{n-s}_{j-1}\gamma^j
 (1-\gamma)^{n-j}\sum_{t=0}^{n-j}\frac{(-1)^{t}\C^{n-j}_t}{m-n+t} \int_{[0,v_i]^{k-1}\times [v_i,x]^{n-k}} F_1^{m-n+t}\left(v_{(s)}\right)\prod_{\ell\in \I_{-i}} f_1(v_{\ell})dv_{-i} \\
 &+ (m-n) \sum_{s=1}^{k-1}\gamma^{n-s+1}
 (1-\gamma)^{s-1}\sum_{j=0}^{s-1}\frac{(-1)^{j}\C^{s-1}_j}{m-n+j}\int_{[0,v_i]^{k-1}\times [v_i,x]^{n-k}} F_1^{m-n+j}\left(v_{(s)}\right)\prod_{\ell\in \I_{-i}} f_1(v_{\ell})dv_{-i}\\
 =& (m-n) \sum_{s=1}^{k-1}\sum_{j=1}^{n-s} \C^{n-s}_{j-1}\gamma^j
 (1-\gamma)^{n-j}\sum_{t=0}^{n-j}\frac{(-1)^{t}\C^{n-j}_t\C^{k-1}_{k-s}}{(m-n+t)\C^{m-n+t+k-1}_{k-s}}F_1^{m-n+t+k-1}(v_i)(F_1(x) - F_1(v_i))^{n-k}\\
 & + (m-n) \sum_{s=1}^{k-1}\gamma^{n-s+1}
 (1-\gamma)^{s-1}\sum_{j=0}^{s-1}\frac{(-1)^{j}\C^{s-1}_j\C^{k-1}_{k-s}}{(m-n+j)\C^{m-n+j+k-1}_{k-s}}F_1^{m-n+j+k-1}(v_i)(F_1(x) - F_1(v_i))^{n-k},
\end{align*}
and 
{\small
\begin{align*}
&\sum_{s= k+1}^{n}\int_{[0,v_i]^{k-1}\times [v_i,x]^{n-k}} \hat{\psi}_s(v_{(s)};n,\gamma)\prod_{\ell\in \I_{-i}} f_1(v_{\ell})dv_{-i}  \\
=&(m-n) \sum_{s=k+1}^{n-1}\sum_{j=1}^{n-s} \C^{n-s}_{j-1}\gamma^j
 (1-\gamma)^{n-j}\sum_{t=0}^{n-j}\frac{(-1)^{t}\C^{n-j}_t}{m-n+t} \int_{[0,v_i]^{k-1}\times [v_i,x]^{n-k}} F_1^{m-n+t}\left(v_{(s)}\right)\prod_{\ell\in \I_{-i}} f_1(v_{\ell})dv_{-i} \\
 &+ (m-n) \sum_{s=k+1}^{n}\gamma^{n-s+1}
 (1-\gamma)^{s-1}\sum_{j=0}^{s-1}\frac{(-1)^{j}\C^{s-1}_j}{m-n+j}\int_{[0,v_i]^{k-1}\times [v_i,x]^{n-k}} F_1^{m-n+j}\left(v_{(s)}\right)\prod_{\ell\in \I_{-i}} f_1(v_{\ell})dv_{-i}\\
 =& (m-n) \sum_{s=k+1}^{n-1}\sum_{j=1}^{n-s} \C^{n-s}_{j-1}\gamma^j
 (1-\gamma)^{n-j}\sum_{t=0}^{n-j}\frac{(-1)^{t}\C^{n-j}_t}{m-n+t}\sum_{r=0}^{m-n+t}\frac{\C^{n-k}_{n-s+1}\C^{m-n+t}_r}{\C^{n+r-k}_{n-s+1}}F_1^{m-n+t+k-r-1}(v_i)(F_1(x) - F_1(v_i))^{n+r-k}\\
 & + (m-n) \sum_{s=k+1}^{n}\gamma^{n-s+1}
 (1-\gamma)^{s-1}\sum_{j=0}^{s-1}\frac{(-1)^{j}\C^{s-1}_j}{m-n+j}\sum_{r=0}^{m-n+j}\frac{\C^{n-k}_{n-s+1}\C^{m-n+j}_r}{\C^{n+r-k}_{n-s+1}}F_1^{m-n+j+k-r-1}(v_i)(F_1(x) - F_1(v_i))^{n+r-k}.
\end{align*}}

Combining the above analysis, we have
{\footnotesize
\begin{align*}
&\frac{H(x\mid v_i;n,\gamma)}{\kappa(v_i;n,\gamma)}\\
=& \frac{(1-\gamma)^n}{\C^m_n}F_1^{n-1}(x) + (m-n)\sum_{k=1}^{n-1}\C^{n-1}_{k-1} F_1^{m-n+k-1}(v_i)(F_1(x) - F_1(v_i))^{n-k} \sum_{j=1}^{n-k} \C^{n-k}_{j-1}\gamma^j
 (1-\gamma)^{n-j}\sum_{t=0}^{n-j}\frac{(-1)^{t}\C^{n-j}_t}{m-n+t}F_1^{t}\left(v_{i}\right)\\
 & + (m-n)\sum_{k=1}^{n}\C^{n-1}_{k-1} \gamma^{n-k+1}
 (1-\gamma)^{k-1}F_1^{m-n+k-1}\left(v_{i}\right)(F_1(x) - F_1(v_i))^{n-k}\sum_{j=0}^{k-1}\frac{(-1)^{j}\C^{k-1}_j}{m-n+j}F_1^{j}\left(v_{i}\right)\\
 & + (m-n) \sum_{k=2}^{n}\C^{n-1}_{k-1}\sum_{s=1}^{k-1}\sum_{j=1}^{n-s} \C^{n-s}_{j-1}\gamma^j
 (1-\gamma)^{n-j}\sum_{t=0}^{n-j}\frac{(-1)^{t}\C^{n-j}_t\C^{k-1}_{k-s}}{(m-n+t)\C^{m-n+t+k-1}_{k-s}}F_1^{m-n+t+k-1}(v_i)(F_1(x) - F_1(v_i))^{n-k}\\
 & + (m-n)\sum_{k=2}^{n} \C^{n-1}_{k-1}\sum_{s=1}^{k-1}\gamma^{n-s+1}
 (1-\gamma)^{s-1}\sum_{j=0}^{s-1}\frac{(-1)^{j}\C^{s-1}_j\C^{k-1}_{k-s}}{(m-n+j)\C^{m-n+j+k-1}_{k-s}}F_1^{m-n+j+k-1}(v_i)(F_1(x) - F_1(v_i))^{n-k}\\
 & + (m-n)\sum_{k=1}^{n-2}\C^{n-1}_{k-1} \sum_{s=k+1}^{n-1}\sum_{j=1}^{n-s} \C^{n-s}_{j-1}\gamma^j
 (1-\gamma)^{n-j}\sum_{t=0}^{n-j}\frac{(-1)^{t}\C^{n-j}_t}{m-n+t}\sum_{r=0}^{m-n+t}\frac{\C^{n-k}_{n-s+1}\C^{m-n+t}_r}{\C^{n+r-k}_{n-s+1}}F_1^{m-n+t+k-r-1}(v_i)(F_1(x) - F_1(v_i))^{n+r-k}\\
 & + (m-n)\sum_{k=1}^{n-1} \C^{n-1}_{k-1}\sum_{s=k+1}^{n}\gamma^{n-s+1}
 (1-\gamma)^{s-1}\sum_{j=0}^{s-1}\frac{(-1)^{j}\C^{s-1}_j}{m-n+j}\sum_{r=0}^{m-n+j}\frac{\C^{n-k}_{n-s+1}\C^{m-n+j}_r}{\C^{n+r-k}_{n-s+1}}F_1^{m-n+j+k-r-1}(v_i)(F_1(x) - F_1(v_i))^{n+r-k}.
\end{align*}
}

\medskip
\noindent\underline{\textit{Proof of (ii). }}
\medskip

The density
$h(x\mid v_i;n,\gamma)$ can be obtained directly by taking the derivative with $x$ from $H(x\mid v_i;n,\gamma)$. This completes the proof.
\end{proof}

\begin{proof}[\textbf{{Proof of \Cref{cor:special_case_H_h}}}]

We provide the proof for analyzing $\frac{h(x\mid v_i;n=2,\gamma)}{H(x\mid v_i;n=2,\gamma)}$ and $\frac{h(x\mid v_i;n=2,\gamma)}{H(x\mid v_i;n=2,\gamma)}v_i$ with $m=3$. 
Other parts can be directly obtained by \Cref{lem:H_h}.

\medskip

\noindent\textit{\underline{Step 1: Analysis of $\frac{h(x\mid v_i;n=2,\gamma)}{H(x\mid v_i;n=2,\gamma)}$.}}

\medskip

\noindent (i). When $x\leq v_i$, we have
\begin{align*}
\frac{h(x\mid v_i;n=2,\gamma)}{H(x\mid v_i;n=2,\gamma)} =& \frac{\frac{ (1-\gamma)^2}{3}  + \gamma (1-\gamma)F_1\left(v_i\right) - \frac{\gamma(1-\gamma)}{2} F_1^2\left(v_i\right) + \gamma F_1\left(x\right) - \frac{\gamma(1-\gamma)}{2} F_1^2\left(x\right)}{\frac{ (1-\gamma)^2}{3}  + \gamma (1-\gamma)F_1\left(v_i\right) - \frac{\gamma(1-\gamma)}{2} F_1^2\left(v_i\right) + \frac{\gamma}{2}  F_1\left(x\right)  - \frac{\gamma(1-\gamma)}{6}F_1^2\left(x\right)} \frac{f_1(x)}{F_1(x)} \\
=&\left(1 + \frac{\frac{\gamma}{2} F_1\left(x\right) - \frac{\gamma(1-\gamma)}{3} F_1^2\left(x\right)}{\frac{ (1-\gamma)^2}{3}  + \gamma (1-\gamma)F_1\left(v_i\right) - \frac{\gamma(1-\gamma)}{2} F_1^2\left(v_i\right) + \frac{\gamma}{2}  F_1\left(x\right)  - \frac{\gamma(1-\gamma)}{6}F_1^2\left(x\right)} \right)\frac{f_1(x)}{F_1(x)}.
\end{align*}
Observe that $\gamma (1-\gamma)F_1\left(v_i\right) - \frac{\gamma(1-\gamma)}{2} F_1^2\left(v_i\right)$ is non-decreasing with $v_i$, and $\frac{\gamma}{2} F_1\left(x\right) - \frac{\gamma(1-\gamma)}{3} F_1^2\left(x\right)\geq 0$. Hence, $\frac{h(x\mid v_i;n=2,\gamma)}{H(x\mid v_i;n=2,\gamma)}$ is non-increasing with $v_i$. This implies that 
\begin{align*}
    \frac{h(\tilde{v}_i\mid v_i;n=2,\gamma)}{H(\tilde{v}_i\mid v_i;n=2,\gamma)}  \leq \frac{h(\tilde{v}_i\mid \tilde{v}_i;n=2,\gamma)}{H(\tilde{v}_i\mid \tilde{v}_i;n=2,\gamma)}~, \forall \tilde{v}_i \leq v_i.
\end{align*}

\medskip

\noindent (ii). When $x>v_i$, 
\begin{align*}
&\frac{h(x\mid v_i;n=2,\gamma)}{H(x\mid v_i;n=2,\gamma)}\\
=&\frac{\frac{(1-\gamma)^2}{3} + \gamma(1-\gamma) F_1\left(x\right) - \frac{\gamma(1-\gamma)}{2} F_1^2\left(x\right) + \gamma F_1\left(v_i\right) - \frac{\gamma (1-\gamma)}{2} F_1^2\left(v_i\right)}{\frac{(1-\gamma)^2}{3}F_1\left(x\right) + \frac{\gamma(1-\gamma)}{2} F_1^2\left(x\right) - \frac{\gamma(1-\gamma)}{6} F_1^3\left(x\right) + \gamma F_1\left(v_i\right)F_1\left(x\right) - \frac{\gamma^2}{2} F_1^2\left(v_i\right) - \frac{\gamma (1-\gamma)}{2} F_1^2\left(v_i\right) F_1\left(x\right)} f_1(x)\\
=& \frac{\frac{(1-\gamma)^2}{3} + \gamma(1-\gamma) F_1\left(x\right) - \frac{\gamma(1-\gamma)}{2} F_1^2\left(x\right) + \gamma F_1\left(v_i\right) - \frac{\gamma (1-\gamma)}{2} F_1^2\left(v_i\right)}{\frac{(1-\gamma)^2}{3} + \frac{\gamma(1-\gamma)}{2} F_1\left(x\right) - \frac{\gamma(1-\gamma)}{6} F_1^2\left(x\right) + \gamma F_1\left(v_i\right) - \frac{\gamma^2}{2} \frac{F_1^2\left(v_i\right)}{F_1(x)} - \frac{\gamma (1-\gamma)}{2} F_1^2\left(v_i\right) } \frac{f_1(x)}{F_1(x)} \\
=&\left(1 + r(x\mid v_i;\gamma) \right)\frac{f_1(x)}{F_1(x)},
\end{align*}
where
\begin{align*}
    r(x\mid v_i;\gamma) := \frac{\frac{\gamma^2}{2} \frac{F_1^2(v_i)}{F_1(x)} + \frac{\gamma(1-\gamma)}{2} F_1\left(x\right) - \frac{\gamma(1-\gamma)}{3}F_1^2(x)}{\frac{(1-\gamma)^2}{3} + \frac{\gamma(1-\gamma)}{2} F_1\left(x\right) - \frac{\gamma(1-\gamma)}{6} F_1^2\left(x\right) + \gamma F_1\left(v_i\right) - \frac{\gamma^2}{2} \frac{F_1^2\left(v_i\right)}{F_1(x)} - \frac{\gamma (1-\gamma)}{2} F_1^2\left(v_i\right) }.
\end{align*}
Define 
\begin{align*}
    R(x\mid v_i;\gamma) =& \frac{\gamma^3}{2}F_1^2(v_i) -\frac{\gamma^2(1-\gamma)}{2}F_1^2(x) - \frac{\gamma(1-\gamma)}{3}F_1^3(x) \\
    & + \left(-\frac{\gamma^2(1-\gamma)^2}{3}F_1^3(x) + \frac{\gamma^2(1-\gamma)(1-2\gamma)}{2}F_1^2(x) + \gamma^3(1-\gamma)F_1(x) + \frac{\gamma^2(1-\gamma)^2}{3}\right)F_1(v_i),
\end{align*}
then we have 
\begin{align*}
\frac{\partial r(x\mid v_i;\gamma)}{\partial v_i} = \frac{R(x\mid v_i;\gamma)}{\frac{(1-\gamma)^2}{3} + \frac{\gamma(1-\gamma)}{2} F_1\left(x\right) - \frac{\gamma(1-\gamma)}{6} F_1^2\left(x\right) + \gamma F_1\left(v_i\right) - \frac{\gamma^2}{2} \frac{F_1^2\left(v_i\right)}{F_1(x)} - \frac{\gamma (1-\gamma)}{2} F_1^2\left(v_i\right)} \frac{f_1(v_i)}{F_1(x)}.   
\end{align*}
Since $\frac{\partial R(x\mid v_i;\gamma)}{\partial v_i}\geq 0$ and
\begin{align*}
    R(x\mid v_i=0;\gamma) = -\frac{\gamma^2(1-\gamma)}{2}F_1^2(x) - \frac{\gamma(1-\gamma)}{3}F_1^3(x) \leq 0,
\end{align*}
we conclude that $r(x\mid v_i;\gamma)$ and $\frac{h(x\mid v_i;n=2,\gamma)}{H(x\mid v_i;n=2,\gamma)}$
either weakly decreasing or first weakly decreasing then weakly increasing with $v_i \in[0,x]$. 

Notice that 
\begin{align*}
&\frac{h(x\mid 0;n=2,\gamma)}{H(x\mid 0;n=2,\gamma)} - \frac{h(x\mid x;n=2,\gamma)}{H(x\mid x;n=2,\gamma)}\\
=& \frac{\frac{(1-\gamma)^2}{3} + \gamma(1-\gamma) F_1\left(x\right) - \frac{\gamma(1-\gamma)}{2} F_1^2\left(x\right) }{\frac{(1-\gamma)^2}{3}F_1\left(x\right) + \frac{\gamma(1-\gamma)}{2} F_1^2\left(x\right) - \frac{\gamma(1-\gamma)}{6} F_1^3\left(x\right) } - \frac{\frac{(1-\gamma)^2}{3} + \gamma(2-\gamma) F_1\left(x\right) - \gamma(1-\gamma)F_1^2\left(x\right) }{\frac{(1-\gamma)^2}{3}F_1\left(x\right) + \frac{\gamma(3-2\gamma)}{2} F_1^2\left(x\right) - \frac{2\gamma(1-\gamma)}{3} F_1^3\left(x\right) }\\
=&\frac{\gamma^2(1-\gamma)}{12}\frac{\left(-2(1-\gamma)+3(1-\gamma)F_1(x)-(7-6\gamma)F_1^2(x) +2(1-\gamma)F_1^3(x)\right)F_1^2(x)}{\left(\frac{(1-\gamma)^2}{3}F_1\left(x\right) + \frac{\gamma(1-\gamma)}{2} F_1^2\left(x\right) - \frac{\gamma(1-\gamma)}{6} F_1^3\left(x\right)\right)\left(\frac{(1-\gamma)^2}{3}F_1\left(x\right) + \frac{\gamma(3-2\gamma)}{2} F_1^2\left(x\right) - \frac{2\gamma(1-\gamma)}{3} F_1^3\left(x\right)\right)}\\
\leq &\frac{\gamma^2(1-\gamma)^2}{12}\frac{\left(-2+6F_1(x)-6F_1^2(x) +2F_1^3(x)\right)F_1^2(x)}{\left(\frac{(1-\gamma)^2}{3}F_1\left(x\right) + \frac{\gamma(1-\gamma)}{2} F_1^2\left(x\right) - \frac{\gamma(1-\gamma)}{6} F_1^3\left(x\right)\right)\left(\frac{(1-\gamma)^2}{3}F_1\left(x\right) + \frac{\gamma(3-2\gamma)}{2} F_1^2\left(x\right) - \frac{2\gamma(1-\gamma)}{3} F_1^3\left(x\right)\right)}\\
\leq & 0,
\end{align*}
hence $\frac{h(x\mid v_i;n=2,\gamma)}{H(x\mid v_i;n=2,\gamma)}$ must be first weakly decreasing then weakly increasing with $v_i \in[0,x]$. Therefore, we conclude that $\frac{h(x\mid v_i;n=2,\gamma)}{H(x\mid v_i;n=2,\gamma)} \leq \frac{h(x\mid x;n=2,\gamma)}{H(x\mid x;n=2,\gamma)}$.

\medskip

 Combining (i) and (ii), we know that for any $v_i^\prime\leq \tilde{v}_i\leq v_i$, 
\begin{align*}  \frac{h(\tilde{v}_i\mid v_i^\prime;n=2,\gamma)}{H(\tilde{v}_i\mid v_i^\prime;n=2,\gamma)} 
\leq \frac{h(\tilde{v}_i\mid \tilde{v}_i;n=2,\gamma)}{H(\tilde{v}_i\mid \tilde{v}_i;n=2,\gamma)} 
\geq 
 \frac{h(\tilde{v}_i\mid v_i;n=2,\gamma)}{H(\tilde{v}_i\mid v_i;n=2,\gamma)}.
\end{align*}

\noindent\textit{\underline{Step 2: $\frac{h(x\mid v_i;n=2,\gamma)}{H(x\mid v_i;n=2,\gamma)}v_i$ is non-decreasing with $v_i \in [x,1]$ for $F_1(x) = x^c$, $0< c\leq 1$.}}

\medskip

 Similar to (i) of Step 1, we have 
\begin{align*}
\frac{h(x\mid v_i;n=2,\gamma)}{H(x\mid v_i;n=2,\gamma)}v_i 
=&\left(v_i + \frac{\left(\frac{\gamma}{2} F_1\left(x\right) - \frac{\gamma(1-\gamma)}{3} F_1^2\left(x\right)\right)v_i}{\frac{ (1-\gamma)^2}{3}  + \gamma (1-\gamma)F_1\left(v_{i}\right) - \frac{\gamma(1-\gamma)}{2} F_1^2\left(v_{i}\right) + \frac{\gamma}{2}  F_1\left(x\right)  - \frac{\gamma(1-\gamma)}{6}F_1^2\left(x\right)} \right)\frac{f_1(x)}{F_1(x)}\\
=&\left(v_i + \frac{\frac{\gamma}{2} F_1\left(x\right) - \frac{\gamma(1-\gamma)}{3} F_1^2\left(x\right)}{\frac{ (1-\gamma)^2}{3v_i}  + \gamma (1-\gamma)v_i^{c-1} - \frac{\gamma(1-\gamma)}{2} v_i^{2c-1} + \frac{\gamma}{2v_i}  F_1\left(x\right)  - \frac{\gamma(1-\gamma)}{6v_i}F_1^2\left(x\right)} \right)\frac{f_1(x)}{F_1(x)}.
\end{align*}
Observe that 
$\frac{ (1-\gamma)^2}{3v_i}  + \gamma (1-\gamma)v_i^{c-1} - \frac{\gamma(1-\gamma)}{2} v_i^{2c-1} + \frac{\gamma}{2v_i}  F_1\left(x\right)  - \frac{\gamma(1-\gamma)}{6v_i}F_1^2\left(x\right)$ is non-increasing with $v_i$ when $c\leq 1$, hence $\frac{h(x\mid v_i;n=2,\gamma)}{H(x\mid v_i;n=2,\gamma)}v_i$ is non-decreasing with $v_i \in [x,1]$.
\end{proof}

\begin{proof}[\textbf{{Proof of \Cref{lem:marginals_joint_g}}}]
We prove each case separately.

\medskip
\noindent\underline{\textit{Proof of (i).}}
\medskip

By \Cref{eq:margindist_kappa}, 
\begin{align*}
    g^{\mar}(x ;n,\gamma)
 =
 \frac{\C^m_n}{\kappa(x;n,\gamma)} \cdot f_1(x) = (1-\gamma)f_1(x) + (m-n)C^m_n\gamma \cdot \sum_{i=0}^{n-1}\frac{(-1)^iC^{n-1}_i}{m-n+i} F_1^{m-n+i}(x)f_1(x),
\end{align*}
hence 
\begin{align*}
   G^{\mar}(x ; n,\gamma)  = (1-\gamma)F_1(x) + (m-n)C^m_n\gamma \cdot \sum_{i=0}^{n-1}\frac{(-1)^iC^{n-1}_i}{(m-n+i)(m-n+i+1)} F_1^{m-n+i+1}(x).
 \end{align*}

\noindent\underline{\textit{Proof of (ii).}}

\medskip

By definition, we have 
\begin{align*}
    G^{\lar}(x ; n,\gamma) = &\int_{[0,x]^n} g(v;n,\gamma) dv =  \int_{[0,x]^n}  \C^m_n \psi(v;n,\gamma)
 \prod_{j\in\I} f_1(v_j)dv\\
  =& \C^m_n \int_{ [0,x]^{n}}  \left(\int_0^1\prod_{i\in \I}\left[1-F_{2\mid 1}(t\mid v_i)\right] dF_2^{m-n}(t)\right)\prod_{j\in \I} f_1(v_j) dv\\
 =& \C^m_n\int_0^1 \prod_{j\in \I}\left(\int_0^x [1-F_{2\mid 1}(t\mid v_j)]dF_1(v_j)\right)dF_2^{m-n}(t).
\end{align*}
Observe that 
\begin{align*}
&\int_0^x [1-F_{2\mid 1}(t\mid v_j)]dF_1(v_j) = F_1(x) - \lambda \cdot \min\left\{F_1(x), F_2(t)\right\} - (1-\lambda)\cdot F_1(x)F_2(t).
\end{align*}
Thus, we have
\begin{align*}
&G^{\lar}(x; n,\gamma)\\ = &\C^m_n\int_0^{F_2^{-1}(F_1(x))}  \left[ F_1(x) - F_2(t)(\gamma + (1-\gamma)F_1(x)) \right]^ndF_2^{m-n}(t) + \C^m_n\int_0^1 (1-\gamma)^nF_1^n(x)(1-F_2(t))^n dF_2^{m-n}(t) \\
& - \C^m_n\int_0^{F_2^{-1}(F_1(x))} (1-\gamma)^nF_1^n(x)(1-F_2(t))^n dF_2^{m-n}(t)\\
=& (m-n)\C^m_n\int_0^{F_1(x)}  \sum_{i=0}^n \C^n_i (-1)^i(\gamma + (1-\gamma)F_1(x))^{i}F_1^{n-i}(x) t^{m-n + i-1}dt + \C^m_n\frac{(1-\gamma)^n}{\C^m_n}F_1^n(x)\\
& - (m-n)\C^m_n(1-\gamma)^n F_1^n(x) \int_0^{F_1(x)}  \sum_{i=0}^n \C^n_i (-1)^i t^{m-n+i-1}dt\\
=& (1-\gamma)^nF_1^n(x) + (m-n)\C^m_n F_1^m(x)\sum_{i=0}^n \C^n_i\frac{(-1)^i}{m-n+i}\left((\gamma + (1-\gamma)F_1(x))^i -(1-\gamma)^nF_1^i(x)\right),
\end{align*}
as desired. 

The density $g^{\mar}(x ;n,\gamma)$ can then be directly obtained by taking the derivative of $G^{\lar}(x; n,\gamma)$.
This completes the proof.
\end{proof}

\begin{proof}[\textbf{Proof of \Cref{lem:useful_identity}}]
When $j\leq k-1$, we have
\begin{align*}
    &\int_{ [0,v_i]^{k-1}\times[v_i,x]^{n-k}}F_1^s\left(v_{(j)}\right) \prod_{\ell \in \I_{-i}}f_1(v_{\ell})dv_{-i} \\
    =& \frac{(k-1)!}{(j-1)!(k-1-j)!}\left(F_1(x)-F_1(v_i)\right)^{n-k}\int_0^{v_i}F_1^{s+j-1}\left(v_{(j)}\right)\left(F_1(v_i) - F_1\left(v_{(j)}\right)\right)^{k-1-j}dF_1\left(v_{(j)}\right)\\
    =& j\C_j^{k-1} \left(F_1(x)-F_1(v_i)\right)^{n-k} \int_0^{F_1(v_i)}t^{s+j-1}(F_1(v_i) - t)^{k-1-j}dt\\
    =& j\C_j^{k-1} F_1^{s+k-1}(v_i)\left(F_1(x)-F_1(v_i)\right)^{n-k} \int_0^1t^{s+j-1}(1-t)^{k-1-j}dt \\
    =& \frac{\C^{k-1}_{k-j}}{\C^{s+k-1}_{k-j}} F_1^{s+k-1}(v_i)\left(F_1(x)-F_1(v_i)\right)^{n-k}.
\end{align*}
When $k+1\leq j \leq n$, we have
\begin{align*}
    &\int_{[0,v_i]^{k-1}\times[v_i,x]^{n-k}}F_1^s\left(v_{(j)}\right) \prod_{\ell \in \I_{-i}}f_1(v_{\ell})dv_{-i} \\
    =& \frac{(n-k)!}{(n-j)!(j-k-1)!}F_1^{k-1}(v_i)\int_{v_i}^xF_1^{s}\left(v_{(j)}\right)\left( F_1\left(v_{(j)}\right) - F_1(v_i)\right)^{j-k-1}\left( F_1(x) - F_1\left(v_{(j)}\right)\right)^{n-j}dF_1\left(v_{(j)}\right)\\
    =& (j-k)\C_{n-j}^{n-k} F_1^{k-1}(v_i)\int_{F_1(v_i)}^{F_1(x)}t^{s}\left( t - F_1(v_i)\right)^{j-k-1}\left( F_1(x) - t\right)^{n-j}dt\\
    =& (j-k)\C_{n-j}^{n-k} F_1^{k-1}(v_i)\left(F_1(x)-F_1(v_i)\right)^{n-k} \int_0^1\left((F_1(x)-F_1(v_i))t+F_1(v_i)\right)^st^{j-k-1}(1-t)^{n-j}dt \\
    =& (j-k)\C_{n-j}^{n-k} F_1^{k-1}(v_i)\left(F_1(x)-F_1(v_i)\right)^{n-k} \sum_{r=0}^s \C_{r}^s \int_0^1(F_1(x)-F_1(v_i))^rF_1(v_i)^{s-r}t^{j+r-k-1}(1-t)^{n-j}dt\\
    =&\sum_{r=0}^s\frac{\C_{n-j+1}^{n-k}\C_r^s}{\C_{n-j+1}^{n+r-k}}F_1^{s+k-r-1}(v_i)(F_1(x)-F_1(v_i))^{n+r-k}.
\end{align*}
This completes the proof.
\end{proof}

%% file: fig/general_predictor/FGM/allpay_rev_n_c_1.tikz
\begin{tikzpicture}
  \begin{axis}[
    xlabel={$n$},
    ylabel={Revenue},
    xmin=2, xmax=7,
    ymin=0.3, ymax=0.8,
    xtick={2,3,4,5,6,7},
    legend style={
      at={(0.98,0.02)},   %
      anchor=south east,
      draw=none,
      fill=white,
      /tikz/every even column/.append style={column sep=1em}
    },
    grid=none,
    width=9cm,
    height=7cm
  ]

    \addplot[
      color=black,
      mark=*,
      solid,
      thick
    ] table[
      x index=0,
      y index=5,
      col sep=space
    ] {fig/general_predictor/FGM/data/revenue_allpay_different_n_m_7_c_1.txt};
    \addlegendentry{$\alpha=1.00$}

    \addplot[
      color=blue,
      mark=square*,
      dashed,
      thick
    ] table[
      x index=0,
      y index=4,
      col sep=space
    ] {fig/general_predictor/FGM/data/revenue_allpay_different_n_m_7_c_1.txt};
    \addlegendentry{$\alpha=0.90$}

    \addplot[
      color=red,
      mark=triangle*,
      mark size=4pt,
      dashdotted,
      thick
    ] table[
      x index=0,
      y index=3,
      col sep=space
    ] {fig/general_predictor/FGM/data/revenue_allpay_different_n_m_7_c_1.txt};
    \addlegendentry{$\alpha=0.70$}

    \addplot[
      color=magenta,
      mark=diamond*,
      mark options={fill=magenta},
      solid,
      thick
    ] table[
      x index=0,
      y index=2,
      col sep=space
    ] {fig/general_predictor/FGM/data/revenue_allpay_different_n_m_7_c_1.txt};
    \addlegendentry{$\alpha=0.40$}

    \addplot[
      color=gray,
      mark=otimes*,
      solid,
      thick
    ] table[
      x index=0,
      y index=1,
      col sep=space
    ] {fig/general_predictor/FGM/data/revenue_allpay_different_n_m_7_c_1.txt};
    \addlegendentry{$\alpha=0.00$}

  \end{axis}
\end{tikzpicture}

%% file: fig/general_predictor/FGM/allpay_rev_n_c_5.tikz
\begin{tikzpicture}
  \begin{axis}[
    xlabel={$n$},
    ylabel={Revenue},
    xmin=2, xmax=7,
    ymin=0.75, ymax=0.95,
    xtick={2,3,4,5,6,7},
    legend style={
      at={(0.98,0.02)},   %
      anchor=south east,
      draw=none,
      fill=white,
      /tikz/every even column/.append style={column sep=1em}
    },
    grid=none,
    width=9cm,
    height=7cm
  ]

    \addplot[
      color=black,
      mark=*,
      solid,
      thick
    ] table[
      x index=0,
      y index=5,
      col sep=space
    ] {fig/general_predictor/FGM/data/revenue_allpay_different_n_m_7_c_5.txt};
    \addlegendentry{$\alpha=1.00$}

    \addplot[
      color=blue,
      mark=square*,
      dashed,
      thick
    ] table[
      x index=0,
      y index=4,
      col sep=space
    ] {fig/general_predictor/FGM/data/revenue_allpay_different_n_m_7_c_5.txt};
    \addlegendentry{$\alpha=0.90$}

    \addplot[
      color=red,
      mark=triangle*,
      mark size=4pt,
      dashdotted,
      thick
    ] table[
      x index=0,
      y index=3,
      col sep=space
    ] {fig/general_predictor/FGM/data/revenue_allpay_different_n_m_7_c_5.txt};
    \addlegendentry{$\alpha=0.70$}

    \addplot[
      color=magenta,
      mark=diamond*,
      mark options={fill=magenta},
      solid,
      thick
    ] table[
      x index=0,
      y index=2,
      col sep=space
    ] {fig/general_predictor/FGM/data/revenue_allpay_different_n_m_7_c_5.txt};
    \addlegendentry{$\alpha=0.40$}

    \addplot[
      color=gray,
      mark=otimes*,
      solid,
      thick
    ] table[
      x index=0,
      y index=1,
      col sep=space
    ] {fig/general_predictor/FGM/data/revenue_allpay_different_n_m_7_c_5.txt};
    \addlegendentry{$\alpha=0.00$}

  \end{axis}
\end{tikzpicture}

%% file: fig/general_predictor/FGM/allpay_rev_n_c_10.tikz
\begin{tikzpicture}
  \begin{axis}[
    xlabel={$n$},
    ylabel={Revenue},
    xmin=2, xmax=7,
    ymin=0.85, ymax=0.99,
    xtick={2,3,4,5,6,7},
    legend style={
      at={(0.98,0.02)},   %
      anchor=south east,
      draw=none,
      fill=white,
      /tikz/every even column/.append style={column sep=1em}
    },
    grid=none,
    width=9cm,
    height=7cm
  ]

    \addplot[
      color=black,
      mark=*,
      solid,
      thick
    ] table[
      x index=0,
      y index=5,
      col sep=space
    ] {fig/general_predictor/FGM/data/revenue_allpay_different_n_m_7_c_10.txt};
    \addlegendentry{$\alpha=1.00$}

    \addplot[
      color=blue,
      mark=square*,
      dashed,
      thick
    ] table[
      x index=0,
      y index=4,
      col sep=space
    ] {fig/general_predictor/FGM/data/revenue_allpay_different_n_m_7_c_10.txt};
    \addlegendentry{$\alpha=0.90$}

    \addplot[
      color=red,
      mark=triangle*,
      mark size=4pt,
      dashdotted,
      thick
    ] table[
      x index=0,
      y index=3,
      col sep=space
    ] {fig/general_predictor/FGM/data/revenue_allpay_different_n_m_7_c_10.txt};
    \addlegendentry{$\alpha=0.70$}

    \addplot[
      color=magenta,
      mark=diamond*,
      mark options={fill=magenta},
      solid,
      thick
    ] table[
      x index=0,
      y index=2,
      col sep=space
    ] {fig/general_predictor/FGM/data/revenue_allpay_different_n_m_7_c_10.txt};
    \addlegendentry{$\alpha=0.40$}

    \addplot[
      color=gray,
      mark=otimes*,
      solid,
      thick
    ] table[
      x index=0,
      y index=1,
      col sep=space
    ] {fig/general_predictor/FGM/data/revenue_allpay_different_n_m_7_c_10.txt};
    \addlegendentry{$\alpha=0.00$}

  \end{axis}
\end{tikzpicture}

%% file: fig/general_predictor/FGM/allpay_rev_n_c_17.tikz
\begin{tikzpicture}
  \begin{axis}[
    xlabel={$n$},
    ylabel={Revenue},
    xmin=2, xmax=7,
    ymin=0.91, ymax=1,
    xtick={2,3,4,5,6,7},
    legend style={
      at={(0.98,0.02)},   %
      anchor=south east,
      draw=none,
      fill=white,
      /tikz/every even column/.append style={column sep=1em}
    },
    grid=none,
    width=9cm,
    height=7cm
  ]

    \addplot[
      color=black,
      mark=*,
      solid,
      thick
    ] table[
      x index=0,
      y index=5,
      col sep=space
    ] {fig/general_predictor/FGM/data/revenue_allpay_different_n_m_7_c_17.txt};
    \addlegendentry{$\alpha=1.00$}

    \addplot[
      color=blue,
      mark=square*,
      dashed,
      thick
    ] table[
      x index=0,
      y index=4,
      col sep=space
    ] {fig/general_predictor/FGM/data/revenue_allpay_different_n_m_7_c_17.txt};
    \addlegendentry{$\alpha=0.90$}

    \addplot[
      color=red,
      mark=triangle*,
      mark size=4pt,
      dashdotted,
      thick
    ] table[
      x index=0,
      y index=3,
      col sep=space
    ] {fig/general_predictor/FGM/data/revenue_allpay_different_n_m_7_c_17.txt};
    \addlegendentry{$\alpha=0.70$}

    \addplot[
      color=magenta,
      mark=diamond*,
      mark options={fill=magenta},
      solid,
      thick
    ] table[
      x index=0,
      y index=2,
      col sep=space
    ] {fig/general_predictor/FGM/data/revenue_allpay_different_n_m_7_c_17.txt};
    \addlegendentry{$\alpha=0.40$}

    \addplot[
      color=gray,
      mark=otimes*,
      solid,
      thick
    ] table[
      x index=0,
      y index=1,
      col sep=space
    ] {fig/general_predictor/FGM/data/revenue_allpay_different_n_m_7_c_17.txt};
    \addlegendentry{$\alpha=0.00$}

  \end{axis}
\end{tikzpicture}

%% file: fig/general_predictor/AMH/allpay_rev_n_c_1.tikz
\begin{tikzpicture}
  \begin{axis}[
    xlabel={$n$},
    ylabel={Revenue},
    xmin=2, xmax=7,
    ymin=0.3, ymax=0.8,
    xtick={2,3,4,5,6,7},
    legend style={
      at={(0.98,0.02)},   %
      anchor=south east,
      draw=none,
      fill=white,
      /tikz/every even column/.append style={column sep=1em}
    },
    grid=none,
    width=9cm,
    height=7cm
  ]

    \addplot[
      color=black,
      mark=*,
      solid,
      thick
    ] table[
      x index=0,
      y index=5,
      col sep=space
    ] {fig/general_predictor/AMH/data/revenue_allpay_different_n_m_7_c_1.txt};
    \addlegendentry{$\alpha=1.00$}

    \addplot[
      color=blue,
      mark=square*,
      dashed,
      thick
    ] table[
      x index=0,
      y index=4,
      col sep=space
    ] {fig/general_predictor/AMH/data/revenue_allpay_different_n_m_7_c_1.txt};
    \addlegendentry{$\alpha=0.90$}

    \addplot[
      color=red,
      mark=triangle*,
      mark size=4pt,
      dashdotted,
      thick
    ] table[
      x index=0,
      y index=3,
      col sep=space
    ] {fig/general_predictor/AMH/data/revenue_allpay_different_n_m_7_c_1.txt};
    \addlegendentry{$\alpha=0.70$}

    \addplot[
      color=magenta,
      mark=diamond*,
      mark options={fill=magenta},
      solid,
      thick
    ] table[
      x index=0,
      y index=2,
      col sep=space
    ] {fig/general_predictor/AMH/data/revenue_allpay_different_n_m_7_c_1.txt};
    \addlegendentry{$\alpha=0.40$}

    \addplot[
      color=gray,
      mark=otimes*,
      solid,
      thick
    ] table[
      x index=0,
      y index=1,
      col sep=space
    ] {fig/general_predictor/AMH/data/revenue_allpay_different_n_m_7_c_1.txt};
    \addlegendentry{$\alpha=0.00$}

  \end{axis}
\end{tikzpicture}

%% file: fig/general_predictor/AMH/allpay_rev_n_c_5.tikz
\begin{tikzpicture}
  \begin{axis}[
    xlabel={$n$},
    ylabel={Revenue},
    xmin=2, xmax=7,
    ymin=0.75, ymax=0.95,
    xtick={2,3,4,5,6,7},
    legend style={
      at={(0.98,0.02)},   %
      anchor=south east,
      draw=none,
      fill=white,
      /tikz/every even column/.append style={column sep=1em}
    },
    grid=none,
    width=9cm,
    height=7cm
  ]

    \addplot[
      color=black,
      mark=*,
      solid,
      thick
    ] table[
      x index=0,
      y index=5,
      col sep=space
    ] {fig/general_predictor/AMH/data/revenue_allpay_different_n_m_7_c_5.txt};
    \addlegendentry{$\alpha=1.00$}

    \addplot[
      color=blue,
      mark=square*,
      dashed,
      thick
    ] table[
      x index=0,
      y index=4,
      col sep=space
    ] {fig/general_predictor/AMH/data/revenue_allpay_different_n_m_7_c_5.txt};
    \addlegendentry{$\alpha=0.90$}

    \addplot[
      color=red,
      mark=triangle*,
      mark size=4pt,
      dashdotted,
      thick
    ] table[
      x index=0,
      y index=3,
      col sep=space
    ] {fig/general_predictor/AMH/data/revenue_allpay_different_n_m_7_c_5.txt};
    \addlegendentry{$\alpha=0.70$}

    \addplot[
      color=magenta,
      mark=diamond*,
      mark options={fill=magenta},
      solid,
      thick
    ] table[
      x index=0,
      y index=2,
      col sep=space
    ] {fig/general_predictor/AMH/data/revenue_allpay_different_n_m_7_c_5.txt};
    \addlegendentry{$\alpha=0.40$}

    \addplot[
      color=gray,
      mark=otimes*,
      solid,
      thick
    ] table[
      x index=0,
      y index=1,
      col sep=space
    ] {fig/general_predictor/AMH/data/revenue_allpay_different_n_m_7_c_5.txt};
    \addlegendentry{$\alpha=0.00$}

  \end{axis}
\end{tikzpicture}

%% file: fig/general_predictor/AMH/allpay_rev_n_c_10.tikz
\begin{tikzpicture}
  \begin{axis}[
    xlabel={$n$},
    ylabel={Revenue},
    xmin=2, xmax=7,
    ymin=0.85, ymax=0.99,
    xtick={2,3,4,5,6,7},
    legend style={
      at={(0.98,0.02)},   %
      anchor=south east,
      draw=none,
      fill=white,
      /tikz/every even column/.append style={column sep=1em}
    },
    grid=none,
    width=9cm,
    height=7cm
  ]

    \addplot[
      color=black,
      mark=*,
      solid,
      thick
    ] table[
      x index=0,
      y index=5,
      col sep=space
    ] {fig/general_predictor/AMH/data/revenue_allpay_different_n_m_7_c_10.txt};
    \addlegendentry{$\alpha=1.00$}

    \addplot[
      color=blue,
      mark=square*,
      dashed,
      thick
    ] table[
      x index=0,
      y index=4,
      col sep=space
    ] {fig/general_predictor/AMH/data/revenue_allpay_different_n_m_7_c_10.txt};
    \addlegendentry{$\alpha=0.90$}

    \addplot[
      color=red,
      mark=triangle*,
      mark size=4pt,
      dashdotted,
      thick
    ] table[
      x index=0,
      y index=3,
      col sep=space
    ] {fig/general_predictor/AMH/data/revenue_allpay_different_n_m_7_c_10.txt};
    \addlegendentry{$\alpha=0.70$}

    \addplot[
      color=magenta,
      mark=diamond*,
      mark options={fill=magenta},
      solid,
      thick
    ] table[
      x index=0,
      y index=2,
      col sep=space
    ] {fig/general_predictor/AMH/data/revenue_allpay_different_n_m_7_c_10.txt};
    \addlegendentry{$\alpha=0.40$}

    \addplot[
      color=gray,
      mark=otimes*,
      solid,
      thick
    ] table[
      x index=0,
      y index=1,
      col sep=space
    ] {fig/general_predictor/AMH/data/revenue_allpay_different_n_m_7_c_10.txt};
    \addlegendentry{$\alpha=0.00$}

  \end{axis}
\end{tikzpicture}

%% file: fig/general_predictor/AMH/allpay_rev_n_c_17.tikz
\begin{tikzpicture}
  \begin{axis}[
    xlabel={$n$},
    ylabel={Revenue},
    xmin=2, xmax=7,
    ymin=0.91, ymax=1,
    xtick={2,3,4,5,6,7},
    legend style={
      at={(0.98,0.02)},   %
      anchor=south east,
      draw=none,
      fill=white,
      /tikz/every even column/.append style={column sep=1em}
    },
    grid=none,
    width=9cm,
    height=7cm
  ]

    \addplot[
      color=black,
      mark=*,
      solid,
      thick
    ] table[
      x index=0,
      y index=5,
      col sep=space
    ] {fig/general_predictor/AMH/data/revenue_allpay_different_n_m_7_c_17.txt};
    \addlegendentry{$\alpha=1.00$}

    \addplot[
      color=blue,
      mark=square*,
      dashed,
      thick
    ] table[
      x index=0,
      y index=4,
      col sep=space
    ] {fig/general_predictor/AMH/data/revenue_allpay_different_n_m_7_c_17.txt};
    \addlegendentry{$\alpha=0.90$}

    \addplot[
      color=red,
      mark=triangle*,
      mark size=4pt,
      dashdotted,
      thick
    ] table[
      x index=0,
      y index=3,
      col sep=space
    ] {fig/general_predictor/AMH/data/revenue_allpay_different_n_m_7_c_17.txt};
    \addlegendentry{$\alpha=0.70$}

    \addplot[
      color=magenta,
      mark=diamond*,
      mark options={fill=magenta},
      solid,
      thick
    ] table[
      x index=0,
      y index=2,
      col sep=space
    ] {fig/general_predictor/AMH/data/revenue_allpay_different_n_m_7_c_17.txt};
    \addlegendentry{$\alpha=0.40$}

    \addplot[
      color=gray,
      mark=otimes*,
      solid,
      thick
    ] table[
      x index=0,
      y index=1,
      col sep=space
    ] {fig/general_predictor/AMH/data/revenue_allpay_different_n_m_7_c_17.txt};
    \addlegendentry{$\alpha=0.00$}

  \end{axis}
\end{tikzpicture}